\documentclass[11pt, reqno]{amsart}

\usepackage[utf8]{inputenc}
\usepackage[T1]{fontenc}
\usepackage[english]{babel}
\usepackage{amsmath, amsthm}
\usepackage{ae}
\usepackage{icomma}
\usepackage{units}
\usepackage{color}
\usepackage{graphicx}
\usepackage{bbm}
\usepackage{caption}
\usepackage{array}
\usepackage[hmarginratio=1:1]{geometry}
\usepackage[hyphens]{url}
\usepackage[pdfpagelabels=false, hidelinks]{hyperref}
\usepackage{mathrsfs}
\usepackage{amssymb}
\usepackage{placeins} 
\usepackage{tikz}
\usetikzlibrary{patterns}
\usepackage{float} 
\usepackage[nodayofweek]{datetime}
\usepackage{enumitem}
\usepackage{mathtools}
\usepackage[numbers]{natbib}
\usepackage{dsfont}
\usepackage{ifthen}
\usepackage{comment}
\usepackage{upgreek,textgreek}

\newcommand{\R}{\ensuremath{\mathbb{R}}}
\newcommand{\N}{\ensuremath{\mathbb{N}}}
\newcommand{\Q}{\ensuremath{\mathbb{Q}}}
\newcommand{\C}{\ensuremath{\mathbb{C}}}
\newcommand{\Z}{\ensuremath{\mathbb{Z}}}

\newcommand{\E}{\ensuremath{\mathbb{E}}}
\DeclareMathOperator{\dist}{\textnormal{dist}}

\DeclareMathOperator{\Tr}{Tr}
\let\para\S
\renewcommand{\S}{\ensuremath{\mathbb{S}}}

\let\eps\varepsilon

\renewcommand{\P}{\ensuremath{\mathbb{P}}}
\newcommand{\1}{\ensuremath{\mathds{1}}}

\DeclareMathOperator{\ran}{ran}

\newcommand{\Ai}{\ensuremath{\mathrm{Ai}}}
\newcommand{\Bi}{\ensuremath{\mathrm{Bi}}}

\newtheorem{theorem}{Theorem}[section]

\newtheorem{lemma}[theorem]{Lemma}
\newtheorem{proposition}[theorem]{Proposition}
\newtheorem{corollary}[theorem]{Corollary}

\numberwithin{theorem}{section}
\numberwithin{definition}{section}
\theoremstyle{remark}
\newtheorem{remark}[theorem]{Remark}

\newcommand{\limplus}{{\mathchoice{\vcenter{\hbox{$\scriptstyle +$}}}
  {\vcenter{\hbox{$\scriptstyle +$}}}
  {\vcenter{\hbox{$\scriptscriptstyle +$}}}
  {\vcenter{\hbox{$\scriptscriptstyle +$}}}
}}
\newcommand{\limminus}{{\mathchoice{\vcenter{\hbox{$\scriptstyle -$}}}
  {\vcenter{\hbox{$\scriptstyle -$}}}
  {\vcenter{\hbox{$\scriptscriptstyle -$}}}
  {\vcenter{\hbox{$\scriptscriptstyle -$}}}
}}
\newcommand{\limpm}{{\mathchoice{\vcenter{\hbox{$\scriptstyle \pm$}}}
  {\vcenter{\hbox{$\scriptstyle \pm$}}}
  {\vcenter{\hbox{$\scriptscriptstyle \pm$}}}
  {\vcenter{\hbox{$\scriptscriptstyle \pm$}}}
}}

\newcommand{\refsol}{\ensuremath{\zeta}}

\begin{document}

\title[The Kronig--Penney model in a constant electric field]{On the spectrum of the Kronig--Penney model\\ in a constant electric field}

\author[R. L. Frank]{Rupert L. Frank}
\address{\textnormal{(R. L. Frank)} Mathematisches Institut, Ludwig-Maximilians Universit\"at M\"unchen, Theresienstr. 39, 80333 M\"unchen, Germany, and Munich Center for Quantum Science and Technology (MCQST), Schellingstr. 4, 80799 M\"unchen, Germany, and Department of Mathematics, California Institute of Technology, Pasadena, CA 91125, USA}
\email{r.frank@lmu.de, rlfrank@caltech.edu}

\author[S. Larson]{Simon Larson}
\address{\textnormal{(S. Larson)} Mathematical Sciences, Chalmers University of Technology and the University of Gothenburg, SE-41296 Gothenburg, Sweden}
\email{larsons@chalmers.se}

\subjclass[2010]{}

\thanks{\copyright\, 2022 by the authors. This paper may be reproduced, in its entirety, for non-commercial purposes.\\
U.S.~National Science Foundation grants DMS-1363432 and DMS-1954995, Germany’s Excellence Strategy EXC-2111-390814868 (R.L.F.) and Knut and Alice Wallenberg Foundation grants KAW~2018.0281 and KAW~2021.0193 (S.L.) are acknowledged.}

\begin{abstract} 
We are interested in the nature of the spectrum of the one-dimensional Schr\"odinger operator
\begin{equation*}
  - \frac{d^2}{dx^2}-Fx + \sum_{n \in \mathbb{Z}}g_n \delta(x-n)
  \qquad\text{in}\ L^2(\R)
\end{equation*}
with $F>0$ and two different choices of the coupling constants $\{g_n\}_{n\in \mathbb{Z}}$. In the first model $g_n \equiv \lambda$ and we prove that if $F\in \pi^2 \mathbb{Q}$ then the spectrum is $\R$ and is furthermore absolutely continuous away from an explicit discrete set of points. In the second model the $g_n$ are independent random variables with mean zero and variance $\lambda^2$. Under certain assumptions on the distribution of these random variables we prove that almost surely the spectrum is $\R$ and it is dense pure point if $F < \lambda^2/2$ and purely singular continuous if $F> \lambda^2/2$.
\end{abstract}

\maketitle

\section{Introduction and main results}

In this paper we investigate spectral properties of Schr\"odinger operators defined through the formal differential expression
\begin{equation}\label{eq: Op deterministic}
  -\frac{d^2}{d x^2} - Fx + \lambda\sum_{n\in \Z} \delta(x-n)
\end{equation}
where $F\geq 0$ and $\lambda\in \R$ (a precise definition is given in Section~\ref{sec: Prel}). In particular, we study the absence or presence of absolutely continuous, singular continuous, and pure point spectrum. 

When $F=0$ or $\lambda =0$ the operators defined by~\eqref{eq: Op deterministic} are completely understood. Indeed, in these cases the operator is the Laplace operator ($F=\lambda=0$), the Stark operator ($F\neq0, \lambda=0$), or that of the Kronig--Penney model ($F=0, \lambda \neq 0$). While the spectrum in these three cases differs rather drastically, it is always purely absolutely continuous. When neither of $F$ and $\lambda$ is zero, the structure of the spectrum is less clear. Starting with the work of Berezhkovskii and Ovchinnikov~\cite{BerezhkovskiiOvchinnikov_1976} the model~\eqref{eq: Op deterministic} was frequently studied in the physics literature, as we will review below. Berezhkovskii and Ovchinnikov arrived at the conclusion that localization prevails (in their words, `the width of the electron levels is exactly zero') independently of the parameters $F, \lambda\neq 0$. Later, Ao~\cite{Ao_1990} suggested that the nature of the spectrum (or rather dynamical localization and delocalization properties) depends both on the size of $F/\lambda^2$ and on whether a certain resonance condition on $F$ is satisfied. In \cite[Eq.~(9)]{Ao_1990} the resonance condition is stated explicitly as an integer condition, but other parts of the paper rather suggest a rationality condition. Borysowicz \cite{Borysowicz_PLA_1997} arrives at a similar size condition on $F/\lambda^2$ determining the nature of the spectrum, but not at Ao's resonance condition. Through more quantitative, but still non-rigorous arguments Buslaev~\cite{Buslaev_KronigPenney_99} notices that a special role is played by a rationality condition which, when satisfied, would lead to absolutely continuous spectrum.

While we are still rather far from fully understanding the spectrum in the general case, in this paper we provide evidence for the validity of Buslaev's analysis and almost completely solve the problem under a rationality assumption. In particular, Ao's integer condition \cite[Eq.~(9)]{Ao_1990} needs to be replaced by a rationality condition and Borysowicz's predictions can at most be correct in the irrational case. We would also like to emphasize that, while the rationality assumption is required for our definitive result, the bulk of our analysis is valid without it and we expect this part to play an important role in any future advances in the irrational case. 

Our main result on the spectral structure for this model is summarized in the following theorem. 
\begin{theorem}\label{thm: full-line deterministic}
  Fix $F\in \pi^2\Q_\limplus$, $\lambda \in \R$, and write $F= \frac{\pi^2 q}{3 p}$ with $p, q\in \N$, $\gcd(p, q)=1$. Let $L_{F, \lambda}$ be the self-adjoint realization of~\eqref{eq: Op deterministic} in $L^2(\R)$, then
  \begin{equation*}
    \sigma_{ac}(L_{F,\lambda})=\R\,,\quad \sigma_{sc}(L_{F,\lambda})=\emptyset\,, \quad \mbox{and}\quad  \sigma_{pp}(L_{F,\lambda}) \subseteq \Bigl\{ \frac{\pi^2}{3p}m +\lambda : m \in \Z\Bigr\}\,. 
  \end{equation*}
\end{theorem}

\begin{remark}\label{remarkmainthm}
A few remarks are in order:
\begin{enumerate}
  \item by a unit translation $L_{F, \lambda}$ is unitarily equivalent to $L_{F,\lambda}+F$ and, in particular, the spectrum of $L_{F, \lambda}$ is $F$ periodic. Therefore, the existence of an eigenvalue at $E=\frac{\pi^2}{3p}m+\lambda$ only depends on $m \mod q$. Consequently, to fully determine the spectrum it remains to understand $q$ additional values of $E$. 

  \item the assumption $F\geq 0$ is purely for convenience. Indeed, by the change of variables $x\mapsto -x$ the operator $L_{F, \lambda}$ is unitarily equivalent to $L_{-F,\lambda}$.

  \item While our proof does not exclude the presence of point spectrum in general, we can in particular cases prove that the inclusion of the point spectrum in the theorem is strict. A condition on $p, q, m$ under which $\frac{\pi^2}{3p}m+\lambda$ is not an eigenvalue is given in Theorem~\ref{thm: convergence R deterministic}. 
\end{enumerate}
\end{remark}

We will discuss the strategy of the proof of Theorem \ref{thm: full-line deterministic} later in this introduction. There we will also comment on its relation to the works of Buslaev \cite{Buslaev_KronigPenney_99}, whose predictions we prove rigorously, and Perelman \cite{perelman_AsympAnal2005}, whose techniques we adapt to do so.

\subsection{A related random model}
In addition to the operators $L_{F, \lambda}$ we shall also consider the following closely related random model, which was suggested in \cite{Soukoulis_etal_PhysRevLet_83}.
Our results refine those obtained by Delyon, Simon, and Souillard in~\cite{DelyonSimonSouillard_PRL84,delyonSimonSouillard_AHP85}. 

Let $(\Omega, \mathcal{F}, \P)$ be a probability space, $g \colon \Omega \to \R^\Z$, $\omega \mapsto \{g_n(\omega)\}_{n\in \Z}$ a measurable function, and consider the random differential expression
\begin{equation}\label{eq: Op random}
  -\frac{d^2}{d x^2} - Fx + \sum_{n\in \Z}g_n(\omega) \delta(x-n)\,.
\end{equation}
Throughout we assume that the random sequence of coupling constants $\{g_n(\omega)\}_{n\in \Z}$ satisfies 
  \begin{enumerate}[label=(\roman*)]
    \item\label{ass1: gn indep} $g_n$ are independent for different $n\in \Z$,
    \item\label{ass2: gn mean} $\E_\omega[g_n] =0$,
    \item\label{ass3: gn var} $\E_\omega[g_n^2]=\lambda^2$, for some $\lambda \neq 0$ and all $n$,
    \item\label{ass4: gn 4moment} $\sum_{n \geq 1}\E_\omega[g_n^4]n^{-2}<\infty$,
    \item\label{ass5: gn limit} $\P_\omega\bigl[\lim_{n\to \infty} n^{-1/4}g_n=0\bigr]=\P_\omega\bigl[\lim_{n\to -\infty} |n|^{-1/2}g_n=0\bigr]=1$, and
    \item\label{ass6: gn ac dist} there exists $n_0 \in \Z$ such that the distribution of $g_{n_0}$ is absolutely continuous with respect to Lebesgue measure.
  \end{enumerate}
  
  We note that if $g_n$ are i.i.d.\ copies of a random variable $X$, the assumptions~\ref{ass4: gn 4moment}-\ref{ass5: gn limit} are satisfied if and only if $\E_\omega[X^4]<\infty$; for details, see Remark~\ref{rem: assumptions random model}. In particular, the assumptions are valid with $g_n$ chosen as independent $\mathcal{N}(0, \lambda^2)$. If the $g_n$ are not identically distributed, assumptions~\ref{ass4: gn 4moment}-\ref{ass5: gn limit} are valid if $\E_\omega[|g_n|^\alpha]$ are uniformly bounded for some $\alpha>4$ (see Remark~\ref{rem: assumptions random model}).

\begin{theorem}\label{thm: full-line random}
  Let $\{g_n(\omega)\}_{n\in \Z}$ satisfy assumptions~\ref{ass1: gn indep}-\ref{ass6: gn ac dist}. Then~\eqref{eq: Op random} almost surely defines a self-adjoint operator $L^\omega_{F,\lambda}$ in $L^2(\R)$. Moreover, $\sigma_{ess}(L^\omega_{F,\lambda})=\R$ almost surely and the spectrum is almost surely purely
  \begin{itemize}
    \item singular continuous if $F> \lambda^2/2$,
    \item dense pure point if $F<\lambda^2/2$. 
  \end{itemize}
\end{theorem}

As mentioned above, this theorem refines a result obtained in~\cite{DelyonSimonSouillard_PRL84,delyonSimonSouillard_AHP85}. What the authors of these papers prove is that if $F$ is sufficiently small (depending on $\lambda$) the spectrum is almost surely pure point, while if $F$ is sufficiently large (depending on $\lambda$) the spectrum is almost surely purely continuous. Our result refines this in two aspects; first, we obtain a precise transition point $F=\lambda^2/2$, which separates the two regimes, and secondly we classify the continuous spectrum above the transition point as being singular continuous. That the latter aspect was left open is explicitly mentioned in \cite{DelyonSimonSouillard_PRL84}. Analogous sharp spectral transitions have been shown to occur by Minami~\cite{minami_random_1992} in a model related to that studied here and by Kiselev, Last, and Simon~\cite{kiselev_EFGP_1998} for certain Schr\"odinger operators in $l^2(\Z)$ with decaying random potential. In particular, the techniques in the latter paper will be important for us.

We wish to point out that the transition from singular continuous to pure point spectrum in Theorem~\ref{thm: full-line random} might appear somewhat more drastic than it actually is. As we shall see, the spectral transition can be traced to the asymptotic behavior of (possibly) generalized eigenfunctions. In what follows we prove that generalized eigenfunctions have power-like decay given by $x^{-1/4-\lambda^2/(8F)}$ as $x\to \infty$ (up to oscillations and small corrections). The nature of the spectrum is classified depending on whether these generalized eigenfunctions are square integrable or not. So while the spectral nature changes instantaneously when $F/\lambda^2$ crosses the point $1/2$, the change at the level of behavior of eigenfunctions is continuous as a function of $F/\lambda^2$.

\subsection{Strategy of proof}
While our main interest is towards the spectral theoretic results in Theorems~\ref{thm: full-line deterministic} and~\ref{thm: full-line random}, the bulk of the paper will concern understanding the asymptotic behavior of solutions of the generalized eigenvalue equation
\begin{equation}\label{eq: Generalized eigeneq intro}
  \begin{cases}
    -\psi''(x) -Fx \psi(x) = E\psi(x) &\mbox{for } x\in \R\setminus\Z\,, \\
  \lim_{\eps\to 0^\limplus}(\psi(n+\eps)-\psi(n-\eps))=0 &\mbox{for }n \in \Z\,, \\
  \lim_{\eps\to 0^\limplus}(\psi'(n+\eps)-\psi'(n-\eps))= g_n \psi(n) &\mbox{for }n \in \Z\,,
  \end{cases}
\end{equation}
where $g_n = \lambda$ and $g_n=g_n(\omega)$ in the first and second models, respectively. Passing from precise asymptotics of the solutions $\psi$ to spectral properties of $L_{F,\lambda}$ and $L_{F, \lambda}^\omega$ can be accomplished through arguments which are fairly standard in the field. Indeed, in the case of Theorem~\ref{thm: full-line deterministic} it consists of an application of Gilbert--Pearson subordinacy theory~\cite{gilbert_subordinacy_1989,GilbertPearson_subordinacy_1987} while Theorem~\ref{thm: full-line random} follows from rank-one perturbation theory and, in particular, a spectral averaging argument (see~\cite{simon_TraceIdeals_2005} and references therein). Moreover, the refinements of Gilbert--Pearson theory developed by Jitomirskaya and Last~\cite{JitomirskayaLast_Acta_99,JitomirskayaLast_CMP_00} and furthered by Damanik, Killip, and Lenz~\cite{Damanik_etal_CMP_00} imply that for $F>\frac{\lambda^2}{2}$ the spectral measure of $L_{F, \lambda}^\omega$ almost surely vanishes on sets of Hausdorff dimension less than $1-\frac{\lambda^2}{2F}$; see Proposition \ref{prop: Hausdorff dimension}.

The key in both arguments will be to study whether there exist solutions of~\eqref{eq: Generalized eigeneq intro} which are subordinate at $\infty$ and whether these solutions are square integrable or not. We recall that a non-trivial solution $\psi$ of~\eqref{eq: Generalized eigeneq intro} is subordinate at $\infty$ if for any linearly independent solution $\varphi$ it holds that
\begin{equation*}
  \lim_{x\to \infty} \frac{\int_0^x |\psi(t)|^2\,dt}{\int_0^x |\varphi(t)|^2\,dt} = 0\,.
\end{equation*}

What we shall eventually prove for the ODE in the first and second models differs significantly; in the deterministic case subordinate solutions do not exist in our parameter regime while the converse is almost surely true in the random setting. Nonetheless, the overall strategy we employ to prove these ODE results is the same. However, the analysis required to understand the deterministic model is considerably more intricate. In both cases the key part of the proof is to understand the nature of certain exponential sums with amplitude given in terms of the coupling constants. In the random case, the assumption that $\E_\omega[g_n]=0$ leads to large cancellations in these sums with high probability. In the deterministic model, no such probabilistic cancellations take place and it becomes necessary to understand to a precise degree how cancellations arise from the oscillatory nature of the summand.

\subsection{Reduction to a discrete system}
We wish to emphasize that it is only in the very last step of our proof of Theorem~\ref{thm: full-line deterministic} that the rationality assumption $F\in \pi^2\Q_\limplus$ enters. In fact, what we believe to be our main contribution to the analysis of the operators $L_{F, \lambda}$ is not that stated in the theorem, but rather a reduction of the understanding of their spectral properties to the understanding of a certain discrete system. The discrete system we arrive at resembles those studied in the theory of orthogonal polynomials on the unit circle (OPUC), for which we refer the reader to~\cite{simon_OPUC1_2005,simon_OPUC2_2005}. (As an aside we mention that the $F$ periodicity of the spectrum of $L_{F,\lambda}$ mentioned in Remark \ref{remarkmainthm} is reflected in the fact that the spectrum of CMV matrices underlying OPUC is a subset of the unit circle.)

The idea of a reduction to a discrete system goes back to Buslaev~\cite{Buslaev_KronigPenney_99} and the discrete system that we arrive at is essentially equivalent to his. Buslaev, however, is very clear that his arguments are not rigorous (see, in particular, the last paragraph in Section~1 in~\cite{Buslaev_KronigPenney_99}). Our approach to rigorously providing this reduction follows closely that of Perelman in~\cite{perelman_AsympAnal2005} who performed a similar analysis for less singular periodic perturbations of the Stark operator; see also~\cite{Pozharskii_AlgAnal_02}. While our overall strategy is essentially the same as Perelman's, there are significant technical differences and, in particular, subtle exponential sum estimates play a much more pronounced role in our paper. In this regard we also refer to the discussion at the beginning of Section \ref{sec: Refined exponential bounds}.

While it should be emphasized that neither the Stark nor the Kronig--Penney part of the potential is in any proper sense weaker than the other, the essence of our analysis is to relate the solutions of~\eqref{eq: Generalized eigeneq intro} to the corresponding solutions of the Stark equation. Note that this is somewhat different from Ao's tilted band picture, which rather emphasizes the solutions of the Kronig--Penney equation. A distinguished role in this analysis will be played by a particular solution of
\begin{equation}\label{eq: Stark eq intro}
  -\refsol''(x) -Fx\refsol(x) = E \refsol(x)\,, \quad \mbox{for }x\in \R\,,
\end{equation}
namely
\begin{equation}\label{eq: ref solution intro}
  \refsol(x) = \Bigl(\frac{\pi}{F^{1/3}}\Bigr)^{1/2}\bigl(i\Ai(-F^{1/3}(x+E/F)) + \Bi(-F^{1/3}(x+E/F))\bigr)\,,
\end{equation}
where $\Ai, \Bi$ denote the standard Airy functions~\cite[\para 9]{NIST}.

\begin{theorem}\label{thm: l scale equations intro}
Fix $F>0, E \in \R,$ and $\lambda \in \R$. Let $\psi$ be a real-valued solution of~\eqref{eq: Generalized eigeneq intro} then there are functions $\mathcal{R}\colon \N \to \R_\limplus$ and $\Lambda \colon \N \to \R$ such that, for $x \in\bigl(\frac{\pi^2}{F}(l-\frac12)^2, \frac{\pi^2}{F}(l+\frac12)^2\bigr]$,
  \begin{equation}\label{eq: generalized eigenfunction asymptotic form}
    \psi(x) = \mathcal{R}(l)\frac{e^{i(\Lambda(l)-\lambda\sqrt{\lceil x\rceil/F})}\refsol(x)-e^{-i(\Lambda(l)-\lambda\sqrt{\lceil x\rceil/F})}\bar\refsol(x)}{2i}
    + O\bigl(|\refsol(x)|\mathcal{R}(l)l^{-1/2}\bigr)\,,
  \end{equation}
  and
  \begin{equation}\label{eq: generalized eigenfinction asymptotic integral}
     \int_{\frac{\pi^2}{F}(l-\frac{1}{2})^2}^{\frac{\pi^2}{F}(l+\frac{1}{2})^2}|\psi(x)|^2 \,dx = \frac{\pi}{F}\, \mathcal{R}(l)^2 \left(1+O(l^{-1/2})\right).
  \end{equation}
  Moreover, the functions $\mathcal{R}$ and $\Lambda$ satisfy
  \begin{equation}
    \begin{aligned}
        \log\Bigl(\frac{\mathcal{R}(l+1)}{\mathcal{R}(l)}\Bigr) 
        &=
        \frac{\lambda}{\sqrt{2Fl}}\sin(2\Theta(l))
        + \frac{\lambda^2}{4F l}\bigl(1+\cos(4\Theta(l))\bigr)
        + O(l^{-5/4})\,,\\[7pt]
        \Lambda(l+1)-\Lambda(l) &=
        \frac{\lambda}{\sqrt{2F l}}\cos(2\Theta(l))- \frac{\lambda^2}{4Fl}\sin(4\Theta(l)) + \lambda^2\mathcal{S}(l) + O(l^{-5/4})
    \end{aligned}
  \end{equation}
  where
  \begin{equation*}
    \Theta(l) = \Lambda(l)+ \Gamma(l)\,,\qquad \Gamma(l) = - \frac{\pi^3l^3}{3F}+ \frac{\pi l}{F}(E-\lambda) + \frac{5\pi}{8}\,,
  \end{equation*}
  and $\mathcal{S}$ is independent of $\psi, \mathcal{R}, \Lambda$, and satisfies
  \begin{equation*}
    \mathcal{S}(l) = O(l^{-3/4})\,.
  \end{equation*}
  
\end{theorem}

\begin{remark}
  A couple of remarks:
  \begin{enumerate}
    \item As we shall see below, $|\refsol(x)| = |Fx|^{-1/4}(1+o(1))$ as $x \to \infty$, so the error term in~\eqref{eq: generalized eigenfunction asymptotic form} can be written as $O(\mathcal{R}(l)/l)$.

    \item Since $\psi, \refsol,$ and $\bar\refsol$ all solve the ODE in~\eqref{eq: Stark eq intro} between consecutive integers, the same is true for both the main term and the remainder in~\eqref{eq: generalized eigenfunction asymptotic form}.

    \item While the approximate equation for $\mathcal{R}$ does not fully determine $\mathcal{R}$, it is sufficient to determine the asymptotic behavior of $\mathcal{R}$ up to a constant factor. Similarly, the equation for $\Lambda$ suffices to determine its asymptotic behavior up to a bounded additive error. 

    \item In the proof of Theorem~\ref{thm: l scale equations intro} we will provide an explicit expression for $\mathcal{S}$. However, this expression is somewhat complicated and not particularly illuminating, so we refrain from writing it out here. Furthermore, one can probably show that $\mathcal{S}$ is substantially smaller than stated (in fact, $O(l^{-1-\eps})$ for some $\eps>0$), but the rough bound that we state is sufficient for our purposes.
  \end{enumerate}
\end{remark}

\subsection{Historical remarks \& related work}
While a complete mathematical treatment of the spectral properties of the operator~\eqref{eq: Op deterministic} is still missing, both this model and random variants thereof were a topic of some discussion in theoretical solid state physics in the 80's~\cite{BerezhkovskiiOvchinnikov_1976,Soukoulis_etal_PhysRevLet_83,Bentosela_etal_PhysRevB_1985,LenstraHaeringen_1986,Lenstra_etal_PhysScr_1986,Landauer_PhysRevLet_87,LenstraHaeringen_1987,GefenThouless_PhysRevLet_1987,BlatterBrowne_PhysRevB_1988,Ao_1990,Borysowicz_PLA_1997} (see also~\cite{Joye_2011} for a review of mathematical results on a different, but related model). We note that the Hamiltonian \eqref{eq: Op deterministic} arises in two different physical models, namely, in that of a crystal in an electric field and in that of a conducting ring threaded by a magnetic flux which increases linearly in time.

If the periodic $\delta$-array in the definition of the operator is replaced by a smoother periodic potential the mathematical literature is more extensive. In particular, it has been shown that under appropriate assumptions the absolutely continuous spectrum is all of $\R$, see for instance~\cite{Walter_1972,Avron_etal_1977,Bentosela_etal_CMP_1983,Pozharskii_TMP_2000,Pozharskii_AlgAnal_02,perelman_CMP2003,ChristKiselev_AfM_03,Pozharskii_AlgAnal_04}.
In the opposite direction, it has been shown that if the $\delta$-array is replaced by the even more singular potential given by a $\delta'$-array, the absolutely continuous spectrum is empty~\cite{AvronExnerLast_1994,MaioliSacchetti_JoP_95,Exner_JMP_1995,AschDuclosExner_JStatPhys_98,Albevario_etal_SolvableModels}. As such, the case that we consider here is in some sense critical and the spectral properties of~\eqref{eq: Op deterministic} could very well depend in a very subtle manner on the parameters $F$ and $\lambda$. 

A well-studied problem related to our deterministic model is that of Wannier--Stark ladders. Specifically, one is interested in complex resonances of $-\frac{d^2}{dx^2}-Fx + V$ where $V$ is periodic. While it is not entirely clear to us why, it is interesting that certain cubic exponential sums (see~\eqref{eq: w at exceptional energies}) that appear in the proof of Theorem~\ref{thm: full-line deterministic} also appear in work of Fedotov and Klopp~\cite{fedotov_starkwannier_2016} on Wannier--Stark ladders with potential $V(x)=2\cos(2\pi x)$.

\medskip

{\noindent \bf Acknowledgements} The authors would like to express their gratitude to Pavel Exner, who drew the attention of the first author to this problem in 2005 and shared some helpful remarks on a preliminary version of the manuscript. Several stimulating discussions with Helge Kr\"uger, Jonathan Breuer, and Yoram Last are acknowledged.

\section{Preliminaries \& Notation}
\label{sec: Prel}

Throughout the paper we shall frequently use the standard asymptotics notation $\lesssim$,~$\gtrsim$, $\sim$, $o$, and $O$. When we use this notation the dependence of the implicit constants on the parameters of the problem varies from time to time. In particularly, certain implicit constants in our analysis of the random model are allowed to depend on the realization~$\omega$ and may only be finite with probability one.

We write $\R_\limplus, \R_\limminus$ for $(0, \infty)$ and $(-\infty, 0)$, respectively. Similarly we let $\Q_\limpm = \Q\cap \R_\limpm$.

For a function $u \in L^2(\R)$ such that $u, u'$ are locally absolutely continuous in $\R\setminus \Z$ we define the jumps of $u$ and $u'$ at $n $ by
\begin{align*}
  Ju(n) &= \lim_{\eps \to 0^\limplus} (u(n+\eps)-u(n-\eps))\,,\\
   Ju'(n) &= \lim_{\eps \to 0^\limplus} (u'(n+\eps)-u'(n-\eps))\,,
\end{align*}
whenever these limits exist.

To unify our notation we shall for $F\geq 0$ and $g = \{g_n\}_{n\in\Z}\subset \R$ write $L_{F,g}$ for the Schr\"odinger operator in $L^2(\R)$ defined by the differential expression
\begin{equation*}
  - \frac{d^2}{dx^2}-Fx
\end{equation*}
with domain
\begin{align*}
  D(L_{F,g}) &= \{u \in L^2(\R): u, u' \mbox{ locally abs. cont. in }\R\setminus \Z,\ -u''-Fxu \in L^2(\R\setminus \Z)\\
  &\qquad 
  Ju(n) = 0,\ Ju'(n) = g_n u(n)\}\,.
\end{align*}
That is, in the deterministic model $g$ is the constant sequence, $g_n\equiv\lambda$, while in the random case $g=\{g_n(\omega)\}_{n\in \Z}$. Analogously we can define corresponding Schr\"odinger operators in $L^2(I)$ for any $I\subseteq \R$ by additionally prescribing appropriate boundary conditions.

The following theorem ensures that at least under certain assumptions on the coupling constants the Schr\"odinger operator $L_{F,g}$ is well-defined~\cite{ChristStolz_94}. In particular, we emphasize that these assumptions are valid for the operators considered here (almost surely in the setting of Theorem~\ref{thm: full-line random}).
\begin{lemma}
  Let $L_{F,g}$ be the Schr\"odinger operator defined as above. Then provided there exists constants $C_1, C_2$ such that $g_n \geq -C_1 |n| - C_2$ the operator $L_{F,g}$ is self-adjoint and the differential equation~\eqref{eq: Generalized eigeneq intro} is limit point at $\pm \infty$.
\end{lemma}

As noted above, the assumption $F>0$ is only for definiteness. The reason that this assumption is convenient is made clear by the following lemma, which essentially tells us all that we need about the eigenvalue equation corresponding to our operators on the negative half-line.
\begin{lemma}\label{lem: existence L2 to the left}
  Let $F>0$ and $\{g_n\}_{n< 0}$ satisfy $\liminf_{n< 0}(|n|^{-1/2}g_n)>-\sqrt{F}$. For any $E \in \R$ there exists a non-trivial $\psi \in L^2(\R_\limminus)$ which solves
  \begin{equation*}
    \begin{cases}
      -\psi''(x) -Fx \psi(x) = E\psi(x) \quad &\mbox{in }\R_\limminus\!\setminus \Z\,,\\
    J\psi(n) =0 &\mbox{for } n \in \Z \cap \R_\limminus\,,\\
    J\psi'(n)=g_n \psi(n) \quad &\mbox{for } n \in \Z \cap \R_\limminus\,.
    \end{cases}
  \end{equation*}
\end{lemma}

\begin{proof}
Let $L_\theta$ be the corresponding operator in $L^2(\R_-)$ defined by the boundary condition
\begin{equation*}
  \psi(0)\sin(\theta)+ \psi'(0)\cos(\theta)=0\,.
\end{equation*}
Let also $Q_\theta$ denote the corresponding quadratic form. For $\theta \neq 0$ $Q_\theta$ is defined by taking the closure of the quadratic form
\begin{equation*}
  L^2\cap C^\infty(\R_\limminus)\ni\psi \mapsto\|\psi'\|_{L^2(\R_\limminus)}^2+F\|\sqrt{|x|}\psi\|_{L^2(\R_\limminus)}^2 + \sum_{n< 0}g_n |\psi(n)|^2 + \cot(\theta)|\psi(0)|^2\,,
\end{equation*}
while for $\theta=0$ it is the closure of
\begin{equation*}
  C^\infty_0(\R_\limplus)\ni\psi \mapsto \|\psi'\|_{L^2(\R_\limplus)}^2+F\|\sqrt{x}\psi\|_{L^2(\R_\limplus)}^2 + \sum_{n< 0}g_n |\psi(n)|^2\,.
\end{equation*}

By the Sobolev inequality $\psi$ in the domain of $Q_\theta$ is continuous and for any $x\in \R_\limminus, \eps>0,$ and bounded interval $I\subset \R_\limminus$ containing $x$,
\begin{equation}\label{eq: Sobolev}
  |u(x)|^2 \leq \eps \|u'\|^2_{L^2(I)}+(|I|^{-1}+\eps^{-1}) \|u\|^2_{L^2(I)}\,.
\end{equation}

To avoid cumbersome notation, the term $\cot(\theta)|\psi(0)|^2$ should in what follows be interpreted as zero in the case $\theta=0$.
By assumption there exist $\delta\in (0, 1)$ and $N$ so that $g_n |n|^{-1/2} \geq -\sqrt{F}(1-\delta)$ for $n< -N$. Using~\eqref{eq: Sobolev}, with $I = (-n-1, -n]$ and $\eps = (\sqrt{F}n^{1/2})^{-1}$ to bound $|\psi(-n)|^2$, we find for $\psi$ in the quadratic form domain
\begin{align*}
  Q_\theta(\psi) 
  &\geq 
  \|\psi'\|_{L^2(\R_\limminus)}^2+F\|\sqrt{|x|}\psi\|_{L^2(\R_\limminus)}^2 + \sum_{n=1}^{N} g_{-n} |\psi(-n)|^2+ \cot(\theta)|\psi(0)|^2\\
  & \quad - \sum_{n> N}\sqrt{F}(1-\delta) n^{1/2} |\psi(-n)|^2  \\
  &\geq 
  \delta\|\psi'\|_{L^2(\R_\limminus)}^2+ \sum_{n=1}^{N} g_{-n} |\psi(-n)|^2 + \cot(\theta)|\psi(0)|^2\\
   & \quad
   +\int_{\R_\limminus}\!\!|\psi(x)|^2\Bigl[F|x| -\sum_{n>N}\1_{(-n-1, -n]}(x)\sqrt{F}(1-\delta)n^{1/2}(1+\sqrt{F} n^{1/2})\Bigr]dx\,.
\end{align*}
Since $\delta >0$, the `effective potential'
\begin{equation*}
  F|x| -\sum_{n>N}\1_{(-n-1, -n]}(x)\sqrt{F}(1-\delta) n^{1/2}(1+\sqrt{F} n^{1/2})
\end{equation*}
tends to infinity at $-\infty$. Consequently, the domain of $Q_\theta$ is compactly embedded in $L^2(\R_\limminus)$ and the spectrum of $L_\theta$ is discrete. Therefore at all $E$ which are in the spectrum of $L_\theta$ for some $\theta$ we have an $L^2$-solution of the eigenvalue equation. It remains to prove that $\cup_{\theta \in [0, \pi)}\sigma(L_\theta)=\R$.

  Let $\lambda_k(\theta)$ denote the $k$-th smallest eigenvalue of $L_\theta$. By the variational principle, for each $k$ the function $\theta \mapsto \lambda_k(\theta)$ is continuous and strictly decreasing. Thus the image of $[0, \pi)$ under the map $\lambda_k$ is a half-open interval, and these intervals are disjoint for different $k$. To prove that the eigenvalues cover all of $\R$ it remains to prove that there are no gaps between these intervals. 

  We first note that $\lambda_1(\theta) \to -\infty$ as $\theta \to \pi$. Indeed, this follows by taking as a test function in the Rayleigh quotient any smooth and $L^2$-normalized function which vanishes outside $(-1/2,0]$ and is non-zero at $0$. 

  By the continuity of $\theta \mapsto \lambda_k(\theta)$ it remains to show that $\lambda_{k}(0) = \lim_{\eps \to 0^\limplus}\lambda_{k+1}(\pi-\eps)$ for all $k\geq 1$. Assume for contradiction that this is not the case, and let $k_0$ be the smallest $k$ for which the statement fails. We aim to construct an $L^2$-normalized function $\tilde\phi$ in the form domain of $L_0$ that is orthogonal to the first $k_0$ eigenfunctions of $L_0$ and satisfies $Q_0(\tilde \phi) < \lambda_{k_0+1}(0)$. By the variational characterisation of the eigenvalues this is a contradiction. 

	Let $\phi_{k, \theta}$ be an $L^2$-normalized eigenfunction of $L_\theta$ corresponding to $\lambda_k(\theta)$. We abbreviate $\phi_j = \phi_{j,0}$ and $\psi_\eps = \phi_{k_0+1, \pi -\eps}$. Fix $\eta \in C^\infty(\R_\limminus)$ with $0\leq \eta \leq 1$ with $\eta(x)= 0$ for $x\geq -1$ and $\eta(x)= 1$ for $x\leq -2$. Let $\psi_{\eps,\delta}(x) = \eta(x/\delta)\psi_\eps(x)$. By an integration by parts using that $\psi_\eps, \phi_j$ solve the eigenvalue equations we find
  \begin{equation}\label{eq: almost orthogonality}
  \begin{aligned}
    \int_{\R_\limminus} \phi_j(x)\psi_{\eps,\delta}(x)\,dx 
    &= 
    \frac{\lambda_{k_0+1}(\pi-\eps)}{\lambda_j(0)}\int_{\R_\limminus} \phi_j(x)\psi_{\eps,\delta}(x)\,dx\\
    &\quad 
     - \frac{1}{\lambda_j(0)}\int_{-2\delta}^{-\delta} \left( 2\delta^{-1}\psi'_\eps(x)\eta'(x/\delta)+\delta^{-2}\psi_\eps(x)\eta''(x/\delta)\right) \phi_j(x)\,dx\,.
  \end{aligned}
  \end{equation}
  
  For the second integral on the right side Cauchy--Schwarz and $\eta\in C^\infty$ imply
  \begin{equation}\label{eq: Sobolev interp}
  \begin{aligned}
    \delta^{-1}\biggl|\int_{-2\delta}^{-\delta} \psi'_\eps(x)\eta'(x/\delta)\phi_j(x)\,dx\biggr| \lesssim \delta^{-1}\biggl(\int_{-2\delta}^{-\delta} |\psi'_\eps(x)|^2\,dx\biggr)^{1/2}\biggl(\int_{-2\delta}^{-\delta}|\phi_j(x)|^2\,dx\biggr)^{1/2}\,,\\
    \delta^{-2}\biggl|\int_{-2\delta}^{-\delta} \psi_\eps(x)\eta''(x/\delta)\phi_j(x)\,dx\biggr| \lesssim \delta^{-2}\biggl(\int_{-2\delta}^{-\delta} |\psi_\eps(x)|^2\,dx\biggr)^{1/2}\biggl(\int_{-2\delta}^{-\delta}|\phi_j(x)|^2\,dx\biggr)^{1/2}\,.
  \end{aligned}
  \end{equation}
  
  Since $\phi_j\in H^1_0(\R_\limminus)$ we can bound
  \begin{align*}
    \int_{-2\delta}^{-\delta}|\phi_j(x)|^2\,dx
    &= 
    \int_{-2\delta}^{-\delta} \biggl|\int_{x}^0 \phi'_j(t)\,dt\biggr|^2\,dx\\
    & \leq 
    \int_{-2\delta}^{-\delta}\biggl(\int_x^0 1\,dt\biggr) \biggl(\int_x^0 |\phi'_j(t)|^2\,dt\biggr)\,dx\\
     &\leq \frac32\, \delta^2 \int_{-2\delta}^0 |\phi_j'(x)|^2\,dx\,.
  \end{align*}
  Since $\phi_j' \in L^2(\R_\limminus)$, the second integral on both right sides of~\eqref{eq: Sobolev interp} tends to zero. This together with the fact that $\psi_\eps'$ is bounded in $L^2(\R_-)$ uniformly in $\eps$, as is easily checked, implies that the left side in the first equation in~\eqref{eq: Sobolev interp} tends to zero as $\delta\to 0$, uniformly in $j=1, \ldots, k_0$ and $\eps>0$.
  
  For the second equation in~\eqref{eq: Sobolev interp} the same argument as above yields
  \begin{align*}
    \int_{-2\delta}^{-\delta}|\psi_\eps(x)|^2\,dx &= \int_{-2\delta}^{-\delta} \biggl|\psi_\eps(0)-\int_{x}^0 \psi'_\eps(t)\,dt\biggr|^2\,dx\\
    &\lesssim \delta |\psi_\eps(0)|^2+ \int_{-2\delta}^{-\delta}\biggl(\int_x^0 1\,dt\biggr) \biggl(\int_{x}^0 |\psi'_\eps(t)|^2\,dt\biggr)\,dx\\
    &\leq 
    \delta |\psi_\eps(0)|^2+ \frac32\, \delta^2 \int_{-2\delta}^0 |\psi'_\eps(x)|^2\,dx\,.
  \end{align*}
  Since $\lambda_j(\theta)$ are ordered and bounded for all $\theta \in [0, \pi)$, we conclude that $\phi_{j,\theta}(0)\to 0$ as $\theta \to \pi$ for all $j\geq 2$. In particular, for given $\eps>0$ we can choose $\delta = |\psi_\eps(0)|^2$. Again by the boundedness of $\psi_\eps'$ in $L^2(\R_-)$ we conclude that the left side in the second equation in~\eqref{eq: Sobolev interp} tends to zero as $\eps\to 0$, uniformly in $j=1, \ldots, k_0$ and with $\delta=|\psi_\eps(0)|$.
  
  Since by assumption $\lim_{\eps \to 0}\lambda_{k_0+1}(\pi-\eps)> \lambda_{k_0}(0)$, we therefore deduce from~\eqref{eq: almost orthogonality} that the functions $\phi_j$ and $\psi_{\eps, \delta}$ are almost orthogonal. By choosing $\eps$ small enough, the function
  \begin{equation*}
    \tilde \phi= \frac{\psi_{\eps,\delta}- \sum_{j=1}^{k_0} \langle \psi_{\eps, \delta}, \phi_j\rangle \phi_j}{\|\psi_{\eps,\delta}- \sum_{j=1}^{k_0} \langle \psi_{\eps, \delta}, \phi_j\rangle \phi_j\|_{L^2}}
  \end{equation*}
  is well-defined, $L^2$-normalized, orthogonal to $\phi_j$ for $j=1,\ldots,\,k_0$, and in the form domain of $L_0$. The estimate $Q_0(\tilde \phi)<\lambda_{k_0+1}(0)$ follows by the triangle inequality, $\lambda_{k_0+1}(\pi -\eps) <\lambda_{k_0+1}(0)$, and choosing $\eps$ appropriately small. We have arrived at the desired contradiction, thus completing the proof of Lemma~\ref{lem: existence L2 to the left}.
\end{proof}

\section{Relative Pr\"ufer variables}

\subsection{Derivation of the Pr\"ufer equations}
As mentioned in the previous subsection our main results will follow from standard techniques once we understand the asymptotic behavior of solutions of the ODE
\begin{equation}\label{eq: generalized eigenequation}
  \begin{cases}
    -\psi''(x) -Fx \psi(x) = E\psi(x) &\mbox{for } x\in \R\setminus\Z\,, \\
  J\psi(n)=0 & \mbox{for }n \in \Z\,,\\
  J\psi'(n)= g_n \psi(n) &\mbox{for }n \in \Z\,.
  \end{cases}
\end{equation}
In what follows we introduce certain modified Pr\"ufer coordinates for $\psi$ solving~\eqref{eq: generalized eigenequation}. This modified Pr\"ufer set-up has been used in several instances before. In particular, we follow the notation used by Kiselev, Remling, and Simon~\cite{kiselev_effective_1999} where they set up such coordinates in a general framework and successfully use them to derive a variety of spectral theoretic results. While our operators are far from those considered in~\cite{kiselev_effective_1999}, the algebra involved in setting up the coordinates and deriving the associated equations is unchanged. It is perhaps worth noting that the exact same change of coordinates is a key step in Minami's analysis~\cite{minami_random_1992}.

The Pr\"ufer coordinates are defined relative to a reference solution of an unperturbed ODE, in our case the Stark equation~\eqref{eq: Stark eq intro}. The reference solution $\refsol$ should be chosen such that the Wronskian satisfies
\begin{equation}\label{eq: Wronskian assumption}
  \{\refsol, \bar \refsol\}(x)= \refsol(x)\bar \refsol'(x)-\refsol'(x)\bar \refsol(x)= \refsol(0)\bar \refsol'(0)-\refsol'(0)\bar \refsol(0)\neq  0\,.
\end{equation}
Here we shall choose our reference solution to be that defined in~\eqref{eq: ref solution intro}. Since $\Ai$ and $\Bi$ are linearly independent, our choice of reference solution satisfies~\eqref{eq: Wronskian assumption}.

By the assumption that $\{\refsol, \bar \refsol\} \neq 0$ any $\psi$ solving~\eqref{eq: Stark eq intro} on an interval $I\subset \R$ can be written as $\psi(x) = \alpha \refsol(x) + \beta \bar \refsol(x)$, $x\in I$, for uniquely determined constants $\alpha, \beta \in \C$. In particular, this applies to any solution of~\eqref{eq: generalized eigenequation} restricted to an interval of the form $(n-1, n)$, with $n \in \Z$. Given a generalized eigenfunction we wish to understand the change of the coefficients $\alpha, \beta$ when going from the interval $(n-1, n)$ to $(n, n+1)$.

Fix $\psi$ solving~\eqref{eq: generalized eigenequation}. Define $\alpha(n), \beta(n)\in \C$ so that for all $x \in (n-1, n)$
\begin{equation*}
  \begin{cases}
    \psi(x) = \alpha(n) \refsol(x)+ \beta(n)\bar\refsol(x)\\
    \psi'(x) = \alpha(n) \refsol'(x)+ \beta(n)\bar\refsol'(x)\,.
  \end{cases}
\end{equation*}
By the continuity of $\psi$ and $\refsol$ we have, for any $n\in \Z$,
\begin{equation}\label{eq: continuity equation}
  \alpha(n) \refsol(n)+ \beta(n) \bar\refsol(n) = \alpha(n+1) \refsol(n)+ \beta(n+1) \bar\refsol(n)\,.
\end{equation}
Using the fact that $\refsol\in C^\infty(\R)$ the jump condition in~\eqref{eq: generalized eigenequation} can be written equivalently as
\begin{equation}\label{eq: jump equation}
  g_n\bigl[\alpha(n)\refsol(n)+\beta(n)\bar\refsol(n)\bigr]
  = \alpha(n+1)\refsol'(n)+\beta(n+1)\bar\refsol'(n)-\alpha(n)\refsol'(n)-\beta(n)\bar\refsol'(n)\,.
\end{equation}

Writing equations~\eqref{eq: continuity equation} and~\eqref{eq: jump equation} in matrix form we have proved
\begin{equation*}
  \biggl(\begin{matrix}
    \refsol(n) & \bar\refsol(n)\\
    g_n \refsol(n)+\refsol'(n) & g_n \bar\refsol(n)+\bar\refsol'(n)
  \end{matrix}\biggr)
  \biggl(\begin{matrix}
    \alpha(n)\\ \beta(n)
  \end{matrix}\biggr)
  =
  \biggl(\begin{matrix}
    \refsol(n) & \bar\refsol(n)\\
    \refsol'(n) & \bar\refsol'(n)
  \end{matrix}\biggr)
  \biggl(\begin{matrix}
    \alpha(n+1)\\ \beta(n+1)
  \end{matrix}\biggr)\,.
\end{equation*}
Since $\det\bigl(\begin{smallmatrix}
  \refsol & \bar \refsol\\ \refsol' & \bar \refsol'
\end{smallmatrix}\bigr)= \{\refsol, \bar\refsol\} \neq 0$ we thus have the recursion
\begin{equation}\label{eq: recursion}
  \biggl(\begin{matrix}
    \refsol(n) & \bar\refsol(n)\\
    \refsol'(n) & \bar\refsol'(n)
  \end{matrix}\biggr)^{-1}\biggl(\begin{matrix}
    \refsol(n) & \bar\refsol(n)\\
    g_n \refsol(n)+\refsol'(n) & g_n \bar\refsol(n)+\bar\refsol'(n)
  \end{matrix}\biggr)
  \biggl(\begin{matrix}
    \alpha(n)\\ \beta(n)
  \end{matrix}\biggr)
  =
  \biggl(\begin{matrix}
    \alpha(n+1)\\ \beta(n+1)
  \end{matrix}\biggr)\,.
\end{equation}
A direct computation yields
\begin{align*}
  \biggl(\begin{matrix}
    \refsol(n) & \bar\refsol(n)\\
    \refsol'(n) & \bar\refsol'(n)
  \end{matrix}\biggr)^{-1}\!\biggl(\begin{matrix}
    \refsol(n) & \bar\refsol(n)\\
    g_n \refsol(n)+\refsol'(n) & g_n \bar\refsol(n)+\bar\refsol'(n)
  \end{matrix}\biggr) 
  &=
  \1+ \frac{g_n}{\{\refsol, \bar\refsol\}}\biggl(
  \begin{matrix}
    -|\refsol(n)|^2 & -\bar\refsol(n)^2\\
    \refsol(n)^2 & |\refsol(n)|^2
  \end{matrix}\biggr)\,.
\end{align*}
As we shall see in the next subsection we can write $\refsol(x) = |\refsol(x)|e^{i\gamma(x)}$ for a smooth increasing function $\gamma\colon \R \to \R$. In particular, by Lemma~\ref{lem: wronskian identities} below we can write the one-step transfer matrices as
\begin{equation*}
  A_n = \1+ \frac{g_n}{\{\refsol, \bar\refsol\}}\biggl(
  \begin{matrix}
    -|\refsol(n)|^2 & -\bar\refsol(n)^2\\
    \refsol(n)^2 & |\refsol(n)|^2
  \end{matrix}\biggr) = \1+ \frac{g_n}{2i\gamma'(n)}\biggl(
  \begin{matrix}
    1 & e^{-2i\gamma(n)}\\
    -e^{2i\gamma(n)} & -1
  \end{matrix}\biggr)\,.
\end{equation*}
We note that the transfer matrices $A_n$ are elements of the group $\mathbb{SU}(1, 1)$, i.e.\ for each~$n$
\begin{equation*}
  A_n^*\sigma_3A_n = \sigma_3\quad \mbox{and}\quad \det(A_n)=1\,,
\end{equation*}
where 
$\sigma_3 = 
\bigl(\begin{smallmatrix}
  1 & 0\\
  0 & -1
\end{smallmatrix}\bigr)$
 is the third Pauli matrix.

Since the coefficients in the ODE are all real it is sufficient to study real-valued solutions $\psi$ of~\eqref{eq: generalized eigenequation}. For such $\psi$ the vectors $(\begin{smallmatrix}
  \alpha\\ \beta
\end{smallmatrix}) \in \C^2$ in the above representation are of the form $\beta = \bar \alpha$. \emph{From here on we restrict ourselves to studying real-valued generalized eigenfunctions.}

Following~\cite{kiselev_effective_1999} we represent and study our real-valued solution $\psi$ in terms of the complex Pr\"ufer coordinate
\begin{equation*}
  \rho(n)= 2i\alpha(n)
\end{equation*}
and the real-valued Pr\"ufer radius $R\colon \N \to (0, \infty)$ and angle $\eta\colon \N \to \R$ defined by
\begin{equation*}
  \rho(n)=R(n)e^{i\eta(n)}\,,
\end{equation*}
with $\eta(1) \in (-\pi , \pi]$ and $\eta(n+1)-\eta(n)\in (-\pi, \pi]$. For notational convenience we also define $\theta \colon \N \to \R$ by
\begin{equation*}
    \theta(n)= \eta(n)+\gamma(n)\,.
\end{equation*}

To simplify notation we extend the functions $\rho, R, \eta, \theta$ to $\R$ as left continuous step functions by setting $\rho(x) = \rho(\lceil x\rceil)$ and similarly for $R, \eta, \theta$. With these definitions
\begin{equation*}
  \psi(x)= \frac{1}{2i}\Bigl(\rho(x)\refsol(x)-\bar\rho(x)\bar\refsol(x)\Bigr)\,\quad \mbox{for all } x\in (0, \infty).
\end{equation*}

Our understanding of the asymptotic behavior of $\psi$ will be based on studying the recursion equations satisfied by the Pr\"ufer coordinates. We gather these equations in the following lemma, which is essentially~\cite[Theorem~3.3]{kiselev_effective_1999}.
\begin{lemma}\label{lem: Prufer equations}
  Set
  \begin{equation}\label{eq: def U}
    U(n)= \frac{g_n}{\gamma'(n)}\,.
  \end{equation}
  Then
  \begin{align}
    \rho(n+1)-\rho(n)&= U(n)\rho(n)\sin(\theta(n))e^{-i\theta(n)}\label{eq: rho equation}\\
    R(n+1)^2 &= R(n)^2\bigl[1+U(n)\sin(2\theta(n))+U(n)^2\sin(\theta(n))^2\bigr] \label{eq: R equation}\\
    \cot(\eta(n+1)+\gamma(n)) &= \cot(\eta(n)+\gamma(n))+U(n)\, \label{eq: phase equation}.
  \end{align}
  If $|U(n)| \lesssim 1$ then
  \begin{equation}\label{eq: R approx equation}
    \begin{aligned}
      \log\Bigl(\frac{R(n+1)}{R(n)}\Bigr)
    &=
    \frac{U(n)}{2}\sin(2\theta(n)) + \frac{U(n)^2}{8}\\
    &\quad  -\frac{U(n)^2}{8}\Bigl(2\cos(2\theta(n))-\cos(4\theta(n))\Bigr) + O(|U(n)|^3)
    \end{aligned}
  \end{equation}
  and
  \begin{equation}\label{eq: eta approx equation}
  \begin{aligned}
    \eta(n+1)-\eta(n) &=
    -\frac{U(n)}{2}+\frac{U(n)}{2}\cos(2\theta(n))\\
    &\quad 
    +\frac{U(n)^2}{8}\Bigl(2\sin(2\theta(n))- \sin(4\theta(n))\Bigr)+O(|U(n)|^3)\,.
  \end{aligned}
  \end{equation}
  Moreover, if $|U(n)|\leq1$ then 
  \begin{equation}\label{eq: eta jump bound}
    |\eta(n+1)-\eta(n)|
    \leq \frac{\pi}{2}|U(n)|\,.
  \end{equation}
\end{lemma}

\begin{remark}
  The assumption $|U(n)|\lesssim 1$ is true uniformly for $n\geq 1$ in the case of our deterministic model since $|\gamma'(x)| \sim |x|^{1/2}$ as $x\to \infty$ (see Lemma~\ref{lem: gamma asymptotics}). Similarly, in the random model by assumption there exists almost surely $C_\omega<\infty$ such that $|U(n)|= \frac{|g_n(\omega)|}{\gamma'(n)} \leq \frac{C_\omega}{n^{1/4}}$ for all $n\geq 1$.
\end{remark}

\begin{proof}
Equation~\eqref{eq: rho equation} follows from equation~\eqref{eq: recursion} together with the identities in Lemma~\ref{lem: wronskian identities}, and the fact that $\rho(n)=2i\alpha(n)$. Equations~\eqref{eq: R equation} and~\eqref{eq: phase equation} along with the bound~\eqref{eq: eta jump bound} can be deduced form~\eqref{eq: rho equation} precisely as in the proof of Theorem~3.3 in~\cite{kiselev_effective_1999}.

What remains is to prove~\eqref{eq: R approx equation} and~\eqref{eq: eta approx equation}. This is simply a question of appropriate Taylor expansions. 
  By~\eqref{eq: R equation} we have
  \begin{equation*}
    \log\Bigl(\frac{R(n+1)}{R(n)}\Bigr)= \frac{1}{2}\log\Bigl[1+U(n)\sin(2\theta(n))+U(n)^2\sin(\theta (n))^2\Bigr]
  \end{equation*}
  so using the Taylor expansion
  \begin{equation*}
    \log(1+x)=x- \frac{x^2}{2}+O(x^3)
  \end{equation*}
  one finds
  \begin{align*}
    \log\Bigl(\frac{R(n+1)}{R(n)}\Bigr)
    &= 
    \frac{1}{2}\biggl[
      U(n)\sin(2\theta(n))+U(n)^2\sin(\theta (n))^2\\
      &\quad
      -\frac{1}{2}\Bigl(U(n)\sin(2\theta(n))+U(n)^2\sin(\theta (n))^2\Bigr)^2\\
      &\quad 
      +
      O\Bigl(\bigl|U(n)\sin(2\theta(n))+U(n)^2\sin(\theta (n))^2\bigr|^3\Bigr)
    \biggr]\\
     &=
    \frac{U(n)}{2}\sin(2\theta(n)) + \frac{U(n)^2}{8} -\frac{U(n)^2}{8}\Bigl(2\cos(2\theta(n))-\cos(4\theta(n))\Bigr)\\
    &\quad + O(|U(n)|^3)\,,
  \end{align*}
  where we used the identity
  \begin{equation*}
    \sin(x)^2-\frac{1}{2}\sin(2x)^2= \frac{1}{4}+ \frac{1}{4}\cos(4x)- \frac{1}{2}\cos(2x)\,.
  \end{equation*}
Similarly, using the fact that
\begin{equation*}
  \arg(1+ z) = \Im(z)-\Im(z)\Re(z) + O(|z|^3) 
\end{equation*}
and elementary trigonometric identities
\begin{align*}
  \eta&(n+1)-\eta(n)\\
  &=
    \arg\Bigl(\frac{\rho(n+1)}{\rho(n)}\Bigr)\\
  &= 
    \arg\bigl(1+U(n)\sin(\theta(n))e^{-i\theta(n)}\bigr)\\
  &=
    -U(n)\sin(\theta(n))^2+U(n)^2\cos(\theta(n))\sin(\theta(n))^3+O(|U(n)|^3)\\
  &=
    -\frac{U(n)}{2}+\frac{U(n)}{2}\cos(2\theta(n))+\frac{U(n)^2}{4}\Bigl(\sin(2\theta(n))- \frac{1}{2}\sin(4\theta(n))\Bigr)+O(|U(n)|^3)\,.
\end{align*}
This completes the proof of Lemma~\ref{lem: Prufer equations}.
\end{proof}

\subsection{Properties of a particular solution to the Stark equation}
\label{sec: Reference solution}

In the construction that follows we shall need to use certain properties of our reference solution
\begin{equation}\label{eq: ref solution}
  \refsol(x) = \Bigl(\frac{\pi}{F^{1/3}}\Bigr)^{1/2}\bigl(i\Ai(-F^{1/3}(x+E/F)) + \Bi(-F^{1/3}(x+E/F))\bigr)\,.
\end{equation}
Since $\Ai, \Bi$ solve the Airy equation it is easily checked that $\refsol$ is indeed a solution of~\eqref{eq: Stark eq intro}. Since $\Ai, \Bi$ are linearly independent their real zeroes are distinct. Therefore $|\refsol(x)|^2>0$ for all $x$ and $\arg(\refsol)$ is well-defined. We can thus define $\gamma \in C^\infty(\R)$ by
\begin{equation*}
  \refsol(x) = |\refsol(x)|e^{i\gamma(x)}\quad \mbox{with }\gamma(0)\in (-\pi, \pi]\,.
\end{equation*}
Most of the properties of $\refsol$ which are of interest to us concern $\gamma$ and its derivatives. 

While we do not reflect this in our notation, we note that both $\refsol, \gamma$ implicitly depend on the values of $F, E$. It should however be emphasized that $\refsol, \gamma$ do not depend on the coupling constants $\{g_n\}_{n\in \Z}$. We also note that the dependence on the parameters $F, E$ is very simple. Indeed, from the explicit expression for $\refsol$ we have that
\begin{equation*}
  \refsol_{F, E}(x) = F^{-1/6}\refsol_{1,0}(F^{1/3}(x+E/F)) \quad \mbox{and} \quad \gamma_{F,E}(x) = \gamma_{1,0}(F^{1/3}(x+E/F))\,.
\end{equation*}
In particular, all quantitative properties of $\refsol, \gamma$ which we shall discuss are uniform for $F, E$ in compact subsets of their respective domains.

For future reference in the next lemmas we collect some properties of $\refsol$.

\begin{lemma}\label{lem: wronskian identities}
  For any $F>0, E\in \R$ and with $\refsol, \gamma$ as above it holds that
  \begin{equation*}
    \{\refsol, \bar \refsol\} = -2i \quad \mbox{and} \quad |\refsol(x)|^2 = \frac{1}{\gamma'(x)} \mbox{ for all } x\in \R\,.
  \end{equation*}
\end{lemma}

For the record we note the following corollary of the second identity in Lemma~\ref{lem: wronskian identities} when combined with the fact that $\Ai, \Bi$ are continuous and tend to zero as $x \to -\infty$.
\begin{corollary}\label{cor: gamma' lower bound}
  For any $F>0, E\in \R$ there exists a $\delta>0$ such that $\gamma'(x)> \delta$ for all $x>0$. Moreover, $\delta$ can be chosen uniform for $F, E$ in compact subsets of $\R_\limplus$ and $\R$, respectively.
\end{corollary}

\begin{lemma}\label{lem: gamma asymptotics}
  As $x \to \infty$ it holds that
\begin{align*}
  \gamma(x) &= \frac{2\sqrt{F}}{3}x^{3/2} + \frac{E}{\sqrt{F}}x^{1/2} + \frac{\pi}{2} + O(x^{-1/2})\,,\\
  \gamma'(x) &= \sqrt{F}x^{1/2} +  O(x^{-1/2})\,,\\
  \gamma''(x)&= \frac{\sqrt{F}}{2}x^{-1/2}+ O(x^{-3/2})\,.
\end{align*}
In particular, $\gamma$ is asymptotically increasing and convex. Corresponding asymptotic expansions can be proved for higher derivatives, here we shall only need that $|\partial_x^k\gamma(x)| \lesssim |x|^{3/2-k}$ for $k=3, \ldots, 7$.
Here all the implicit constants are uniform for $F, E$ in compact subsets of $\R_\limplus$ and $\R$, respectively.
\end{lemma}

\begin{proof}[Proof of Lemma~\ref{lem: wronskian identities}]
By the definition of $\gamma$ there exists a branch of the logarithm such that $\gamma(x)=\Im(\log(\refsol(x)))$. By differentiating we find
\begin{equation*}
  \gamma'(x) = \Im\Bigl(\frac{\refsol'(x)}{\refsol(x)}\Bigr) = \frac{\Im(\refsol'(x)\bar\refsol(x))}{|\refsol(x)|^2} = -\frac{\{\refsol, \bar\refsol\}(x)}{2i|\refsol(x)|^2}
\end{equation*}
and thus
\begin{equation*}
  |\refsol(x)|^2= -\frac{\{\refsol, \bar \refsol\}(x)}{2i\gamma'(x)}.
\end{equation*}

Since $\refsol, \bar\refsol$ solve the same ODE the Wronskian is constant and it suffices to compute its value at any given point. By setting $x=-E/F$ the value claimed in the lemma follows from the fact that (see~\cite[\para 9.2(ii)]{NIST})
\begin{align*}
  \Ai(0) = \frac{1}{3^{2/3}\Gamma(2/3)}\,, \quad \Ai'(0) = - \frac{1}{3^{1/3}\Gamma(1/3)}\,, \quad 
  \Bi(0) = \frac{1}{3^{1/6}\Gamma(2/3)}\,, \quad \Bi'(0) = \frac{3^{1/6}}{\Gamma(1/3)}\,.
\end{align*}
\end{proof}

\begin{proof}[Proof of Lemma~\ref{lem: gamma asymptotics}]
The asymptotic expansion for $\gamma$ can be deduced from that in~\cite[9.8.22]{NIST} by an appropriate change of variables. Asymptotic expansions for the derivatives of $\gamma$ are obtained by justifying term-by-term differentiation. Alternatively it can be reduced to asymptotic formulas for Bessel functions which are somewhat more standard. Indeed, for $x>-E/F$, by the definition of $\gamma$ combined with~\cite[9.8.9 \& 9.8.14]{NIST} we have the identity, with $y = F^{1/3}(x+E/F)$,
\begin{align*}
   \gamma'(x) 
   &= 
   \frac{3}{\pi (x+E/F)(J_{1/3}^2(\frac{2}{3}F^{1/2}(x+E/F)^{3/2})+Y_{1/3}^2(\frac{2}{3}F^{1/2}(x+E/F)^{3/2}))}\\
   &=
   \frac{3F^{1/3}}{\pi y(J_{1/3}^2(\frac{2}{3}y^{3/2})+Y_{1/3}^2(\frac{2}{3}y^{3/2}))}\,.
 \end{align*} 
 By differentiating the identity and using the asymptotic formulae for $J_\nu, Y_\nu$ in~\cite[10.17.3 \& 10.17.4]{NIST} one arrives at the desired expansions. Alternatively, one can use Nicholson's integral representation for $J_\nu^2+Y_\nu^2$~\cite[10.9.30]{NIST} to obtain asymptotics for the denominator directly.
\end{proof}

\subsection{Comparability of asymptotics}
In our analysis we need to understand how the asymptotic behavior of $R$ relates to the growth or decay of the generalized eigenfunction~$\psi$. Specifically we need to understand the behavior of
\begin{equation*}
  \int_0^x |\psi(t)|^2\,dt \quad \mbox{as } x \to \infty\,.
\end{equation*}
That this asymptotic behavior can be understood in terms of that of $R$ is the content of the following lemma.

\begin{lemma}\label{lem: L2 norm comparability}
  Fix $F>0, E\in \R$, and a real-valued $\psi$ which solves~\eqref{eq: Stark eq intro} for $x\in(n-1, n)$ and let $R(n)$ be the associated Pr\"ufer radius. Then
  \begin{equation*}
    \int_{n-1}^{n} |\psi(x)|^2\,dx =\frac{R(n)^2}{2\sqrt{ Fn}}(1+O(n^{-1/2}))\,,
  \end{equation*}
  and
  \begin{equation*}
    \int_{n-1}^{n} |\psi'(x)|^2\,dx = \frac{\sqrt{Fn}R(n)^2}{2}(1+O(n^{-1/2})) \,.
  \end{equation*}
  Moreover, the implicit constants can be chosen uniform for $F, E$ in compact subsets of their respective domains.
\end{lemma}

\begin{proof}
Since the proofs of the two bounds are analogous we write out only the first in detail. The additional ingredient needed for the proof of the second bound is the identity
\begin{equation*}
  \refsol'(x) = \refsol(x)\Bigl(- \frac{\gamma''(x)}{2\gamma'(x)}+i \gamma'(x)\Bigr)\,,
\end{equation*}
which follows from the second identity in Lemma~\ref{lem: wronskian identities}.

For any $x \in (n-1, n)$ we have
\begin{align*}
  |\psi(x)|^2=
  \frac{1}{4}|\rho(n)\refsol(x)- \bar \rho(n)\bar\refsol(x)|^2
  =
  \frac{R(n)^2}{2}|\refsol(x)|^2\Bigl(1-\cos(2\eta(n)+2\gamma(x))\Bigr)\,.
\end{align*}
By the mean value theorem there exists an $x_0\in [n-1, n]$ so that
\begin{align*}
  \int_{n-1}^{n} |\refsol(x)|^2&(1-\cos(2\gamma(x)+2\eta(n))\,dx\\
  &=
  |\refsol(x_0)|^2 \biggl(1-\int_{n-1}^{n} \cos(2\gamma(x)+2\eta(n))\,dx\biggr).
\end{align*}
We claim that the remaining integral is $O(n^{-1/2})$ uniformly in $\eta(n)$. Indeed, by Corollary~\ref{cor: gamma' lower bound} $\gamma$ is increasing, and therefore a change of variables yields
\begin{align*}
  \int_{n-1}^n \cos(2\gamma(x)+2\eta(n))\,dx=
  \int_{\gamma(n-1)}^{\gamma(n)} \frac{\cos(2y+2\eta(n))}{\gamma'(\gamma^{-1}(y))}\,dy\,.
\end{align*}
Moreover, by Lemma~\ref{lem: gamma asymptotics} we conclude that
\begin{align*}
  &\gamma'(x) = \sqrt{F n}\,(1+O(n^{-1})) \quad \mbox{for }x\in [n-1, n]\,,\\
  &|\gamma(n)-\gamma(n-1)| \lesssim \sqrt{n}\,.
\end{align*}
Consequently,
\begin{align*}
  \int_{\gamma(n-1)}^{\gamma(n)} \frac{\cos(2y+2\eta(n))}{\gamma'(\gamma^{-1}(y))}\,dy = \frac{1}{\sqrt{Fn}}\int_{\gamma(n-1)}^{\gamma(n)} \cos(2y+2\eta(n))\,dy + O(n^{-1})\,,
\end{align*}
and the claim follows since the modulus of the cosine integral is at most $2$.

Therefore, there exists $x_0 \in [n-1, n]$ such that
\begin{align*}
  \int_{n-1}^{n}|\psi(x)|^2\,dx = \frac{1}{2}R(n)^2|\refsol(x_0)|^2\bigl(1+O(n^{-1/2})\bigr)\,.
\end{align*}
By Lemmas~\ref{lem: wronskian identities} and~\ref{lem: gamma asymptotics}
\begin{equation*}
   |\refsol(x_0)|^2 = \frac{1}{\sqrt{F n}}(1+O(n^{-1})) \quad \mbox{as }n\to \infty\,,
 \end{equation*} 
 and hence we conclude that
\begin{equation*}
   \int_{n-1}^{n}|\psi(x)|^2\,dx = \frac{R(n)^2}{2\sqrt{ Fn}}\bigl(1+O(n^{-1/2})\bigr)\,,
\end{equation*}
which completes the proof of the lemma.
\end{proof}

\section{Analysis of the random model}

In this section we analyse the random model with the goal of proving Theorem~\ref{thm: full-line random}. The reader who is only interested in the deterministic model can with the exception of Section~\ref{sec: preliminary exp sum bounds} jump directly to Section~\ref{sec: Refined exponential bounds}.

\subsection{Exponential sums with phase \texorpdfstring{$\gamma$}{gamma}}
\label{sec: preliminary exp sum bounds}

A central part of the analysis which is to follow will be understanding the behavior of certain exponential sums of the form
\begin{equation*}
  \sum_{a<n\leq b}u(n) e^{i\varphi(n)}\,.
\end{equation*}
Specifically, we will encounter such sums where the amplitude $u$ is a negative power of $\gamma'$ and the phase $\varphi$ given in terms of $\gamma(n)$. In this section we collect a number of bounds for such sums which we shall frequently use later on.

We begin by recalling the classical bound of van der Corput (see for instance~\cite[Theorem~5.9]{Titchmarsh_RiemanZeta86}).
\begin{lemma}\label{lem:Van der Corput}
 Let $a<b$ with $b-a\geq 1$ and let $f\in C^2([a, b])$ be a real function satisfying
 \begin{equation*}
   \kappa \leq |f''(x)|\leq h \kappa \quad \mbox{for all } x\in [a, b]\,.
 \end{equation*}
 Then
 \begin{equation*}
   \Biggl|\sum_{a<n\leq b}e^{2\pi i f(n)}\Biggr| \lesssim h(b-a)\kappa^{1/2}+ \kappa^{-1/2}\,.
 \end{equation*}
\end{lemma}

In the setting considered here the following almost direct corollary of Lemma~\ref{lem:Van der Corput} will be most relevant.
\begin{corollary}\label{cor: Rough exp sum bound}
Fix $\mu, F>0$ and $E\in \R$. Then for any $0<a<b$,
  \begin{equation*}
      \biggl|\sum_{a<j\leq b} e^{\mu i \gamma(j)}\biggr| \lesssim \frac{b^{1/4}}{a^{1/2}}(b-a+ a^{1/2})\,.
  \end{equation*}
  Moreover, the implicit constant can be taken uniform for $\mu, F, E$ in compact subsets of their respective domains.
\end{corollary}
\begin{proof}
  By the asymptotic behavior of $\gamma'''$ there exists an $A>0$ large enough so that $\gamma''$ is monotone decreasing on $[A, \infty)$.

  If $a < b \leq  A$ we bound the sum trivially by $b-a$, which is less than the right-hand side of the claimed bound provided the implicit constant is sufficiently large. If $a< A < b$ we split the sum into two pieces: one with $a<j \leq A$ and one with $A<j\leq b$. The first of the two we bound by $A-a$, which again is bounded by the right-hand side of the claimed inequality. Since the right-hand side of the inequality is increasing with respect $b$ and decreasing with respect to $a$ it only remains to prove the bound under the assumption $a\geq A$.

 If $a\geq A$, the monotonicity of $\gamma''$ implies
\begin{equation*}
  |\gamma''(b)|\leq |\gamma''(x)| \leq |\gamma''(a)| \quad \mbox{for all }x\in (a, b]\,.
\end{equation*}
Thus with $\kappa= \mu|\gamma''(b)|/(2\pi)$ and $h = \frac{|\gamma''(a)|}{|\gamma''(b)|} \lesssim \bigl(\frac{b}{a}\bigr)^{1/2}$ the conditions of Lemma~\ref{lem:Van der Corput} hold and we find
\begin{align*}
  \Biggl|\sum_{a<j \leq b} e^{i\mu \gamma(j)}\Biggr| 
  &\lesssim 
  (b/a)^{1/2}(b-a)|\gamma''(b)|^{1/2}+ |\gamma''(b)|^{-1/2}\\
  &\lesssim
  \frac{b^{1/4}}{a^{1/2}}(b-a+ a^{1/2})\,,
\end{align*}
which completes the proof of the lemma.
\end{proof}

For our proof of Theorem~\ref{thm: full-line deterministic} we shall require very precise estimates for exponential sums with amplitude given in terms of $\gamma'$ and phase function $\gamma$. These estimates, which become rather technical, are the topic of Section~\ref{sec: Refined exponential bounds}. However, in our analysis of the random model it shall suffice to have the following bound which is similar in spirit but much cruder than what we shall prove later on.

\begin{theorem}\label{thm: crude bound expsum with decay}
  Let $\mu>0$ and $h\colon \N \to \R$ satisfy
  \begin{equation}\label{eq: continuity h}
      \lim_{j\to \infty}|h(j+1)-h(j)|j^{1/4} =0\,.
  \end{equation}
  Then, 
  \begin{equation*}
    \sum_{j=1}^N \frac{e^{i(\mu\gamma(j)+h(j))}}{\gamma'(j)^2} = o(\log(N))
  \end{equation*}
  where the error term can be quantified in terms of that in~\eqref{eq: continuity h} and chosen uniform for $\mu, F, E$ in compact subsets of their respective domains.
\end{theorem}

\begin{remark}
  If the convergence rate in~\eqref{eq: continuity h} is sufficiently rapid, for instance $j^{-\eps}$ for some $\eps>0$, then a straight-forward modification of the proof below implies the convergence of the series
  \begin{equation*}
    \sum_{j=1}^\infty \frac{e^{i(\mu\gamma(j)+h(j))}}{\gamma'(j)^2}\,.
  \end{equation*}
\end{remark}

\begin{proof}
  By assumption there exists a function $\epsilon\colon \R_\limplus \to [0, \infty)$ with $\lim_{j \to \infty}\epsilon(j)=0$ such that
  \begin{equation*}
      |h(j+1)-h(j)|j^{1/4}\leq \epsilon(j)\,.
  \end{equation*}
  Without loss of generality we may assume that $\epsilon$ is bounded and monotonically decreasing. 
  
  We decompose the sum into sums with $(l-1)^2<j\leq l^2$, writing
  \begin{align*}
    \sum_{j=1}^N\frac{e^{i(\mu\gamma(j)+h(j))}}{\gamma'(j)^2}
    &=
    \sum_{1\leq l\leq \sqrt{N}}\Biggl[\sum_{(l-1)^2<j\leq l^2}\frac{e^{i(\mu\gamma(j)+h(j))}}{\gamma'(j)^2}\Biggr]
    +
    \sum_{\lfloor \sqrt{N}\rfloor^2 <j \leq N}\frac{e^{i(\mu\gamma(j)+h(j))}}{\gamma'(j)^2}
    \,.
  \end{align*}
  For the second sum the asymptotic behavior of $\gamma'$ in Lemma~\ref{lem: gamma asymptotics} implies
  \begin{align*}
    \biggl|\sum_{\lfloor \sqrt{N}\rfloor^2 <j \leq N}\frac{e^{i(\mu\gamma(j)+h(j))}}{\gamma'(j)^2}\biggr|
    \lesssim
    \frac{N-\lfloor\sqrt{N}\rfloor^2}{N} \lesssim \frac{1}{\sqrt{N}}\,.
  \end{align*}

  For fixed $l$ a summation by parts yields
  \begin{align*}
    \sum_{(l-1)^2<j\leq l^2}\frac{e^{i(\mu\gamma(j)+h(j))}}{\gamma'(j)^2}
    &=
     \sum_{(l-1)^2<j\leq l^2-1}\biggl[\frac{e^{i h(j)}}{\gamma'(j)^2}-\frac{e^{i h(j+1)}}{\gamma'(j+1)^2}\biggr]\sum_{(l-1)^2<k\leq j}e^{i\mu\gamma(k)}\\
     &\quad
     + \frac{e^{i h(l^2)}}{\gamma'(l^2)^2}\sum_{(l-1)^2<j\leq l^2}e^{i \mu\gamma(j)}\,.
  \end{align*}
  Using the fundamental theorem of calculus, the asymptotics of Lemma~\ref{lem: gamma asymptotics}, and the assumption on $h$ we can estimate
  \begin{align*}
    \Bigl|\frac{e^{i h(j)}}{\gamma'(j)^2}-\frac{e^{i h(j+1)}}{\gamma'(j+1)^2}\Bigr|
    &\leq
    \Bigl|\frac{1}{\gamma'(j)^2}-\frac{1}{\gamma'(j+1)^2}\Bigr| + \frac{1}{\gamma'(j+1)^2}\Bigl|e^{ih(j)}-e^{ih(j+1)}\Bigr|\\
    &\lesssim
    \frac{1}{\gamma'(j)^2\gamma'(j+1)^2}\Bigl|\int_j^{j+1}\gamma'(t)\gamma''(t)\,dt\Bigr| + \frac{\epsilon(j)}{j^{1/4}\gamma'(j+1)^2}\\
    &\lesssim j^{-5/4}(j^{-3/4}+\epsilon(j))\,.
  \end{align*}

  Inserting this into the above we find 
  \begin{align*}
    \biggl|\sum_{(l-1)^2<j\leq l^2}\frac{e^{i(\mu\gamma(j)+h(j))}}{\gamma'(j)^2}\biggr|
    &\lesssim
     l^{-5/2}(l^{-3/2}+\epsilon(l^2))\sum_{(l-1)^2<j\leq l^2}\biggl|\sum_{(l-1)^2<k\leq j}e^{i\mu\gamma(k)}\biggr|\\
     &\quad
     + \frac{1}{l^2}\biggl|\sum_{(l-1)^2<j\leq l^2}e^{i \mu\gamma(j)}\biggr|\,.
  \end{align*}
  By Corollary~\ref{cor: Rough exp sum bound} the remaining exponential sums are $\lesssim l^{1/2}$ so that
   \begin{align*}
    \biggl|\sum_{(l-1)^2<j\leq l^2}\!\!\!\frac{e^{i(\mu\gamma(j)+h(j))}}{\gamma'(j)^2}\biggr|
    \lesssim
    l^{-5/2}(l^{-3/2}+\epsilon(l^2))\!\!\sum_{(l-1)^2<j\leq l^2}\!\!l^{1/2}
     + l^{-3/2}
     \lesssim
     l^{-3/2}(1+l^{1/2}\epsilon(l^2))\,.
  \end{align*}
  Since $l^{-3/2}$ is summable and 
  \begin{align*}
    \sum_{1\leq l \leq \sqrt{N}}l^{-1}\epsilon(l^2)&\leq \|\epsilon\|_{L^\infty}\sum_{1\leq l \leq e^{\sqrt{\log(N)}}}l^{-1} + \epsilon(e^{2\sqrt{\log(N)}})\sum_{e^{\sqrt{\log(N)}}< l\leq \sqrt{N}}l^{-1}\\
    &\lesssim \|\epsilon\|_{L^\infty} \sqrt{\log(N)} + \epsilon(e^{2\sqrt{\log(N)}})\log(N)\,,
  \end{align*}
  this concludes the proof.
\end{proof}

\subsection{Asymptotics of \texorpdfstring{$R$}{R} for non-subordinate solutions}

The aim of this section is to analyse the behavior of non-subordinate solutions of~\eqref{eq: generalized eigenequation}. Since the subordinate solutions constitute a one-dimensional subspace of the two-dimensional solution space, all non-subordinate solutions have essentially the same asymptotic behavior. Describing the asymptotic behavior is precisely the content of our next theorem, which is a direct analogue of~\cite[Theorem 8.2]{kiselev_EFGP_1998} and our proof follows that given there.

\begin{theorem}\label{thm: nonsub. asymptotics random model}
  Suppose that $\{g_n(\omega)\}_{n\in \Z}$ are independent random variables satisfying
  \begin{enumerate}[label=(\roman*)]
    \item\label{ass: expectation} $\E_\omega[g_n]=0$,
    \item\label{ass: variance} $\E_\omega[g_n^2]= \lambda^2$,
    \item\label{ass: fourth moment} $\sum_{n \geq 1}\E_\omega[g_n^4]n^{-2} <\infty$, and
    \item\label{ass: limit} $\P_\omega\bigl[\lim_{n \to \infty} g_n n^{-1/4}=0\bigr]=1$.
  \end{enumerate}

  Fix $F>0, E\in \R$, and $\theta_0\in [0, \pi)$. Then, almost surely, the Pr\"ufer radius $R$ associated to the solution $\psi$ of~\eqref{eq: generalized eigenequation} with
  \begin{equation*}
    \psi(0) = \sin(\theta_0) \quad \mbox{and}\quad \psi'(0^\limplus)=\cos(\theta_0)
  \end{equation*}
satisfies
  \begin{equation*}
   \lim_{n\to \infty} \frac{\log(R(n))}{\log(n)} = \frac{\lambda^2}{8 F}\,.
  \end{equation*}
\end{theorem}

\begin{remark}\label{rem: assumptions random model}
  As noted earlier, the assumptions are valid for $g_n$ chosen as i.i.d.\ copies of a random variable $X$ with mean zero and variance $\lambda^2$ if and only if $\E_\omega[|X|^4]<\infty$. Let us show this.

  By the dominated convergence theorem
  \begin{equation*}
    \P_\omega\Bigl[\,\lim_{n\to \infty} g_n n^{-1/4}=0\Bigr] = \lim_{\substack{N\to \infty\\\eps \to 0^\limplus}} \P_\omega \Bigl[\,\sup_{n\geq N}|g_n|n^{-1/4}<\eps\Bigr]\,.
  \end{equation*}
  Assume that $N\in \N$ is chosen so large that $\P_\omega[|g_n|< \eps n^{1/4}]>1/2$ for all $n\geq N$ (since $\E_\omega[|g_n|^2]=\lambda^2$ this is possible by Chebyshev's inequality). By the independence of the $g_n$ and the inequality $-2x \leq \log(1-x)\leq -x$, for $x\in [0, 1/2]$, we find
  \begin{equation*}
  e^{-2\sum_{n=N}^\infty \P_\omega[|g_n|\geq\eps n^{1/4}]}\leq
  \P_\omega\Bigl[\,\sup_{n\geq N} |g_n|n^{-1/4}< \eps \Bigr]
  \leq 
  e^{-\sum_{n=N}^\infty \P_\omega[|g_n|\geq \eps n^{1/4}]}\,.
  \end{equation*}
  If for some $\eps>0$ the series in the exponents diverges then the probability that the limit vanishes is zero. Conversely, the probability that the limit vanishes is positive if for any $\eps>0$ there exists an $N$ so that the series converges. If for any $\eps>0$ such an $N$ exists the monotone convergence theorem allows us to send $N$ to infinity and conclude that the probability is in fact $1$. We also note that if for some $\eps>0, N>0$ the series is convergent, then since $\{|g_n|n^{1/4}\}_{n=1}^N$ are uniformly bounded with probability $1$ the monotone convergence theorem this time applied as $\eps \to \infty$ implies
  \begin{equation*}
    \P_\omega\Bigl[\,\sup_{n\geq 1}|g_n|n^{-1/4}<\infty\Bigr]=1\,.
  \end{equation*}

  If the $g_n$ are i.i.d.\ copies of a random variable $X$ the layer cake representation implies that for every non-negative function $f$, $\P_\omega[f(|X|)\geq n]\in l^1$ is equivalent to $\E_\omega[f(|X|)]<\infty$. Since $\E_\omega[(|X|/\eps)^4]= \eps^{-4}\E_\omega[|X|^4]$ for all $\eps>0$ the series above is convergent if and only if $\E_\omega[|X|^4]<\infty$. By the above argument~\ref{ass: limit} is valid if and only if $\E_\omega[|X|^4]<\infty$.

  For non-identically distributed $g_n$ we can use essentially the same argument with an additional application of Markov's inequality $\P_\omega[ |g_n|\geq  x] \leq \frac{\E_\omega[f(|g_n|)]}{f(x)}$, for a non-negative increasing function $f$, to deduce the validity of~\ref{ass: limit} if $\frac{\E_\omega[f(|g_n|)]}{f(\eps n^{1/4})}\in l^1$ for all $\eps>0$. In particular, this can be applied to show the validity of our assumptions if $\E_\omega[|g_n|^\alpha]$ is uniformly bounded for some $\alpha>4$. 
\end{remark}

\begin{proof}[Proof of Theorem~\ref{thm: nonsub. asymptotics random model}]
  Fix a typical realization $\{g_n\}_{n\in \Z}=\{g_n(\omega)\}_{n\in \Z}$ and let $R, \theta$ denote the Pr\"ufer coordinates corresponding to $\psi$. Recall that $U(n) = \frac{g_n}{\gamma'(n)}$. 
  Since $|U(n)|\to 0$ as $n\to \infty$ (a.s.), 
  equation~\eqref{eq: R approx equation} yields, for $n$ large enough,
  \begin{equation}\label{eq: R eq nonsub. asymptotics}
  \begin{aligned}
    \log\Bigl(\frac{R(n+1)}{R(n)}\Bigr)
    &=
    \frac{U(n)}{2}\sin(2\theta(n))+ \frac{U(n)^2}{8}\\
    &\quad  - \frac{U(n)^2}{8}\Bigl(2\cos(2\theta(n))-\cos(4\theta(n))\Bigr)
      + O(|U(n)|^3)\,.
  \end{aligned}
  \end{equation}
  By~\ref{ass: limit} almost surely there exists $C_\omega <\infty$ such that  $|g_n(\omega)| \leq C_\omega n^{1/4}$ for all $n\geq 1$. Therefore, by Lemma~\ref{lem: gamma asymptotics} and~\cite[Lemma~8.4]{kiselev_EFGP_1998} almost surely
  \begin{align*}
    \sum_{n\geq 1}|U(n)|^3 &\lesssim 
    C_\omega \sum_{n\geq 1}\frac{g_n(\omega)^2}{n^{5/4}}
    =
    C_\omega \sum_{n\geq 1}\frac{\lambda^2}{n^{5/4}}+C_\omega \sum_{n\geq 1}\frac{g_n(\omega)^2-\lambda^2}{n^{5/4}} \lesssim 1\,,
  \end{align*}
  so the error term in~\eqref{eq: R eq nonsub. asymptotics} is summable. Therefore, repeated use of equation~\eqref{eq: R eq nonsub. asymptotics} and rearranging yields
  \begin{align*}
    \log\Bigl(\frac{R(n+1)}{R(1)}\Bigr) 
    &=
    \frac{1}{8}\sum_{j=1}^n \E_\omega[U(j)^2]
     +
    \frac{1}{2}\sum_{j=1}^n U(j)\sin(2\theta(j))\\
    &\quad +
    \frac{1}{8}\sum_{j=1}^n \bigl[U(j)^2-\E_\omega[U(j)^2]\bigr]\Bigl[1-2\cos(2\theta(j))+\cos(4\theta(j))\Bigr]\\
    &\quad -
    \frac{1}{8}\sum_{j=1}^n \E_\omega[U(j)^2]\Bigl[2\cos(2\theta(j))-\cos(4\theta(j))\Bigr] +
    O(1)\,.
  \end{align*}
  Define
  \begin{align*}
    C_1(n, \omega) &=\frac{1}{2}\sum_{j=1}^n U(j)\sin(2\theta(j))\,,\\
    C_2(n, \omega) &= \frac{1}{8}\sum_{j=1}^n \bigl[U(j)^2-\E_\omega[U(j)^2]\bigr]\Bigl[1-2\cos(2\theta(j))+\cos(4\theta(j))\Bigr]\,,\\
    C_3(n, \omega) &= \frac{1}{8}\sum_{j=1}^n \E_\omega[U(j)^2]\Bigl[2\cos(2\theta(j))-\cos(4\theta(j))\Bigr]\,,
  \end{align*}
  so that
  \begin{equation*}
    \log\Bigl(\frac{R(n+1)}{R(1)}\Bigr) = \frac{1}{8}\sum_{j=1}^n \E_\omega[U(j)^2] + \sum_{m=1}^3C_m(n, \omega) + O(1)\,.
  \end{equation*}

  By application of~\cite[Lemma 8.4]{kiselev_EFGP_1998} we find, almost surely, $|C_1(n, \omega)| \lesssim o(\log(n))$ and $|C_3(n, \omega)| \lesssim 1$ where we used $\gamma'(j)\sim j^{1/2}$ and our assumptions on $g_j$.

  By Lemma~\ref{lem: Prufer equations} and~\ref{ass: limit} we have almost surely that $|\eta(j+1)-\eta(j)| \lesssim |U(j)|\lesssim o(j^{-1/4})$. Thus we can apply Theorem~\ref{thm: crude bound expsum with decay}, with $\mu = 2, 4$ and $h(j) = \mu \eta(j)$, to deduce
  \begin{equation*}
    \biggl|\sum_{j=1}^n \E_\omega[U(j)^2]\cos(\mu\theta(j))\biggr|
    \leq
    \lambda^2\biggl|\sum_{j=1}^n \frac{e^{i(\mu\gamma(j)+\mu\eta(j))}}{\gamma'(j)^2}\biggr|
    = o(\log(n))\,.
  \end{equation*}
Consequently, the term $C_3(n, \omega)$ does not contribute to the limit in the theorem.

By Lemma~\ref{lem: gamma asymptotics}, $\gamma'(x)  = F^{1/2}x^{1/2}+ O(x^{-1/2})$ so that
  \begin{equation*}
    \frac{1}{8}\sum_{j=1}^n \E_\omega[U(j)^2]= \frac{\lambda^2}{8}\sum_{j=1}^n \frac{1}{\gamma'(j)^2} = \frac{\lambda^2}{8F}\log(n)+O(1)\,,
  \end{equation*}
  which completes the proof of Theorem~\ref{thm: nonsub. asymptotics random model}.
\end{proof}

To turn our knowledge of the behavior of non-subordinate solutions into information about subordinate solutions we shall need the following theorem which tells us that the asymptotic behavior of two linearly independent solutions typically differs only by a constant factor. 

\begin{theorem}\label{thm: ratio of nonsub. solutions}
  Fix $F>0, E\in \R$, and $\theta_\limplus, \theta_\limminus \in [0, \pi)$ with $\theta_\limplus \neq \theta_\limminus$. Assume that $\{g_n(\omega)\}_{n\geq 1}$ satisfy the hypothesis of Theorem~\ref{thm: nonsub. asymptotics random model}. Let $\psi_\limplus, \psi_\limminus$ be the two solutions of~\eqref{eq: generalized eigenequation} corresponding to the boundary condition
  \begin{equation*}
    \psi_\limpm(0)= \sin(\theta_\limpm) \quad \mbox{and} \quad \psi_\limpm'(0^\limplus)= \cos(\theta_\limpm)
  \end{equation*}
   and let $R_\limplus, R_\limminus$ denote the corresponding Pr\"ufer radii. Then, almost surely, there exists $\varrho_\infty \in (0, \infty)$ such that
  \begin{equation*}
    \limsup_{n\to \infty}\frac{\log\Bigl|\tfrac{R_\limplus(n)}{R_\limminus(n)} - \varrho_\infty\Bigr|}{\log(n)} \leq - \frac{\lambda^2}{4F}\,.
  \end{equation*}
\end{theorem}

\begin{proof}
Since by assumption $\theta_\limplus \neq \theta_\limminus$ we have
\begin{equation*}
  \{\psi_\limplus, \psi_\limminus\}(x)= \{\psi_\limplus, \psi_\limminus\}(0) = \sin(\theta_\limplus-\theta_\limminus) \neq 0\,,
\end{equation*} 
so the solutions $\psi_\limplus$ and $\psi_\limminus$ are linearly independent.

By expressing the Wronskian in terms of our Pr\"ufer coordinates we also have for $x \in (n-1, n)$
\begin{align*}
  \{\psi_\limplus, \psi_\limminus\}(x)
  &=
  \Bigl\{ \frac{R_\limplus e^{i\eta_\limplus}}{2i}\refsol
  -\frac{R_\limplus e^{-i\eta_\limplus}}{2i}\bar\refsol
  , \frac{R_\limminus e^{i\eta_\limminus }}{2i}\refsol-\frac{R_\limminus e^{-i\eta_\limminus }}{2i}\bar\refsol\Bigr\}(x)\\
  &= 
  \frac{i}{2}R_\limplus(n)R_\limminus(n)\sin(\eta_\limplus(n)- \eta_\limminus(n))\{\refsol, \bar\refsol\}(x)\\
  &= 
  R_\limplus(n)R_\limminus(n)\sin(\eta_\limplus(n)- \eta_\limminus(n))
\end{align*}
and so, since the Wronskian is constant and $R_\limplus, R_\limminus$ non-zero,
\begin{equation}\label{eq: R1R2 relation}
  {\sin(\eta_\limplus(n)- \eta_\limminus(n))}  =\frac{\{\psi_\limplus, \psi_\limminus\}(0)}{R_\limplus(n)R_\limminus(n)}\,.
\end{equation}

By Theorem~\ref{thm: nonsub. asymptotics random model}, almost surely
\begin{equation*}
  \lim_{n\to \infty} \frac{\log (R_\limpm(n))}{\log(n)} = \frac{\lambda^2}{8F}\,.
\end{equation*}
Thus by~\eqref{eq: R1R2 relation} with $\theta_\pm(n)=\eta_\pm(n)+\gamma(n)$, almost surely
\begin{equation}\label{eq: convergence eta1 eta2}
  \lim_{n\to \infty} \frac{\log|\sin(\theta_\limplus(n)-\theta_\limminus(n))|}{\log(n)} = \lim_{n\to \infty} \frac{\log|\sin(\eta_\limplus(n)-\eta_\limminus(n))|}{\log(n)}= -\frac{\lambda^2}{4F}\,,
\end{equation}
which implies that there exists a sequence of integers $k_n$ such that
\begin{equation}\label{eq: convergence theta1 theta2}
   \lim_{n\to \infty} \frac{\log|\theta_\limplus(n)-\theta_\limminus(n)-k_n \pi|}{\log(n)} = -\frac{\lambda^2}{4F}\,.
\end{equation}

Set $\varrho(n) = \frac{R_\limplus(n)}{R_\limminus(n)}$. We wish to prove that $\varrho(n)$ converges sufficiently fast to a limit different from $0$ and $\infty$.

To this end let $M(n)= \log(\varrho(n+1))-\log(\varrho(n))$ and use Lemma~\ref{lem: Prufer equations} to deduce
\begin{align*}
  M(n) &= \log(\varrho(n+1))-\log(\varrho(n))\\
  &= 
  \frac{1}{2}\log\Bigl(\frac{R_\limplus(n+1)^2}{R_\limplus(n)^2}\Bigr)- \frac{1}{2}\log\Bigl(\frac{R_\limminus(n+1)^2}{R_\limminus(n)^2}\Bigr)\\
  &=
  \frac{1}{2}\Bigl[\log\Bigl(1+ U(n)\sin(2\theta_\limplus(n))+U(n)^2\sin(\theta_\limplus(n))^2\Bigr)\\
  &\qquad -\log\Bigl(1+ U(n)\sin(2\theta_\limminus(n))+U(n)^2\sin(\theta_\limminus(n))^2\Bigr)\Bigr]\\
  &=
  \frac{1}{2}\Bigl[\log\Bigl(1+ U(n)\sin(2\theta_\limplus(n)-2k_n\pi)+U(n)^2\sin(\theta_\limplus(n)-k_n\pi)^2\Bigr)\\
  &\qquad -\log\Bigl(1+ U(n)\sin(2\theta_\limminus(n))+U(n)^2\sin(\theta_\limminus(n))^2\Bigr)\Bigr]\,.
\end{align*}

Set
\begin{equation*}
    K(a, \theta) = \log(1+a \sin(2\theta)+a^2\sin(\theta)^2)\,.
\end{equation*}
By a Taylor expansion with respect to $a$,
\begin{equation*}
  K(a, \theta) = \sum_{j=1}^{J-1}a^j P_j(\theta)+O(a^J)
\end{equation*}
with $P_1(\theta)=\sin(2\theta)$ and each $P_j\in C^\infty(\R)$. By~\eqref{eq: convergence theta1 theta2} we know that $|\theta_\limplus(n)-\theta_\limminus(n)-k_n\pi|=O(n^{-\lambda^2/(4F)+\eps})$ for any $\eps>0$, and thus almost surely
\begin{align*}
  K(U(n), \theta_\limplus(n)&-k_n\pi)-K(U(n), \theta_\limminus(n))\\
  &= 
  \sum_{j=1}^{J-1} U(n)^j\bigl[P_j(\theta_\limplus(n)-k_n\pi)-P_j(\theta_\limminus(n))\bigr] + O(|U(n)|^J)\\
  &= 
  U(n)\Bigl[\sin(2\theta_\limplus(n)-2k_n\pi)-\sin(2\theta_\limminus(n))\Bigr]\\
  &\quad +\sum_{j=2}^{J-1}  O\Bigl(|U(n)|^j|\theta_\limplus(n)-\theta_\limminus(n)-k_n\pi|\Bigr) + O(|U(n)|^J)\\
  &=
  U(n)\Bigl[\sin(2\theta_\limplus(n))-\sin(2\theta_\limminus(n))\Bigr]\\
  &\quad +O\Bigl(U(n)^2n^{-\lambda^2/(4F)+\eps}\Bigr) + O\Bigl(U(n)^2 n^{-(J-2)/4}\Bigr)\,. 
\end{align*}
Here we used the fact that by~\ref{ass: limit} and Lemma~\ref{lem: gamma asymptotics} almost surely $|U(n)|\leq C_\omega n^{-1/4}$. 

By choosing $J$ sufficiently large (any $J \geq \frac{\lambda^2}{F}+2$ suffices), and since $\theta_\limpm(n)$ only depend on $\{g_k\}_{k\leq n-1}$ one can use the bounds of~\cite[Lemma 8.3]{kiselev_EFGP_1998} to deduce that almost surely
\begin{equation*}
  \lim_{N\to \infty}\log(\varrho(N)) = \log(\varrho(1))+\lim_{n\to \infty}\sum_{n=1}^{N-1} M(n)
\end{equation*}
exists and is finite and
\begin{equation*}
  \biggl|\sum_{n=N}^\infty M(n)\biggr| \lesssim  N^{-\lambda^2/(4F)+2\eps}\,.
\end{equation*}
Consequently, setting
\begin{equation*}
  \varrho_\infty = \lim_{n\to \infty} \varrho(n)
\end{equation*}
we conclude that, almost surely,
\begin{equation*}
  \frac{\log\Bigl| \frac{R_\limplus(n)}{R_\limminus(n)}- \varrho_\infty\Bigr|}{\log(n)} = \frac{\log\Bigl|1- \frac{\varrho_\infty}{\varrho(n)}\Bigr|}{\log(n)} + o(1) \leq -\frac{\lambda^2}{4F}+2\eps+o(1)\,.
\end{equation*}
Since $\eps>0$ is arbitrary this concludes the proof.
\end{proof}

\subsection{Existence of subordinate solution}

In this subsection we turn the statements concerning the behavior of $R$ studied in the previous section into information about the existence of solutions which fail to behave as the typical ones. Specifically, we are interested in the existence of subordinate solutions and whether or not they are square integrable. How the properties of subordinate solutions of~\eqref{eq: generalized eigenequation} relate to the spectral theory of the corresponding differential operator is the content of the Gilbert--Pearson subordinacy theory originally developed in~\cite{GilbertPearson_subordinacy_1987,gilbert_subordinacy_1989}. While our equation is not a standard differential equation on $\R$, but rather an infinite a system of coupled equations on unit intervals, the subordinacy theory goes through with obvious changes. Before we begin we recall the precise definition of subordinate solutions.

Let $L$ be a second order differential expression and assume that $L$ is limit point at $\infty$. A non-trivial solution to the equation $L \psi = E\psi$ with $R\in \R$ on $(0, \infty)$ is called subordinate at $\infty$ if
\begin{equation*}
  \lim_{x\to \infty} \frac{\int_{0}^x |\psi(t)|^2\,dt}{\int_{0}^x |\varphi(t)|^2\,dt} =0
\end{equation*}
for any solution of the equation $L\varphi = E\varphi$ which is linearly independent of $\psi$. The subordinacy of a solution at $-\infty$ is defined analogously. Since the solution space of $L\varphi = E\varphi$ is two-dimensional it follows that subordinate solutions are unique up to a multiplicative constant.

By Lemma~\ref{lem: existence L2 to the left}, and the fact that there can only be one $L^2$ solution for the differential equation corresponding to the operator $L_{F,g}$ we always have a non-trivial solution that is
subordinate at $-\infty$.

When it comes to subordinate solutions at $+\infty$ our aim is to prove the following theorem.
\begin{theorem}\label{thm: subordinate solution random}
  Let $F>0, E\in \R$ be fixed and assume that $\{g_n(\omega)\}_{n\geq 1}$ satisfy the hypothesis of Theorem~\ref{thm: nonsub. asymptotics random model}. Then almost surely there exists a $\theta_0= \theta_0(\omega)\in [0, \pi)$ such that the solution of~\eqref{eq: generalized eigenequation} with boundary conditions 
  $$
     \psi(0)=\sin(\theta_0)\,, \quad \psi'(0^\limplus)=\cos(\theta_0)
  $$
  is subordinate at $\infty$ and the associated Pr\"ufer radius $R$ satisfies
  \begin{equation}\label{eq: subordinate decay in thm}
    \lim_{n\to \infty} \frac{\log(R(n))}{\log(n)} = - \frac{\lambda^2}{8F}\,.
  \end{equation}
\end{theorem}

For our proof of Theorem~\ref{thm: full-line random} we shall need the following corollary which follows from Theorem~\ref{thm: subordinate solution random} by an application of Fubini's theorem.
\begin{corollary}\label{cor: subordinate solution random}
  Suppose the $\{g_n(\omega)\}_{n\geq 1}$ satisfy the hypothesis of Theorem~\ref{thm: nonsub. asymptotics random model} and $F>0$ is fixed. Then almost surely for almost every $E \in \R$ there exists a $\theta_0 = \theta_0(E, \omega) \in [0, \pi)$ such that the solution of~\eqref{eq: generalized eigenequation} with boundary conditions 
  $$
    \psi(0)=\sin(\theta_0)\,, \quad \psi'(0^\limplus)=\cos(\theta_0)
  $$ 
  is subordinate at $\infty$ and the associated Pr\"ufer radius $R$ satisfies~\eqref{eq: subordinate decay in thm}.
\end{corollary}

To prove Theorem~\ref{thm: subordinate solution random} we begin by turning the information provided by Theorems~\ref{thm: nonsub. asymptotics random model} and~\ref{thm: ratio of nonsub. solutions} into statements about transfer matrices. After we have accomplished this we can apply results from the OPUC literature to extract the desired existence and properties of the subordinate solution. The overall argument closely follows that of Kiselev--Last--Simon~\cite{kiselev_EFGP_1998} but formulated in slightly different terms since we work with transfer matrices in $\mathbb{SU}(1, 1)$ instead of $\mathbb{SL}(2, \R)$. The analogue of the argument in~\cite{kiselev_EFGP_1998} for $\mathbb{SU}(1, 1)$ matrices can be found in~\cite[\para 10.5]{simon_OPUC2_2005}.

Recall that transfer matrices in the complex Pr\"ufer coordinates are given by
\begin{equation}\label{eq: transfer matrix}
  A_n = \1+\frac{U(n)}{2i}\Bigl(\begin{matrix}
    1 & e^{-2i\gamma(n)}\\
    -e^{2i\gamma(n)} & -1
  \end{matrix}\Bigr)\,.
\end{equation}
We note that by the assumptions on $\{g_n(\omega)\}_{n\in \Z}$ the norms $\|A_n\|$ are almost surely uniformly bounded.
Define the $n$-step transfer matrix
\begin{equation*}
  T_n = A_n\cdots A_1\,.
\end{equation*} 
Since $A_n\in \mathbb{SU}(1, 1)$ we have $T_n\in \mathbb{SU}(1, 1)$.

Recall that for $A \in \mathbb{SU}(1, 1)$ the matrix $|A|=(A^*A)^{1/2}$ has two eigenvalues, $\|A\|$ and $\|A\|^{-1}$. If $\|A\|>1$ in what follows we let $P_\limminus(A)$ denote the orthogonal projection onto the one-dimensional eigenspace corresponding to the eigenvalue $\|A\|^{-1}$, if $\|A\|=1$ then $|A| = \1$ and for notational convenience we let $P_\limminus(A)= \bigl(\begin{smallmatrix}
  1 & 0\\ 0 & 0
\end{smallmatrix}\bigr)$.

We begin by proving the analogue of Theorem~\ref{thm: nonsub. asymptotics random model} for the growth of the norm of $\|T_n\|$.
\begin{lemma}\label{lem: Tn asymptotics}
  Under the assumptions of Theorem~\ref{thm: nonsub. asymptotics random model} almost surely 
  \begin{equation*}
    \lim_{n\to \infty}\frac{\log\|T_n\|}{\log(n)} = \frac{\lambda^2}{8F}\,.
  \end{equation*}
\end{lemma}
\begin{proof}
  If $u \in \C^2$ has the form $u=(u_1, \bar u_1)$ with $u_1 \in \C\setminus \{0\}$, then $|T_n u|=R(n)/\sqrt{2}$ with $R$ corresponding to the real valued solution $\psi$ of~\eqref{eq: generalized eigenequation} defined by the boundary conditions 
  $$
  \psi(0)= u_1\refsol(0)+\bar u_1 \bar \refsol(0) \quad \mbox{and}\quad  \psi'(0^\limplus)= u_1\refsol'(0)+\bar u_1 \bar \refsol'(0)\,.$$ 
  Therefore, $\|T_n\|= \sup_{u \in \C^2, |u|=1}|T_n u| \geq R(n)/\sqrt{2}$ where $R$ is the Pr\"ufer radius associated to any fixed choice of boundary conditions. By Theorem~\ref{thm: nonsub. asymptotics random model}, $\log\|T_n\| \geq \frac{\lambda^2}{8F}\log(n)(1+o(1))$. 

  It remains to prove a matching upper bound. To this end note that for any $u \in \C^2$ and any $A \in \mathbb{SU}(1, 1)$
  \begin{equation*}
    |A u|^2 = \|A\|^2|(\1-P_\limminus(A))u|^2+\|A\|^{-2}|P_\limminus(A) u|^2\,.
  \end{equation*}
  For $A = T_n$ and by dropping the non-negative second term and applying Theorem~\ref{thm: nonsub. asymptotics random model} with two initial conditions corresponding to $u = (u_1, \bar u_1), v=(v_1, \bar v_1)$ such that $(u, v)_{\C^2}=0$ and $|u|=|v|=1$ and denoting the corresponding Pr\"ufer radii by $R_u, R_v$ we deduce that, almost surely,
  \begin{align*}
    \log\|T_n\| &\leq \frac{1}{2}\log\Bigl(\frac{|T_n u|^2+|T_n v|^2}{|(\1-P_\limminus(T_n))u|^2+|(\1-P_\limminus(T_n))v|^2}\Bigr)\\
     &=
     \frac{1}{2}\log\Bigl(\frac{R_u(n)^2+R_v(n)^2}{2}\Bigr)\\
     &= \frac{\lambda^2}{8F}\log(n)(1+o(1))\,,
  \end{align*}
  where we used 
  $$
  |(\1-P_\limminus(T_n))u|^2+|(\1-P_\limminus(T_n))v|^2= \Tr(\1-P_\limminus(T_n))^2=1
  $$ 
  since $(u, v)_{\C^2}=0$ and $|u|=|v|=1$.
  This completes the proof of the lemma.
\end{proof}

A key step in proving Theorem~\ref{thm: subordinate solution random} is provided by the following result which can be found in~\cite{simon_OPUC2_2005} (see also Kotani--Ushiroya~\cite{kotani_one-dimensional_1988}).

 \begin{theorem}[{\cite[Theorem 10.5.34]{simon_OPUC2_2005}}]\label{thm: Convergence initial data}
   Let $A_1, A_2, \ldots$ be elements of $\mathbb{U}(1, 1)$. Let $T_n=A_n \cdots A_1$. Define
   \begin{equation*}
   \varrho(T_n)= \frac{|T_n \bigl(\begin{smallmatrix}
     1\\ 1
   \end{smallmatrix}\bigr)|}{|T_n \bigl(\begin{smallmatrix}
     1\\ -1
   \end{smallmatrix}\bigr)|}\,,
   \end{equation*}
   and assume that
   \begin{equation}\label{eq: assumptions 1 Tn thm}
     \lim_{n\to \infty}\|T_n\|=\infty\, \quad \mbox{and}\quad \lim_{n\to \infty} \frac{\|T_n^{-1}T_{n+1}\|}{\|T_n\|\,\|T_{n+1}\|}
     =0\,.
   \end{equation}
   Then
   \begin{enumerate}
     \item $P_-(T_n)$ has a limit $P_\infty$ if and only if $\lim_{n\to \infty}\varrho(T_n)\equiv \varrho_\infty$ exists. That $\varrho_\infty=\infty$ is allowed, and this holds if $\lim_{n\to \infty}P_-(T_n)$ is the projection onto $\bigl(\begin{smallmatrix}1\\ -1\end{smallmatrix}\bigr)$.

     \item If $\varrho(T_n)$ has a limit $\varrho_\infty \neq 0, \infty$, then, with $u_\infty \in \ran P_\infty$,
     \begin{equation*}
       \lim_{n\to \infty} \frac{\log|T_n u_\infty|}{\log\|T_n\|} =-1
     \end{equation*}
     if and only if
     \begin{equation}\label{eq: assumption 2 Tn thm}
       \limsup_{n\to\infty} \frac{\log|\varrho(T_n)-\varrho_\infty|}{\log\|T_n\|}\leq -2\,.
     \end{equation}
   \end{enumerate}
 \end{theorem}

With these transfer matrix results on hand we are ready to prove Theorem~\ref{thm: subordinate solution random}.

\begin{proof}[Proof of Theorem~\ref{thm: subordinate solution random}]

Note that almost surely $T_{n}^{-1}T_{n+1}=A_{n+1}$ is bounded uniformly in $n$, and  $\|T_n\| \to \infty$ by Lemma~\ref{lem: Tn asymptotics}. Therefore, the assumptions in~\eqref{eq: assumptions 1 Tn thm} are almost surely satisfied.
  
Moreover, since $|T_n(\begin{smallmatrix}1\\-1\end{smallmatrix})|=|T_n(\begin{smallmatrix} i\\-i \end{smallmatrix})|$ the quantity $\varrho(T_n)$ in Theorem~\ref{thm: Convergence initial data} is precisely the quantity $\varrho(n)$ of Theorem~\ref{thm: ratio of nonsub. solutions} with $\psi_\limplus, \psi_\limminus$ chosen to satisfy the boundary conditions corresponding to $\alpha(1)=1$ and $\alpha(1)=i$, respectively. Consequently, Theorem~\ref{thm: ratio of nonsub. solutions} and Lemma~\ref{lem: Tn asymptotics} imply the almost sure validity of~\eqref{eq: assumption 2 Tn thm}. 

By Theorem~\ref{thm: Convergence initial data} there almost surely exist a non-zero $u_\infty\in \C^2$ such that
\begin{equation*}
  \log|T_n u_\infty| = -\frac{\lambda^2}{8F}\log(n)(1+o(1))\,.
\end{equation*}
By Theorem~\ref{thm: nonsub. asymptotics random model} and Lemma~\ref{lem: L2 norm comparability} the corresponding solution is subordinate at $\infty$. However, the provided solution need not be real-valued, i.e.\ $u_\infty$ might not have the form $(z, \bar z)$. But since the real and imaginary part of this solution cannot simultaneously be zero, one of the two provides us with a non-trivial and real-valued solution which is subordinate at $\infty$. This completes the proof of Theorem~\ref{thm: subordinate solution random}.
\end{proof}

\subsection{Proof of Theorem~\ref{thm: full-line random}}

With the results of the previous subsections in place we are ready to prove our main theorem concerning the random model, that is Theorem~\ref{thm: full-line random}. However, in addition to the ODE results proved above we need to recall two results both based on a spectral averaging argument. 

While we have not found these precise statements in the literature they can be proved by following the arguments in~\cite[\para 11--12]{simon_TraceIdeals_2005}. The arguments in~\cite{simon_TraceIdeals_2005} are formulated under the assumption that the operator subject to the rank-one perturbation is bounded from below and the associated spectral measure decays sufficiently to have a well-defined Borel transform. However, for what we are interested in both of these assumptions can be removed without substantially changing the arguments. Indeed, the only property of the Borel transform which is crucial is that it is a Nevanlinna (or Herglotz) function (see the remarks at the beginning of~\cite[\para 11.1]{simon_TraceIdeals_2005} and the notes on~\cite[Theorem~2.5.4]{Simon_HarmonicAnalysis}). The analysis for a rank one perturbation by $\delta$ can be carried out, up to small modifications, as in the case of varying boundary conditions covered in~\cite[\para 11.6]{simon_TraceIdeals_2005}. While it might not be true that the corresponding Titchmarsh--Weyl $m$-functions are well-defined as Borel transforms, they will be Nevanlinna functions as long as the differential operator is limit point. Furthermore, there is an a priori cyclicity assumption in many of the relevant statements in~\cite{simon_TraceIdeals_2005} which need not be valid in our setting. Specifically, to utilize the full power of rank-one perturbation theory requires that $H=\{(L_{F,\lambda}^\omega-z)^{-1}\delta: z\in \C\setminus \R\}$ is a total set in $L^2(\R)$. The subspace $H$ is reducing for $L^\omega_{F, \lambda}$ and by the generalised resolvent identity the operators for different values of $g_0$ coincide on $H^\perp$. Without this totality assumption the spectral averaging arguments only yield information concerning the spectrum of the restriction of $L^\omega_{F,\lambda}$ to $H$.

In the following lemmas we write $L_0$ for the differential expression
\begin{equation*}
  L_0 = - \frac{d^2}{dx^2} - F x + \sum_{n \in \Z\setminus\{0\}}g_n \delta(x-n)
\end{equation*}
where the coupling constants $\{g_n\}_{n\neq 0}$ are such that $L_0$ is limit point at $\pm \infty$ (which is satisfied almost surely under the assumptions of our main result).

\begin{lemma}\label{lem: spec averaging pure point} Fix an open interval $I\subseteq \R$. Consider $L_\alpha$ the family of Schr\"odinger operators defined by the differential expression $L_0 + \alpha \delta$ in $L^2(\R)$. Assume that:
\begin{enumerate}
  \item For almost every $E \in I$ there exists a non-trivial solution of $L_0u(x)=Eu(x)$ on $\R_\limminus$ which is in $L^2(\R_\limminus)$.

  \item For almost every $E \in I$ there exists a non-trivial solution of $L_0u(x)=Eu(x)$ on $\R_\limplus$ which is in $L^2(\R_\limplus)$.

  \item The set $\{(L_0-z)^{-1}\delta: z \in \C \setminus \R\}$ is a total set in $L^2(\R)$.
\end{enumerate}
Then $L_\alpha$ has only pure point spectrum in $I$ for a.e.\ $\alpha$.
\end{lemma}

\begin{lemma}\label{lem: spec averaging only sc}
Fix an open interval $I\subseteq \R$. Consider $L_\alpha$ the family of Schr\"odinger operators defined by the differential expression $L_0 + \alpha \delta$ in $L^2(\R)$. Assume that:
\begin{enumerate}
  \item The Schr\"odinger $L_{\theta}^\limminus$ operator in $L^2(\R_\limminus)$ defined by the differential expression $L_0$ and boundary condition
  \begin{equation*}
    \cos(\theta)u(0)+\sin(\theta)u'(0) =0
  \end{equation*}
  is such that $\sigma_{ac}(L^\limminus_\theta)\cap I=\emptyset$ for all $\theta\in [0, \pi)$.

  \item For almost every $E \in I$ there exists a solution of $L_0u(x)=Eu(x)$ on $\R_\limplus$ which is subordinate at $+\infty$ but which is not in $L^2(\R_\limplus)$.

  \item The set $\{(L_0-z)^{-1}\delta: z \in \C \setminus \R\}$ is a total set in $L^2(\R)$.
\end{enumerate}
Then $L_\alpha$ has only singular continuous spectrum in $I$ for a.e.\ $\alpha$.
\end{lemma}

Before turning to the proof of classifying the spectral nature of $L_{F, \lambda}^\omega$ we prove two preparatory results. The first results proves that almost surely $\sigma_{ess}(L^\omega_{F,\lambda})=\R$ while the second verifies that the third assumption of Lemmas~\ref{lem: spec averaging only sc} and~\ref{lem: spec averaging only sc} is almost surely valid. 
\begin{proposition}\label{prop: essential spectrum}
  Fix $F, \lambda >0$ and let $\{g_n(\omega)\}_{n\in \Z}$ satisfy the assumptions of Theorem~\ref{thm: nonsub. asymptotics random model}. Then almost surely $\sigma_{ess}(L_{F, \lambda}^\omega)=\R$.
\end{proposition}

\begin{proof}
  Fix a $E\in \R$ and let $\psi$ be the solution of~\eqref{eq: generalized eigenequation} with $\psi(0)=0$ and $\psi'(0^\limplus)=1$. Let $R$ be the Pr\"ufer radius associated to $\psi$. By Theorem~\ref{thm: nonsub. asymptotics random model} and Lemma~\ref{lem: L2 norm comparability} almost surely
  \begin{equation}\label{eq: integral asymptotics}
    \int_{n-1}^n |\psi(x)|^2\,dx = n^{-1/2 +\lambda^2/(4F)+o(1)} \quad \mbox{and} \quad \int_{n-1}^n |\psi'(x)|^2\,dx = n^{1/2 +\lambda^2/(4F)+o(1)}\,.
  \end{equation}

  Let $\varphi\in C_0^\infty(0, 1)$ with $0\leq \varphi \leq 1$ and $\varphi \equiv 1$ in $(1/4, 3/4)$. Define $\varphi_k(x) = \varphi(\eps_k x-1)$ for $\eps_k>0$ with $\lim_{k\to \infty} \eps_k \to 0$. Define $\psi_k = \|\varphi_k\psi\|_{L^2(\R)}^{-1}\varphi_k \psi$ so that $\|\psi_k\|_{L^2(\R)}=1$ and $\psi_k \rightharpoonup 0$. Since $\varphi \in C^\infty(\R)$ it is easily checked that $\psi_k  \in D(L^\omega_{F,\lambda})$ and, for $x\in \R \setminus \Z$,
  \begin{equation*}
    L_{F,\lambda}^\omega(\psi_k(x)) -E\psi_k(x) = -\frac{\varphi_k''(x)\psi(x)+2\varphi_k'(x)\psi'(x)}{\|\varphi_k\psi\|_{L^2(\R)}}\,.
  \end{equation*}
  Thus~\eqref{eq: integral asymptotics} implies, for any $0<\delta<1/2$ and since $\eps_k \to 0$,
  \begin{align*}
     \int_{0}^\infty |(L_{F,\lambda}^\omega - E)\psi_k|^2 \,dx
     &\leq 
     \|\varphi_k \psi\|_{L^2(\R_\limplus)}^{-2}\int_0^\infty \Bigl(|\varphi''_k|^2 |\psi|^2 + 4|\varphi'_k|^2|\psi'|^2\Bigr)\,dx\\
     &\lesssim
     \|\varphi_k \psi\|_{L^2(\R_\limplus)}^{-2}\int_{1/\eps_k}^{2/\eps_k} \Bigl(\eps_k^4 |\psi|^2 + \eps_k^2|\psi'|^2\Bigr)\,dx
      \\
     &\lesssim 
     \eps_k^{1-2\delta}(\eps_k^3+1)\,.
   \end{align*}
   Since $\eps_k \to 0$ and $\delta>0$ is arbitrary, we conclude that $\psi_k$ forms a Weyl sequence at energy $E$ and therefore $E\in \sigma_{ess}(L_{F,\lambda}^\omega)$. By repeating the argument for a countable dense set of $E$'s we conclude that almost surely $\sigma_{ess}(L^\omega_{F,\lambda})=\R$.
\end{proof}

\begin{lemma}\label{lem: cyclicity}
  Fix $F, \lambda >0$ and let $g_n=\{g_n(\omega)\}_{n\in \Z\setminus\{0\}}$ satisfy assumptions \ref{ass1: gn indep}--\ref{ass5: gn limit} of Theorem~\ref{thm: full-line random}. Then almost surely $\{(L_0-z)^{-1}\delta: z\in \C \setminus \R\}$ is a total set in $L^2(\R)$.
\end{lemma}

\begin{proof}
  Let $H= \{(L_0-\zeta)^{-1}\delta: \zeta \in \C\}$. As noted, $H$ is a reducing subspace for $L_0$ and on $H^\perp$ all the operators $L_\alpha = L_0+ \alpha \delta$ coincide. In particular, $\sigma(L_\alpha |_{H^\perp})$ is independent of $\alpha$. We aim to prove that almost surely $\sigma(L_0 |_{H^\perp})=\emptyset$ which is impossible unless $H^\perp$ is trivial, i.e.\ $H$ is a total set in $L^2(\R)$.  

  By Lemma~\ref{lem: existence L2 to the left}, Theorem~\ref{thm: subordinate solution random}, and Gilbert--Pearson subordinacy theory, almost surely $\sigma_{ac}(L_\alpha)=\emptyset$ for all $\alpha$. In particular, almost surely $\sigma_{ac}(L_0|_{H^\perp})=\emptyset$.

  By Gilbert--Pearson subordinacy theory an essential support for the singular part of the spectral measure of $L_\alpha$ is given by
  \begin{equation*}
     \Sigma_\alpha = \{E\in \R: \exists \psi\not\equiv 0 \mbox{ solving } L_\alpha \psi = E\psi \mbox{ subordinate at }+\infty \mbox{ and }-\infty\}\,.
   \end{equation*} 
   Since subordinate solutions (in either direction) are unique up to multiplication by constants, $E \in \Sigma_\alpha \cap \Sigma_{\alpha'}$ with $\alpha \neq \alpha'$ if and only if there exists a solution $\psi\not\equiv 0$ subordinate in both directions with $\psi(0)=0$. Conversely, if for some $\alpha$ and $E\in \Sigma_\alpha$ the corresponding solution $\psi$ vanishes at $0$, then $E\in \cap_{\alpha\in \R}\Sigma_{\alpha}$. Write the spectral measure of $L_\alpha$ as $\mu_\alpha = \mu^H_\alpha +\mu^\perp$, where $\mu_\alpha^H$ is the spectral measure of $L_\alpha |_{H}$ and $\mu^\perp$ that of $L_\alpha|_{H^\perp}$. Note that $\mu^\perp$ is independent of $\alpha$. By the before mentioned subordinacy result, any essential support of $\mu^\perp$ is up to $\mu^\perp$-negligible sets contained in $\Sigma_\alpha$ for all $\alpha$. We conclude that an essential support of $\mu^\perp$ is contained in
   \begin{align*}
     & \Sigma^\perp := \cap_{\alpha \in \R}\Sigma_\alpha\\
     & \ =\{E\in \R: \exists \psi\not\equiv 0 \mbox{ solving } L_\alpha \psi = E\psi \mbox{ subordinate at }+\infty \mbox{ and }-\infty \mbox{ with }\psi(0)=0\}\,.
   \end{align*}

  Let $L^\limminus$ denote the realisation of $L_0$ in $L^2(\R_\limminus)$ with Dirichlet boundary condition at $0$. By Lemma~\ref{lem: existence L2 to the left} (and its proof) the spectrum of $L^\limminus$ is almost surely discrete. By Gilbert--Pearson subordinacy theory, if $E \notin \sigma(L^\limminus)$ the solution of $L_0\psi= E\psi$ on $\R_\limminus$ with $\psi(0)=0, \psi'(0)=1$ is not subordinate at $-\infty$. Consequently, it holds that $\Sigma^\perp\subseteq \sigma(L^\limminus)$. (As an aside, we mention that at this point we could already conclude that $\sigma_{sc}(L_0|_{H^\perp})=\emptyset$, since an essential support of a singular continuous measure cannot be contained in the discrete set $\sigma(L^\limminus)$.)

  Clearly, the spectrum of $L^\limminus$ depends only on $\{g_n(\omega)\}_{n<0}$. Fix a typical realization of $\{g_n(\omega)\}_{n<0}$. Since the coupling constants are independent, this does not alter the probabilistic properties of $\{g_n(\omega)\}_{n\geq 1}$. We apply Theorem~\ref{thm: nonsub. asymptotics random model} and Lemma~\ref{lem: L2 norm comparability} for a fixed $E\in \sigma(L^\limminus)$ and infer that the extension to $\R$ of the eigenfunction of $L^\limminus$ corresponding to eigenvalue $E$ fails to be subordinate at $+\infty$ almost surely depending only on $\{g_n(\omega)\}_{n\geq 1}$. Since $\sigma(L^\limminus)$ is a countable set and the intersection of countably many almost sure events is almost sure, this property holds simultaneously for all $E\in \sigma(L^\limminus)$ almost surely depending only on $\{g_n(\omega)\}_{n\geq 1}$. We conclude that almost surely $\Sigma^\perp=\emptyset$, implying that the spectral measure of $L_0|_{H^\perp}$ is trivial and we have arrive at the desired contradiction.
\end{proof}

\begin{proof}[Proof of Theorem~\ref{thm: full-line random}]
  The claim that $\sigma_{ess}(L^\omega_{F,\lambda}) = \R$ almost surely is the content of Proposition~\ref{prop: essential spectrum}. By translation we may without loss of generality assume that the distribution of $g_0$ is absolutely continuous with respect to Lebesgue measure. The idea is to apply Theorem~\ref{thm: nonsub. asymptotics random model} and Corollary~\ref{cor: subordinate solution random} in conjunction with Lemmas~\ref{lem: spec averaging pure point} and~\ref{lem: spec averaging only sc}. Since none of our statements depend on the precise value of $E$ we shall in what follows apply the spectral averaging results for all of the spectrum simultaneously, i.e.\ when applying Lemmas~\ref{lem: spec averaging pure point} and~\ref{lem: spec averaging only sc} we always set $I = \R$. Note that the almost sure validity of (3) in Lemmas~\ref{lem: spec averaging pure point} and~\ref{lem: spec averaging only sc} is the content of Lemma~\ref{lem: cyclicity}.

  If $F <\lambda^2/2$ then Lemma~\ref{lem: existence L2 to the left}, Corollary~\ref{cor: subordinate solution random}, and Lemma~\ref{lem: L2 norm comparability} imply that almost surely the restrictions of $L^\omega_{F, \lambda}$ to the positive and negative half-lines satisfy the assumptions of Lemma~\ref{lem: spec averaging pure point}. Consequently almost surely and for Lebesgue almost every $g_0$ the spectrum of $L^\omega_{F, \lambda}$ is pure point. Since the distribution of $g_0$ was assumed absolutely continuous with respect to Lebesgue measure, this implies the almost sure statement in the theorem and completes the proof in the case $F<\lambda^2/2$. 

  If $F >\lambda^2/2$ then Lemma~\ref{lem: existence L2 to the left}, Corollary~\ref{cor: subordinate solution random}, and Lemma~\ref{lem: L2 norm comparability} imply that almost surely the restrictions of $L_{F, \lambda}^\omega$ to the positive and negative half-line satisfy the first and second assumptions of Lemma~\ref{lem: spec averaging only sc}.  
  Consequently, almost surely and for Lebesgue almost every $g_0$ the operator $L_{F,\lambda}^\omega$ has only singular continuous spectrum. Again by the absolute continuity of the distribution of $g_0$ this implies the claim of the theorem and therefore completes the proof.
\end{proof}

As mentioned in the introduction the refinements of Gilbert--Pearson subordinacy theory developed in~\cite{JitomirskayaLast_Acta_99,JitomirskayaLast_CMP_00,Damanik_etal_CMP_00} imply almost sure continuity properties of the spectral measure in the case of singular continuous spectrum. 

\begin{proposition}\label{prop: Hausdorff dimension}
  Fix $F, \lambda >0$ satisfying $F > \frac{\lambda^2}{2}$. If $\{g_n(\omega)\}_{n\in \Z}$ satisfy the assumptions of Theorem~\ref{thm: full-line random} then almost surely the spectral measure of $L^\omega_{F, \lambda}$ vanishes on sets of Hausdorff dimension less than $1-\frac{\lambda^2}{2F}$. 
\end{proposition}

Here spectral measure refers to the measure associated with the Titchmarsh--Weyl $m$-function of $L^\omega_{F,\lambda}$ through its integral representation as a Nevanlinna function. In view of Lemma~\ref{lem: cyclicity} this measure is almost surely the spectral measure of the pair $(L^\omega_{F,\lambda}, \delta)$.

\begin{proof}
  For almost every $E\in \R$, Theorem~\ref{thm: nonsub. asymptotics random model}, Corollary~\ref{cor: subordinate solution random}, and Lemma~\ref{lem: L2 norm comparability} imply that there exists two linearly independent $\eta$ and $\psi$ solving~\eqref{eq: generalized eigenequation} such that
  \begin{equation*}
    \|\eta\|^2_{L^2(0, L)} = L^{1/2-\lambda^2/(4F) +o(1)} \quad \mbox{and}\quad  \|\psi\|^2_{L^2(0, L)}= L^{1/2+\lambda^2/(4F) +o(1)} \quad \mbox{as }L\to \infty\,.
  \end{equation*}
  By Lemma~\ref{lem: cyclicity} and rank-one perturbation theory one concludes that these asymptotics remain true almost surely for almost every $E$ with respect to the spectral measure of $L^\omega_{F,\lambda}$. The proof is completed by applying the continuum analogue of~\cite[Theorem~1]{Damanik_etal_CMP_00} (see their Remark~2).
\end{proof}

\section{Refined bounds for exponential sums}
\label{sec: Refined exponential bounds}

To perform the corresponding analysis for the deterministic model, we shall need to replace the use of Martingale theory with a precise understanding of the exponential sums appearing in the argument. In particular, much of the difficulty in the analysis to follow will arise from the term linear in $U(j)$ which could almost immediately be discarded in the random case as a consequence of $\E[g_n]=0$. To understand this sum we follow arguments of Perelman~\cite{perelman_AsympAnal2005}. However, at this stage of the argument Perelman has oscillatory integrals in place of our exponential sums. So while the overall structure of the analysis is similar, there are different technical difficulties needed to be dealt with.

In the same spirit as stationary points of the phase function play an important role in the study of oscillatory integrals, the regions where the derivative of the phase is close to $2\pi\Z$ play a distinguished role in the corresponding exponential sums. In our case, $n\in \Z$ for which $\min_{\nu \in \Z}|\gamma'(n) - \pi \nu|$ is small will be of particular importance. As we shall see, most of the contribution to our sums come from such regions. Indeed, in the sum over a range of $n$ where $\gamma'(n)$ is far from an integer multiple of $\pi$ the cancellation effects are strong, and the contribution to the asymptotic behavior of the full sum comparatively small.

For $l$ sufficiently large define $X_l$ as the unique solution to the equation
\begin{equation}\label{eq: def Xl}
   \gamma'(X_l)= \pi l\,.
 \end{equation} 
Define also $x_l$ by
\begin{equation}\label{eq: def xl}
  x_l = \Bigl\lceil \frac{\pi^2}{F}\Bigl(l-\frac{1}{2}\Bigr)^2\Bigr\rceil -\frac{1}{2}\,.
\end{equation}
Note that the step functions appearing in our construction are defined in such a way that their value at $x_l$ coincides with that at $\frac{\pi^2}{F}(l-\frac{1}{2})^2$. For technical reasons it is somewhat convenient to ensure that the $x_l$ are half-integers.

By the asymptotics of $\gamma'$,
\begin{equation}\label{eq: asymptotics Xl, xl}
\begin{aligned}
  X_l &= \frac{\pi^2}{F}l^2+ O(1)\,, \quad \mbox{and}\quad 
  \gamma'(x_l) &= \pi l- \frac{\pi}{2} + O(l^{-1})\,,
\end{aligned}
\end{equation}
where the implicit constants can be taken uniform for $F, E$ in compact subsets of their respective domains. The importance of the $X_l$, which we shall refer to as resonant points, can be motivated in several manners. Above we arrived at the relevance of these points by their connection to regions with small cancellations in our exponential sums. However, they can be argued to play an important role for solutions of our generalized eigenvalue equation directly via what is known as the tilted band picture. Namely, one can write the eigenvalue equation as
\begin{equation*}
  H^{\scalebox{0.6}{KP}}_\lambda\psi(x) = (E+Fx)\psi(x)\,,
\end{equation*}
where $H^{\scalebox{0.6}{KP}}_\lambda$ is the Kronig--Penney operator, and, in the spirit of the adiabatic theorem, think of $E+Fx$ as an effective energy. In a region around $X_l$ this effective energy falls into the $l$-th spectral gap of $H^{\scalebox{0.6}{KP}}_\lambda$.

While the definition of the points $X_l$ is uniquely determined by the phase, the precise definition of the $x_l$'s is not as important. Indeed, any choice of points $x'_l$ interlacing the $X_l$'s for which $\min_{\nu\in \Z}|\gamma'(x'_l)-\pi \nu|$ is uniformly bounded away from zero would likely work just as well. However, the choice of $x_l$'s above is in a sense a natural one and their particular structure will lead to certain simplifications in formulas to come.

In what follows our aim is to understand to high precision the exponential sums analogous to those appearing in our treatment of the random model when the summation index $n$ ranges from $x_l$ to $x_{l+1}$ or some subset thereof. When we return in Section~\ref{sec: coarse-graining} to studying the solutions of~\eqref{eq: generalized eigenequation}, understanding these partial sums will allow us to describe how the Pr\"ufer coordinates $R, \eta$ change when we move from $x_l$ to $x_{l+1}$, i.e.\ when we transition over the resonant point $X_l$. 

As in the random model the complex phase in the exponential sums that appear is a linear combination of $\gamma$ and $\eta$. Our aim is to treat $\eta$ perturbatively, however, this is only possible on scales much smaller than that of the $X_l$. This issue will be circumvented by extracting the leading behavior of $\eta$ and absorbing it in the explicit phase $\gamma$. For this reason, the following results concern exponential sums where the main phase $\gamma$ has been perturbed by a function $h$ which is only assumed to vary much slower. When we later apply these bounds $h$ will be given by linear combinations of $\eta(n)$ and $\sqrt{n}$.

The first bound we prove is in the same spirit as Theorem~\ref{thm: crude bound expsum with decay} but capturing the strong cancellations present in regions away from the resonant points $X_l$. 

\begin{theorem}\label{thm: expsum away from Xl}
  Fix $F>0, E\in \R, \alpha \geq 0, \beta >0, \sigma \in [1/2, 1]$, and $h\colon \N \to \R$ satisfying
  \begin{equation*}
    |h(n+1)-h(n)|\leq C n^{-\beta} \quad \mbox{for all } n \in [X_l, X_{l+1}]\,.
  \end{equation*}
  Then, for $l$ sufficiently large and all $X_l+Cl^\sigma \leq a < b \leq X_{l+1}-Cl^\sigma$,
  \begin{align*}
    \biggl|\sum_{a<n\leq b} \frac{e^{i(2 \gamma(n)+h(n))}}{\gamma'(n)^\alpha} \biggr|&\lesssim l^{-\alpha-\sigma +1}(1+l^{1-2\beta})\,.
  \end{align*}
  The implicit constants are uniformly bounded for $E, F, \alpha, \beta, \sigma, C$ in compact subsets of their respective domains.
\end{theorem}

The above bound captures the behavior of our exponential sums to a sufficient precision away from the resonant points $X_l$. However, the main contribution to the sums we are interested in comes from small neighbourhoods of $X_l$ and to carry out our analysis we need to understand this contribution in greater detail. This is accomplished by the next two theorems. 

The first theorem provides an order-sharp bound for the sum over the range $(x_l, x_{l+1}]$. 
\begin{theorem}\label{thm: rough expsum decay Xl interval}
  Fix $F>0, E\in \R, \alpha \geq 0, \beta >0, \mu>0$, and $h\colon \N \to \R$ satisfying
  \begin{equation*}
    |h(n+1)-h(n)|\leq C n^{-\beta} \quad \mbox{for all } n \in [x_l, x_{l+1}]\,.
  \end{equation*}
  Then, for $l$ sufficiently large and all $x_l\leq a<b\leq x_{l+1}$,
  \begin{equation*}
    \biggl|\sum_{a<n\leq b} \frac{e^{i(\mu \gamma(n)+h(n))}}{\gamma'(n)^\alpha} \biggr|\lesssim l^{-\alpha+1/2}(1+l^{1-2\beta})\,.
  \end{equation*}
  The implicit constant is uniformly bounded for $E, F, \alpha, \beta, \mu, C$ in compact subsets of their respective domains.
\end{theorem}

While the previous theorem is often good enough for our purposes when it comes to estimating error terms, we shall need more precise knowledge to analyse the leading order terms of our equations.
\begin{theorem}\label{thm: Precise asymptotics exp sum}
  Fix $F>0$ and $E\in \R$. Let $h\in C^7(0, \infty)$ satisfy $|\partial_x^j h(x)|\leq C x^{1/2-j}$ for $j=0, \ldots, 7$. Then
  \begin{align*}
    \sum_{x_l <n \leq x_{l+1}}\frac{e^{i(2\gamma(n)+h(n))}}{\gamma'(n)}  
    &=
    \Bigl(\frac{2}{Fl}\Bigr)^{1/2}e^{2i\Gamma_h(l)} - \frac{e^{i(2\gamma(x_{l+1})+h(x_{l+1})-\pi/2-\pi (l+1))}}{2\pi (l+1)}\\
    &\quad 
    + \frac{e^{i(2\gamma(x_{l})+h(x_{l})-\pi/2-\pi l)}}{2\pi l} + O(l^{-3/2})\,,
  \end{align*}
  with
  \begin{equation}\label{eq: def effective phase Gammal}
    \Gamma_h(l)=-\frac{\pi^3 l^3}{3F} + \frac{\pi E}{F}l+ \frac{5\pi}{8}+\frac{1}{2}h\Bigl(\frac{\pi^2}{F}l^2\Bigr)\,.
  \end{equation}
  Moreover, the constant in the error term can be chosen uniform for $F, E, C$ in compact subsets of their respective domains.
\end{theorem}

The proofs of Theorems~\ref{thm: expsum away from Xl},~\ref{thm: rough expsum decay Xl interval}, and~\ref{thm: Precise asymptotics exp sum} occupy what remains of this section. Since the proofs get rather technical and do not reveal anything crucial to our main results, we advise the reader who is more interested in the spectral theory to skip ahead to Section~\ref{sec: coarse-graining} where our focus returns to understanding the asymptotic behavior of generalized eigenfunctions.

\subsection{Proof of Theorems~\ref{thm: expsum away from Xl} and~\ref{thm: rough expsum decay Xl interval}}

The proofs of Theorems~\ref{thm: expsum away from Xl} and~\ref{thm: rough expsum decay Xl interval} will follow the same strategy as that of Theorem~\ref{thm: crude bound expsum with decay}. Indeed the bounds in the theorems follow along the same line of reasoning. When we are close to the resonant points $X_l$ we get an order-sharp bound by applying Corollary~\ref{cor: Rough exp sum bound}. However, if we are far from the resonant points in order to capture more of the cancellations we instead apply the following classical inequality due to Kuzmin and Landau (see for instance~\cite[Lemma 4.19]{Titchmarsh_RiemanZeta86}). 
\begin{lemma}\label{lem: Kuzmin-Landau}
  Let $f\in C^1([a, b])$ be a real function that is convex or concave, and let $\min_{\nu \in \Z}|f'(x)-\nu|\geq \kappa>0$ for all $x\in [a, b]$. Then
  \begin{equation*}
    \biggl|\sum_{a<n\leq b}e^{2\pi i f(n)}\biggr| \lesssim \kappa^{-1}\,.
  \end{equation*}
\end{lemma}

\begin{proof}[{Proof of Theorems~\ref{thm: expsum away from Xl} and~\ref{thm: rough expsum decay Xl interval}}]
  Let $a$ and $b$ be the left resp.\ right endpoints of one of the intervals of summation in the theorems. By a summation by parts
  \begin{equation}\label{eq: sum by parts expsum decay proof}
  \begin{aligned}
    \sum_{a<n \leq b}\frac{e^{i(\mu \gamma(n)+h(n))} }{\gamma'(n)^\alpha} 
    &=
    \sum_{a<n \leq b-1}\Bigl[\frac{e^{ih(n)}}{\gamma'(n)^\alpha} -\frac{e^{ih(n+1)}}{\gamma'(n+1)^\alpha}\Bigr]\sum_{a<j\leq n}e^{i\mu \gamma(j)}\\
    &\quad 
    + \frac{e^{ih(\lfloor b\rfloor)}}{\gamma'(\lfloor b\rfloor)^\alpha}\sum_{a<n\leq b}e^{i\mu \gamma(n)} \,.
  \end{aligned}
  \end{equation}
  By Lemma~\ref{lem: gamma asymptotics}, for $l$ sufficiently large, $\gamma'(n)^{-\alpha}\lesssim l^{-\alpha}$ for $a<n\leq b$. Similarly, by Lemma~\ref{lem: gamma asymptotics} and using the fundamental theorem of calculus as in the proof of Theorem~\ref{cor: Rough exp sum bound}, we can estimate
  \begin{align*}
    \Bigl|\frac{e^{ih(n)}}{\gamma'(n)^\alpha} -\frac{e^{ih(n+1)}}{\gamma'(n+1)^\alpha}\Bigr|
    &\leq
    \Bigl|\frac{1}{\gamma'(n)^\alpha} -\frac{1}{\gamma'(n+1)^\alpha} \Bigr| \\
    &\quad + \frac{1}{\gamma'(n+1)^\alpha}|e^{ih(n)}-e^{ih(n+1)}| \\
    &\lesssim
    l^{-\alpha-2} + l^{-\alpha}|h(n)-h(n+1)|\\
    &\lesssim
    l^{-\alpha}(l^{-2}+l^{-2\beta})\,.
  \end{align*}
  Therefore from~\eqref{eq: sum by parts expsum decay proof} we have the bound
  \begin{equation}\label{eq:decay expsum proof bound 2}
  \begin{aligned}
    \biggl|\sum_{a<n \leq b}\frac{e^{i(\mu \gamma(n)+h(n))}}{\gamma'(n)^\alpha}  \biggr|
    &\lesssim
     l^{-\alpha}(l^{-2}+l^{-2\beta})\sum_{a<n \leq b-1}\biggl|\sum_{a<j\leq n}e^{i\mu \gamma(j)}\biggr|\\
    &\quad 
    + l^{-\alpha}\biggl|\sum_{a<n\leq b}e^{i\mu \gamma(n)}\biggr| \,.
  \end{aligned}
  \end{equation}
  At this point the proofs of Theorems~\ref{thm: rough expsum decay Xl interval} and~\ref{thm: expsum away from Xl} diverge.

  \medskip

  In the case of Theorem~\ref{thm: rough expsum decay Xl interval} we apply Corollary~\ref{cor: Rough exp sum bound} to bound the exponential sums on the right-hand side of~\eqref{eq:decay expsum proof bound 2}. Combined with the fact that $|b-a|\leq |x_{l+1}-x_l| \lesssim l$ this completes the proof of Theorem~\ref{thm: rough expsum decay Xl interval}. 

\medskip

  To prove Theorem~\ref{thm: expsum away from Xl} we instead argue as follows. Recall that in this theorem $\mu =2$. By the  definition of the $x_l$'s and the asymptotics of $\gamma'$ we have for any $y\in (x_l, x_{l+1}]$
  \begin{equation*}
    \min_{\nu \in \Z}\Bigl|\frac{\gamma'(y)}{\pi}-\nu\Bigr|  \geq \frac{1}{2}\Bigl|\frac{\gamma'(y)}{\pi}-l\Bigr|\,.
  \end{equation*}
  Moreover, for any $y\in (x_l, x_{l+1}]$ a Taylor expansion around $X_l$ implies that $$\frac{\gamma'(y)}{\pi}-l=\frac{1}{\pi}\gamma''(\xi)(y-X_l)$$ for some $\xi \in (x_l, x_{l+1}]$. Since $\gamma''(x)\sim x^{-1/2}$ for $l$ large enough we have
  \begin{equation}\label{eq: gamma' away from integers}
    \min_{\nu \in \Z}\Bigl|\frac{\gamma'(y)}{\pi}-\nu\Bigr| \gtrsim \frac{|y-X_l|}{l}\,.
  \end{equation}
  By Lemma~\ref{lem: gamma asymptotics} $\gamma$ is convex for $l$ large enough. Therefore, Lemma~\ref{lem: Kuzmin-Landau} and~\eqref{eq: gamma' away from integers} imply that, for $X_l+Cl^\sigma\leq x<y\leq X_{l+1}-Cl^\sigma$,
  \begin{equation}\label{eq: Kuzmin landau away from Xl}
    \biggl|\sum_{x<n \leq y}e^{i2 \gamma(n)}\biggr| \lesssim l^{1-\sigma}\,.
  \end{equation}
  The proof of Theorem~\ref{thm: expsum away from Xl} is completed by using~\eqref{eq: Kuzmin landau away from Xl} to bound the exponential sums on the right-hand side of~\eqref{eq:decay expsum proof bound 2}.
\end{proof}

\subsection{Proof of Theorem~\ref{thm: Precise asymptotics exp sum}}

Our proof of Theorem~\ref{thm: Precise asymptotics exp sum} will combine two classical results: the Poisson summation formula and the method of stationary phase. The particular form of the Poisson summation formula that we shall use is the following which can be found in~\cite[Theorem 45]{Titchmarsh_FourierIntegrals}.

\begin{lemma}\label{lem: Poisson summation}
  For $f \in \textup{BV}(\R) \cap L^1(\R)$ it holds that
  \begin{equation*}
    \sum_{n\in \Z} \tilde f(n) 
    = \hat f(0) + \sum_{\nu = 1}^\infty \bigl(\hat f(\nu)+\hat f(-\nu)\bigr)
  \end{equation*}
  where
  \begin{equation*}
    \tilde f(x) = \lim_{\eps\to 0} \frac{f(x+\eps)+f(x-\eps)}{2} \quad \mbox{and} \quad \hat f(\xi) = \int_\R f(x) e^{-2\pi i \xi x}\,dx\,.
  \end{equation*}
\end{lemma}
\begin{remark}
  It should be emphasized that while the series in the right-hand side is absolutely convergent this is generally not the case for $\sum_{\nu \in \Z} \hat f(\nu)$. Indeed, to see this it suffices to consider $f= \1_{[a, b]}$ for suitably chosen $a, b \in \R$.
\end{remark}

\begin{proof}[Proof of Theorem~\ref{thm: Precise asymptotics exp sum}]
We wish to apply the Poisson summation formula in Lemma~\ref{lem: Poisson summation} with
\begin{equation*}
  f(x)=\frac{e^{i(2\gamma(x)+h(x))}}{\gamma'(x)}\1_{(x_{l}, x_{l+1}]}(x)\,.
\end{equation*}
Since this function has compact support and is smooth away from the jumps at $x_l, x_{l+1}$ it satisfies the assumptions of Lemma~\ref{lem: Poisson summation}. By construction $x_l, x_{l+1} \notin \Z$ so the function $f$ is continuous at the integers. Therefore, the lemma implies
\begin{equation}\label{eq: Poisson sum rewrite}
\begin{aligned}
  \hspace{-12pt}\sum_{x_l<n \leq x_{l+1}}\!\!\!\!\frac{e^{i(2\gamma(n)+h(n))}}{\gamma'(n)} &= 
  \int_{x_l}^{x_{l+1}} \!\frac{e^{i(2\gamma(x)+h(x))}}{\gamma'(x)}dx\\
  &  
  + \sum_{\nu = 1}^\infty \biggl[\int_{x_l}^{x_{l+1}}\! \frac{e^{i(2\gamma(x)+h(x)-2\pi \nu x)}}{\gamma'(x)}dx
  +\int_{x_l}^{x_{l+1}}\! \frac{e^{i(2\gamma(x)+h(x)+2\pi \nu x)}}{\gamma'(x)}dx\biggr].\hspace{-15pt}
\end{aligned}
\end{equation}
We wish to apply the method of stationary phase to find an asymptotic expansion of the integrals in the sum on the right-hand side.

By Lemma~\ref{lem: gamma asymptotics}, the assumptions on $h$, and a Taylor expansion around $X_l$ one finds that, for all $x\in (X_{l-1}, X_{l+1})$,
\begin{equation}\label{eq: phase taylor expansion}
  \begin{aligned}
    2\gamma'(x)+h'(x) 
    &= 
    2\pi l + h'(X_l) + (2\gamma''(X_l)+h''(X_l))(x-X_l) + O(l^{-1})\\
    &=
    2\pi l + 2\gamma''(X_l)(x- X_l) + O(l^{-1})\\
    &= 2\pi l - \sqrt{F}\Bigl(\frac{\pi^2}{F}l^2+O(1)\Bigr)^{-1/2}(x- X_l) + O(l^{-1})\\
    &= 2\pi l - F\frac{x-X_l}{\pi l} + O(l^{-1})\,.
  \end{aligned}
\end{equation}
Consequently, the only $\nu \in \Z$ for which the phase $2\gamma(x)+h(x)-2\pi \nu x$ has a stationary point in $[x_l, x_{l+1}]$ is $\nu = l$. Furthermore, by asymptotic convexity of the phase when $\nu=l$, and $l$ is sufficiently large, there exists a unique stationary point in $(x_l, x_{l+1})$, we denote it by $\tilde X_l$. By the Taylor expansion~\eqref{eq: phase taylor expansion} it holds that
\begin{equation*}
  |X_l-\tilde X_l| \lesssim 1\,.
\end{equation*} 

To apply the method of stationary phase we first rescale the integrals as follows. Let $x(y) = x_l + (x_{l+1}-x_l)y$ and define
\begin{align*}
   u_l(y) &= \frac{l}{\gamma'(x(y))}\,,\\
   \phi_{l,\nu}(y) &= 2\gamma(x(y))+h(x(y))-2\pi \nu x(y)\,,\\
   \omega_{l,\nu} &= \|\phi'_{l,\nu}\|_{L^\infty(0, 1)}\,,\\
   \Phi_{l,\nu}(y) &= \frac{\phi_{l,\nu}(y)-\phi_{l,\nu}(0)}{\omega_{l,\nu}}\,.
\end{align*} 
Note that the normalization has been chosen so that
\begin{align*}
  \Phi_{l,\nu}(0)=0\quad \mbox{for all } l, \nu\,.
\end{align*}
By a change of variables we find
\begin{align*}
  \int_{x_l}^{x_{l+1}}\frac{e^{i(2\gamma(x)+h(x)-2\pi\nu x)}}{\gamma'(x)}\,dx &= (x_{l+1}-x_l)l^{-1}e^{i \phi_{l, \nu}(0)}\int_0^1 u_l(y)e^{i\omega_{l,\nu}\Phi_{l,\nu}(y)}\,dy\,.
\end{align*}
Note that
\begin{align*}
   \omega_{l,\nu} &\gtrsim l|l-\nu + \delta_{l,\nu}|\gg 1\,,
 \end{align*} 
 where $\delta_{l,\nu}$ denotes the Kronecker delta.

By Lemma~\ref{lem: gamma asymptotics} it readily follows that
\begin{equation*}
  \|\partial_y^k u_l\|_{L^\infty(0, 1)} \lesssim l^{-k} \quad \mbox{for } k=0, \ldots, 6\,.
\end{equation*}
Similarly, for $k=2, \ldots, 7$,
\begin{equation*}
  \|\partial_y^k \Phi_{l,\nu}\|_{L^\infty(0, 1)}\lesssim \omega_{l,\nu}^{-1}l^{3-k} \lesssim |l-\nu+\delta_{l,\nu}|^{-1}l^{2-k}\,,
\end{equation*}
where the constants are independent of $l, \nu$. The definition of $\Phi_{l,\nu}$ and~\eqref{eq: phase taylor expansion} ensure that for $l$ sufficiently large, $\nu \neq l$, and all $y \in [0, 1]$
\begin{equation*}
  |\Phi_{l,\nu}'(y)| \gtrsim 1  \,,
\end{equation*}
in particular these phase functions have no stationary points in $[0, 1]$.
For $\Phi_{l,l}$ we have a unique stationary point at $\tilde y_l=\frac{\tilde X_l-x_l}{x_{l+1}-x_l} = \frac{1}{2} + O(l^{-1})$. Moreover, for all $y\in [0, 1]$, by a Taylor expansion
\begin{equation*}
    \frac{|y-\tilde y_l|}{|\Phi'_{l,l}(y)|} 
    = 
    \frac{|y-\tilde y_l|}{|\Phi_{l,l}''(\tilde y_l)(y-\tilde y_l) + \Phi_{l,l}'''(z_l)(y-\tilde y_l)^2|}
    = 
    \frac{1}{|\Phi_{l,l}''(\tilde y_l) + \Phi_{l,l}'''(z_l)(y-\tilde y_l)|}
    = \frac{1}{2} + O(l^{-1})\,,
\end{equation*}
where we use $\Phi_{l,l}''(\tilde y_l)=2+O(l^{-1})$, $|y-\tilde y_l|\lesssim 1$, and $\|\Phi'''_{l,l}\|_{L^\infty}\lesssim l^{-1}$. Indeed, by Lemma~\ref{lem: gamma asymptotics}
\begin{equation*}
  \omega_{l,l} 
   \! =\! (x_{l+1}\!-\!x_l)\max\{|2\gamma'(x_{l+1})+h'(x_{l+1})-2\pi l|,\, |2\gamma'(x_{l})+h'(x_{l})-2\pi l|\}
   \!= \!\frac{2\pi^3l}{F}+O(1)\,,
\end{equation*}
and
\begin{equation*}
    \Phi_{l,l}''(\tilde y_l) 
    =
    \omega_{l,l}^{-1}(x_{l+1}-x_l)^2(2\gamma''(\tilde X_l)+h''(\tilde X_l))
    =
    2 + O(l^{-1})\,.
\end{equation*}
We conclude that the phase functions $\Phi_{l,\nu}$ all belong to a bounded subset of $C^7(0, 1)$ and the stationary point of $\Phi_{l,l}$ is uniformly non-degenerate in $l$. As such we can apply Lemmas~\ref{lem: Principle of non-stationary phase} and~\ref{lem: Principle of stationary phase} to compute the integrals with uniform error estimates.

By Lemma~\ref{lem: Principle of stationary phase} when $\nu=l$ and Lemma~\ref{lem: Principle of non-stationary phase} when $\nu \neq l$, we conclude that for $l$ sufficiently large
\begin{align*}
  \int_0^1 u_l(y)
  e^{i\omega_{l,\nu} \Phi_{l,\nu}(y)}\,dy
  &= 
  \delta_{l,\nu}\frac{(2\pi)^{1/2}e^{i(\omega_{l,l}\Phi_{l,l}(\tilde y_l)+\pi/4)}}{\omega_{l,l}^{1/2} \Phi''_{l,l}(\tilde y_l)^{1/2}}\sum_{j=0}^1 l^{-j}\mathcal{L}_ju_l(\tilde y_l)\\
  &\quad
  +
  \frac{iu_l(0)e^{i\omega_{l,\nu}\Phi_{l,\nu}(0)}}{\omega_{l,\nu}\Phi_{l,\nu}'(0)} 
   - \frac{iu_l(1)e^{i\omega_{l,\nu}\Phi_{l,\nu}(1)}}{\omega_{l,\nu}\Phi_{l,\nu}'(1)}
  +O(\omega_{l,\nu}^{-2})\,,
\end{align*}
where $\mathcal{L}_j$ are as in the lemma and the implicit constant is independent of $l, \nu$. In the particular cases $j=0, 1$,
\begin{align*}
    \mathcal{L}_0v(\tilde y_l)&= v(\tilde y_l)\,,\\
    \mathcal{L}_1v(\tilde y_l) 
    &=
    -i\Bigl(\frac{v''(\tilde y_l)}{2\Phi_{l,l}''(\tilde y_l)}-\frac{4v'(\tilde y_l)\Phi_{l,l}'''(\tilde y_l)+v(\tilde y_l)\Phi_{l,l}''''(\tilde y_l)}{8\Phi_{l,l}''(\tilde y_l)^{2}}+\frac{5v(\tilde y_l)\Phi_{l,l}'''(\tilde y_l)^2}{24\Phi_{l,l}''(\tilde y_l)^{3}}\Bigr)\,.
  \end{align*}
By the definitions of $u_l, \Phi_{l,\nu}$, Lemma~\ref{lem: gamma asymptotics} combined with the assumptions on $h$ implies that
  \begin{equation*}
    \mathcal{L}_0 u_l(\tilde y_l)= \frac{l}{\gamma'(\tilde X_l)} = \frac{1}{\pi}+ O(l^{-2}) \quad \mbox{and}\quad \mathcal{L}_1u_l(\tilde y_l) = O(l^{-2})\,.
  \end{equation*}

By inserting the above into~\eqref{eq: Poisson sum rewrite} and using the definition of $\omega_{l,\nu}$ and $\Phi_{l,\nu}$ we have arrived at
\begin{align*}
  &\sum_{x_l<n \leq x_{l+1}}\!\!\!\frac{e^{i(2\gamma(n)+h(n))}}{\gamma'(n)}\\
   &\quad =
    \Bigl(\frac{2}{\pi l^2 (2\gamma''(\tilde X_l) + h''(\tilde X_l))}\Bigr)^{1/2}e^{i(2\gamma(\tilde X_l)+h(\tilde X_l)-2\pi l \tilde X_l+\pi/4)}\\
     &\quad \quad 
   + \frac{ie^{i(2\gamma(x_l)+h(x_l))}}{\gamma'(x_l)}\frac{1}{2\gamma'(x_{l})+h'(x_{l})}
   - \frac{ie^{i(2\gamma(x_{l+1})+h(x_{l+1}))}}{\gamma'(x_{l+1})}
   \frac{1}{2\gamma'(x_{l+1})+h'(x_{l+1})}\\
   &\quad \quad 
   + \frac{ie^{i(2\gamma(x_l)+h(x_l))}}{\gamma'(x_l)}\sum_{\nu =1}^\infty \biggl[
   \frac{e^{-i2\pi \nu x_l}}{2\gamma'(x_{l})+h'(x_{l})-2\pi \nu}
   +\frac{e^{i2\pi \nu x_l}}{2\gamma'(x_{l})+h'(x_{l})+2\pi \nu}\biggr]\\
   &\quad \quad 
   - \frac{ie^{i(2\gamma(x_{l+1})+h(x_{l+1}))}}{\gamma'(x_{l+1})}
   \sum_{\nu =1}^\infty \biggl[
   \frac{e^{-i2\pi \nu x_{l+1}}}{2\gamma'(x_{l+1})+h'(x_{l+1})-2\pi \nu}
   +\frac{e^{i2\pi \nu x_{l+1}}}{2\gamma'(x_{l+1})+h'(x_{l+1})+2\pi \nu}\biggr]\\
   &\quad\quad  +O(l^{-2})\,,
\end{align*}
where we used $\sum_{\nu \neq l}|\nu-l|^{-2} \lesssim 1$. 

To simplify the boundary contributions we shall use the identity (see, for instance,~\cite[eq.~5.1.3.7]{Prudnikov_SeriesIntegrals})
\begin{equation*}
  \frac{1}{y} +\sum_{\nu = 1}^\infty(-1)^\nu \Bigl[\frac{1}{y-\nu}+ \frac{1}{y+\nu}\Bigr] = \frac{\pi}{\sin(\pi y)}\quad \mbox{for }y \notin \Z\,.
\end{equation*}
Recalling that the $x_l$'s are half-integers and using Lemma~\ref{lem: gamma asymptotics} to deduce
$$
2\gamma'(x_l)+h'(x_l)-2\pi \nu = 2\pi(l-\nu) - \pi +O(l^{-1}) \notin 2\pi \Z
$$ 
we conclude that
\begin{align*}
  \sum_{x_l<n\leq x_{l+1}}\frac{e^{i(2\gamma(n)+h(n))}}{\gamma'(n)} 
  &= 
  \Bigl(\frac{2}{\pi l^2 (2\gamma''(\tilde X_l) + h''(\tilde X_l))}\Bigr)^{1/2}e^{i(2\gamma(\tilde X_l)+h(\tilde X_l)-2\pi l \tilde X_l+\pi/4)}\\
   &
   \quad - \frac{e^{i(2\gamma(x_l)+h(x_l)-\pi/2)}}{2\gamma'(x_l)\sin(\gamma'(x_l)+h'(x_l)/2)}\\
   &
   \quad + \frac{e^{i(2\gamma(x_{l+1})+h(x_{l+1})-\pi/2)}}{2\gamma'(x_{l+1})\sin(\gamma'(x_{l+1})+h'(x_{l+1})/2)}
   +O(l^{-2})
\end{align*}
What remains is to simplify the terms by using Taylor expansion and Lemma~\ref{lem: gamma asymptotics}. 

By Lemma~\ref{lem: gamma asymptotics}, the assumptions on $h$, and since $\tilde X_l = \frac{\pi^2}{F}l^2+O(1)$ we find
\begin{equation*} 
  \frac{1}{2\gamma''(\tilde X_l) + h''(\tilde X_l)} = \frac{\pi l}{F} + O(l^{-1})\,.
\end{equation*}
Similarly, Lemma~\ref{lem: gamma asymptotics}, the assumptions on $h$, and $x_l = \frac{\pi^2}{F}(l-1/2)^2+O(1)$ imply
\begin{equation*}
  \frac{1}{\gamma'(x_l)\sin(\gamma'(x_l)+h'(x_l)/2)} = -\frac{(-1)^l}{\pi l} + O(l^{-2})\,.
\end{equation*}

Finally, for the effective phase in the main term the expansions $2\gamma'\bigl(\frac{\pi^2}{F}l^2\bigr)=2\pi l + O(l^{-1})$ and $\tilde X_l - \frac{\pi^2}{F}l^2 = O(1)$ imply that
  \begin{align*}
    2\gamma(\tilde X_l)+h(\tilde X_l)-2\pi l\tilde X_l
    &=
    2\gamma\Bigl(\frac{\pi^2}{F}l^2\Bigr)+h\Bigl(\frac{\pi^2}{F}l^2\Bigr)-\frac{2\pi^3}{F}l^3\\
    &\quad
    + \Bigl[2\gamma'\Bigl(\frac{\pi^2}{F}l^2\Bigr)+h'\Bigl(\frac{\pi^2}{F}l^2\Bigr)-2\pi l\Bigl]\Bigl(\tilde X_l- \frac{\pi^2}{F}l^2\Bigr) +O(l^{-1})\\
    &=
    -\frac{2\pi^3 l^3}{3F} + \frac{2\pi E}{F}l+ \pi +h\Bigl(\frac{\pi^2}{F}l^2\Bigr) +O(l^{-1})\,.
  \end{align*}
  Consequently,
  \begin{equation*}
    e^{i(2\gamma(\tilde X_l)+h(\tilde X_l)-2\pi l\tilde X_l+\pi/4)} = e^{2i\Gamma_h(l)} + O(l^{-1})\,,
  \end{equation*}
  with $\Gamma_h(l)$ as in the theorem. 
  Combining the last three estimates concludes the proof of Theorem~\ref{thm: Precise asymptotics exp sum}.
\end{proof}

\begin{remark}
Even though in the end the only non-negligible contribution comes from the $j=0$ term in the stationary phase expansion we need to compute the expansion to second order for the error term to be sufficiently precise. Indeed, directly applying Lemma~\ref{lem: Principle of stationary phase} with $k=1$ instead of $2$ would only yield the error estimate
  \begin{equation*}
    l^{-1}\|u_l\|_{C^2(0, 1)} \sim l^{-1}
  \end{equation*}
  which is not good enough. We also emphasize that with the exception of needing higher order terms in the asymptotic expansion for $\gamma$ there is in theory nothing that prevents us from using the above argument to obtain asymptotic expansions to much higher precision.
\end{remark}

\section{Coarse-graining the Pr\"ufer equations}
\label{sec: coarse-graining}

We return to the analysing the asymptotic behavior of solutions of the equation
\begin{equation}\label{eq: generalized eigenfunction deterministic} 
    \begin{cases}
      -\psi''(x) -Fx \psi(x) = E\psi(x) \quad &\mbox{in }\R \setminus \Z\,,\\
    J\psi(n) =0 &\mbox{for } n \in \Z\,,\\ 
    J\psi'(n)=\lambda \psi(n) \quad &\mbox{for } n \in \Z\,.
    \end{cases}
  \end{equation}
for fixed $F>0, E\in\R,$ and $\lambda \in \R$.

As in the case of the random model, the proof of our main results concerning the deterministic operator relies on understanding the Pr\"ufer equations of Lemma~\ref{lem: Prufer equations}. The main difference between the two models is that we can no longer use martingale techniques. In particular, our main challenge lies in understanding the behavior of
\begin{equation*}
  \sum_{n=1}^N U(n)\sin(2\theta(n)),
\end{equation*}
which could be shown to be sub-leading when the $U(n)$ were independent random variables with mean zero. As we shall see, this is not the case in our current setting.

In this section our main goal is to show that we can coarse-grain our system by first understanding the partial sums with $n\in (x_l, x_{l+1}]$. This is where the results of the previous section come into play. Before we can apply the expansion of Theorem~\ref{thm: Precise asymptotics exp sum} we need to remove the dependence of the phase on the Pr\"ufer angle $\eta$. However, $\eta$ varies too much over the interval $[x_l, x_{l+1}]$ to simply pull it out of the sum. To remedy this it is convenient to work with a modification of $\eta$ obtained by extracting its leading-order behavior. To that end we define
\begin{equation}\label{eq: slow angle definition}
  \tilde\eta(x)= \eta(x)+ \lambda\sqrt{\frac{\lceil x\rceil }{F}}\quad \mbox{and}\quad 
  \tilde\gamma(x)=\gamma(x)-\lambda\sqrt{\frac{x}{F}}\,,
\end{equation}
so that $\theta(n)=\gamma(n)+\eta(n)=\tilde\gamma(n)+\tilde\eta(n)$ for $n\in \Z$. Note that this correction of the phase $\gamma$ is sufficiently small to be covered by the assumptions in Theorems~\ref{thm: rough expsum decay Xl interval},~\ref{thm: expsum away from Xl}, and~\ref{thm: Precise asymptotics exp sum}. At certain points it would be more convenient to define $\tilde \eta$ without the $\lceil \cdot \rceil$ so that the equality $\theta(x)=\tilde \gamma(x)+\tilde \eta(x)$ is valid for all $x\in \R$. The current definition is chosen in such a way that the main term in our asymptotic formula for $\psi$ in Theorem~\ref{thm: l scale equations intro} satisfies the correct ODE between consecutive integers.

The main result of this section is to provide asymptotic equations for how the (modified) Pr\"ufer variables $R$ and $\tilde\eta$ evolve along the points $\{x_l\}_{l\in \N}$, that is when we pass from one of the regions where we have strong cancellations to the next and in the process pick up the contribution from the region close to the resonant point $X_l$. The precise statement we shall prove is the content of Theorem~\ref{thm: equation R, eta xl to xl+1}. 

\begin{theorem}\label{thm: equation R, eta xl to xl+1}
Fix $F>0, E \in \R,$ and $\lambda \in \R$. Then the (modified) Pr\"ufer coordinates $R, \tilde \eta$ corresponding to a real-valued solution of~\eqref{eq: generalized eigenfunction deterministic} satisfy
    \begin{align*}
        \log\Bigl(\frac{R(x_{l+1})}{R(x_{l})}\Bigr) 
        &=
        \frac{\lambda}{\sqrt{2Fl}}\sin(2\Gamma(l)+2\tilde\eta(x_l))
        + \frac{\lambda^2}{4F l}\bigl(1+\cos(4\Gamma(l)+4\tilde\eta(x_l))\bigr)\\
        &\quad 
        + \frac{(-1)^{l+1}\lambda }{4\pi (l+1)}\cos(2\tilde\gamma(x_{l+1})+2\tilde\eta(x_{l+1})) \\
        &\quad
        - \frac{(-1)^l\lambda}{4\pi l}\cos(2\tilde\gamma(x_l)+2\tilde\eta(x_{l}))
         + O(l^{-5/4})\,.
    \end{align*}
  and
    \begin{align*}
        \tilde\eta(x_{l+1})-\tilde\eta(x_l) &=
        \frac{\lambda}{\sqrt{2F l}}\cos(2\Gamma(l)+2\tilde\eta(x_l))- \frac{\lambda^2}{4Fl}\sin(4\Gamma(l)+4\tilde\eta(x_l))\\
        &\quad 
        + \frac{\lambda^2}{4}\Im\biggl[\,\sum_{x_l<n\leq x_{l+1}}\sum_{n<j\leq x_{l+1}}\frac{e^{2i(\tilde\gamma(n)-\tilde\gamma(j))}}{\gamma'(n)\gamma'(j)}\biggr]\\
        &\quad 
        - \frac{(-1)^{l+1}\lambda }{4\pi (l+1)} \sin(2\tilde\gamma(x_{l+1})+2\tilde\eta(x_{l+1}))\\
        &\quad
        + \frac{(-1)^l\lambda}{4\pi l} \sin(2\tilde\gamma(x_l)+2\tilde\eta(x_{l})) 
        + O(l^{-5/4})
    \end{align*}
  where
  \begin{equation*}
    \Gamma(l) = - \frac{\pi^3l^3}{3F}+ \frac{\pi l}{F}(E-\lambda) + \frac{5\pi}{8}\,.
  \end{equation*}
\end{theorem}

\begin{remark} A couple of remarks:
\begin{enumerate}
  \item We emphasize that Theorem~\ref{thm: equation R, eta xl to xl+1} holds without the assumption $F \in \pi^2 \Q_\limplus$. In fact, this will be the case with all results proved in the current section. The rationality assumption enters only once we wish to understand the behavior on even larger scales than that determined by the resonant points (see Section~\ref{sec: asymptotics R deterministic}).

  \item We note that $\Gamma$ in Theorem~\ref{thm: equation R, eta xl to xl+1} is $\Gamma_h$ from Theorem~\ref{thm: Precise asymptotics exp sum} with $h(x)=-2\lambda\sqrt{x/F}$.

  \item  While it is not obvious from the equations in the theorem, the leading terms consist almost entirely of the contribution coming from a small region around the point~$X_l$. 
\end{enumerate}
\end{remark}

Before we move on to proving Theorem~\ref{thm: equation R, eta xl to xl+1} we justify the introduction of the modified Pr\"ufer angle $\tilde \eta$ which was chosen in such a manner that it varies little on the scale of the intervals $[x_l, x_{l+1}]$ (which is not the case for $\eta$). Simultaneously, we prove that understanding the asymptotic behavior of $R$ can be reduced to understanding $R$ evaluated at $\{x_l\}_{l\in \N}$. Without this knowledge Theorem~\ref{thm: equation R, eta xl to xl+1} tells us very little about the asymptotic behavior of $R$. Indeed, a priori it could be that the values of $R$ along the sequence $\{x_l\}_{l\in \N}$ do not capture its general behavior. To show that this is not the case we need to prove that also $R$ does not vary too much over the intervals $[x_l, x_{l+1}]$.
 
\begin{theorem}\label{thm: eq tilde eta, prufer coords vary little between Xl}
  Fix $F>0, E\in \R$, $\lambda\in \R$, and let $R, \tilde \eta$ be the (modified) Pr\"ufer coordinates associated to a real-valued solution of~\eqref{eq: generalized eigenfunction deterministic}. Then
  \begin{align*}
  \tilde\eta(n+1)-\tilde\eta(n) 
  &=
  \frac{U(n)}{2}\cos(2\theta(n)) + \frac{U(n)^2}{4}\Bigl(\sin(2\theta(n)) - \frac{1}{2}\sin(4\theta(n))\Bigr)+ O(n^{-3/2})\,.
\end{align*}
  Moreover, for $x_l\leq x\leq x_{l+1}$,
  \begin{equation*}
    \Bigl|\log \Bigl(\frac{R(x)}{R(x_l)}\Bigr)\Bigr| \lesssim l^{-1/2}\quad\mbox{and}\quad |\tilde\eta(x)-\tilde\eta(x_l)| \lesssim  l^{-1/2}\,.
  \end{equation*}
\end{theorem}

\begin{proof}
  From Lemma~\ref{lem: gamma asymptotics} it follows that
\begin{equation*}
    \frac{1}{\gamma'(n)} = 2\sqrt{\frac{n+1}{F}}-2\sqrt{\frac{n}{F}} + O(n^{-3/2})\,.
  \end{equation*}
Therefore, by the definition of $\tilde\eta$,
  \begin{align*}
  \tilde\eta(n+1)-\tilde\eta(n) 
  &= 
  \eta(n+1)+ \lambda\sqrt{\frac{n+1}{F}}-\eta(n)- \lambda\sqrt{\frac{n}{F}}\\
  &=
  \eta(n+1)-\eta(n)+ \frac{\lambda}{2\gamma'(n)}+ O(n^{-3/2})\,.
\end{align*}
Consequently, the equation for $\eta$ in Lemma~\ref{lem: Prufer equations} implies
\begin{align*}
  \tilde\eta(n+1)-\tilde\eta(n) &=
  \frac{U(n)}{2}\cos(2\theta(n)) + \frac{U(n)^2}{4}\sin(2\theta(n))\\
  &\quad - \frac{U(n)^2}{8}\sin(4\theta(n))+ O(n^{-3/2})\,.
\end{align*}
This proves the first part of the lemma. 

By Lemma~\ref{lem: gamma asymptotics}, $U(n) = \frac{\lambda}{\sqrt{Fn}}(1+O(n^{-1}))$. Therefore, equation~\eqref{eq: R approx equation} of Lemma~\ref{lem: Prufer equations} yields
  \begin{align*}
    \log\Bigl(\frac{R(x)}{R(x_l)}\Bigr)
    = 
    \sum_{x_l< n< x} \log\Bigl(\frac{R(n+1)}{R(n)}\Bigr)
    =
    \frac{\lambda}{2}\sum_{x_l< n< x} \frac{\sin(2\theta(n))}{\gamma'(n)} + O(l^{-1})\,.
  \end{align*}
  The desired bound is obtained by an application of Theorem~\ref{thm: rough expsum decay Xl interval} with $\alpha=1, \mu=2, \beta=1/2$, and $h(n)=2\eta(n)$ which, by~\eqref{eq: eta jump bound}, satisfies the assumptions of the theorem. Similarly,
\begin{align*}
  |\tilde\eta(x)-\tilde\eta(x_l)|
  &=
  \biggl| \sum_{x_l<n< x} (\tilde\eta(n+1)-\tilde\eta(n)) \biggr|\\
  &=
  \biggl| \sum_{x_l<n< x} \Bigl[\frac{U(n)}{2}\cos(2\theta(n)) + O(n^{-1})\Bigr]  \biggr|\lesssim
  l^{-1/2}\,.
\end{align*}
This completes the proof of Theorem~\ref{thm: eq tilde eta, prufer coords vary little between Xl}.
\end{proof}

\subsection{Proof of Theorem~\ref{thm: equation R, eta xl to xl+1}}

We shall split the proof of Theorem~\ref{thm: equation R, eta xl to xl+1} into a number of lemmas. A key role in the proof of Theorem~\ref{thm: equation R, eta xl to xl+1} is played by the following result.

\begin{lemma}\label{lem: linear sum expansion}
Fix $F>0, E\in \R,$ and $\lambda \in \R$. Let $\psi$ be a real-valued solution of~\eqref{eq: generalized eigenfunction deterministic} and $\theta, \tilde\eta$ be the corresponding Pr\"ufer variables, then
  \begin{align*}
    \sum_{x_l<n\leq x_{l+1}}U(n)e^{2i\theta(n)} 
    &=
    e^{2i\tilde\eta(x_l)}\sum_{x_l<n\leq x_{l+1}}\!\!U(n)e^{2i\tilde\gamma(n)}+\frac{i}{4}\biggl|\sum_{x_l<n\leq x_{l+1}}\!\!U(n)e^{2i\tilde\gamma(n)}\biggr|^2\\
    &
    - \frac{i}{4}\!\sum_{x_l<n \leq x_{l+1}}\!\!U(n)^2 + \frac{1}{2}\Im\biggl[\,\sum_{x_l<n\leq x_{l+1}}\sum_{n<j\leq x_{l+1}}\!\!U(n)U(j)e^{2i(\tilde\gamma(n)-\tilde \gamma(j))}\biggr]\\
    &
    +\frac{ie^{4i\tilde\eta(x_l)}}{4}\biggl[\,\sum_{x_l<n\leq x_{l+1}}\!\!U(n)e^{2i\tilde\gamma(n)}\biggr]^2 + O(l^{-5/4})\,.
  \end{align*}
\end{lemma}

While at first glance the left-hand side of the identity in Lemma~\ref{lem: linear sum expansion} might look simpler than the right, this is not the case. The crucial point is that the exponential sums on the right-hand side only depend on the explicit functions $\gamma, \tilde \gamma$ and not the unknown function~$\tilde\eta$. This fact allows us to compute the sums to high precision. Indeed, apart from the double sum, this is what was done in Theorem~\ref{thm: Precise asymptotics exp sum}. Later, in the proof of Theorem~\ref{thm: equation R, eta xl to xl+1}, we will see that by Theorem~\ref{thm: Precise asymptotics exp sum} the second and third terms in Lemma~\ref{lem: linear sum expansion} cancel up to small remainders. While it might very well be possible to prove a precise asymptotic expansion also for the double sum in Lemma~\ref{lem: linear sum expansion}, for our purposes it will be sufficient to have the following rough bound.

\begin{lemma}\label{lem: remaining double sum bound}
For any $F>0$ and $E\in \R$,
  \begin{equation*}
    \biggl|\sum_{x_l<n\leq x_{l+1}}\sum_{n<j\leq x_{l+1}}\frac{e^{2i(\tilde\gamma(n)-\tilde \gamma(j))}}{\gamma'(n)\gamma'(j)}\biggr| \lesssim l^{-3/4}\,.
  \end{equation*}
\end{lemma}
\begin{remark}
  While the bound in Lemma~\ref{lem: remaining double sum bound} is good enough for us, we do not believe it to be sharp. In fact, we believe the sum to be $\lesssim l^{-1-\eps}$, for some $\eps>0$.
\end{remark}

Since the proofs of Lemmas~\ref{lem: linear sum expansion} and~\ref{lem: remaining double sum bound} are almost entirely technical exercises in summation by parts, we postpone them until after the proof that they indeed imply Theorem~\ref{thm: equation R, eta xl to xl+1}. However, before we give the proof of Theorem~\ref{thm: equation R, eta xl to xl+1} we need the following simple result which, since $U(n)^2=\lambda^2/\gamma'(n)^2$ (see Lemma~\ref{lem: Prufer equations}), determines the behavior of the non-oscillatory sum on the right-hand side in Lemma~\ref{lem: linear sum expansion}.
\begin{lemma}\label{lem: sum inverse square xl to xl+1}
  For any $F>0$ and $E\in \R$,
  \begin{equation*}
    \sum_{x_l <n \leq x_{l+1}} \frac{1}{\gamma'(n)^2} = \frac{2}{F l} + O(l^{-2})
  \end{equation*}
\end{lemma}

\begin{proof}
By Lemma~\ref{lem: gamma asymptotics} and since $x_l = \frac{\pi^2}{F}(l-1/2)^2 + O(1)$,
\begin{equation*}
  \sum_{x_l <n \leq x_{l+1}} \frac{1}{\gamma'(n)^2}  
  = \sum_{x_l <n \leq x_{l+1}} \frac{1}{Fn} + O(l^{-3})\,.
\end{equation*}
With $\gamma_E$ denoting the Euler--Mascheroni constant we have
\begin{equation*}
  \sum_{n = 1}^N \frac{1}{n} = \log(N)+ \gamma_E + O(N^{-1})\,.
\end{equation*}
Therefore, since by definition $x_l \in \N + 1/2$ and the fact that $x_l = \frac{\pi^2}{F}(l-1/2)^2+ O(1)$,
\begin{equation*}
  \sum_{x_l <n \leq x_{l+1}} \frac{1}{n} = \log\Bigl(\frac{x_{l+1}-1/2}{x_l -1/2}\Bigr) + O(x_l^{-1}) = \frac{2}{l} + O(l^{-2})\,.
\end{equation*}
This concludes the proof of the lemma.
\end{proof}

With these technical lemmas on hand we are ready to prove Theorem~\ref{thm: equation R, eta xl to xl+1}.

\begin{proof}[Proof of Theorem~\ref{thm: equation R, eta xl to xl+1}]
  Since $U(n)\sim n^{-1/2}$, equation~\eqref{eq: R approx equation} in Lemma~\ref{lem: Prufer equations} yields that
  \begin{align*}
  \log\Bigl(\frac{R(x_{l+1})}{R(x_l)}\Bigr)
  &=
  \sum_{x_l<n\leq x_{l+1}}\log\Bigl(\frac{R(n+1)}{R(n)}\Bigr)\\
  &= 
  \sum_{x_l<n\leq x_{l+1}}\biggl[\frac{U(n)}{2}\sin(2\theta(n))+\frac{U(n)^2}{8}\\
  &\qquad  - \frac{U(n)^2}{8}\Bigl(2\cos(2\theta(n))-\cos(4\theta(n))\Bigr) + O(|U(n)|^3)\biggr]\,.
  \end{align*}
  Applying Theorem~\ref{thm: rough expsum decay Xl interval},~\eqref{eq: eta jump bound}, Lemma~\ref{lem: sum inverse square xl to xl+1}, and using the asymptotic behavior of $\gamma'(x)$ and $x_l$, we deduce 
  \begin{equation}\label{eq: linear sum R xl to xl+1}
  \begin{aligned}
  \log\Bigl(\frac{R(x_{l+1})}{R(x_l)}\Bigr)
  &=
  \frac{1}{2}\sum_{x_l<n\leq x_{l+1}}U(n)\sin(2\theta(n))
   +
  \frac{\lambda^2}{4Fl}
   +O(l^{-3/2})\,.
\end{aligned}
\end{equation}

Similarly for $\tilde\eta$ we conclude from Theorem~\ref{thm: eq tilde eta, prufer coords vary little between Xl} and Theorem~\ref{thm: rough expsum decay Xl interval} combined with~\eqref{eq: eta jump bound}, that
\begin{equation}\label{eq: linear sum eta xl to xl+1}
  \tilde\eta(x_{l+1})-\tilde\eta(x_l) = \frac{1}{2}\sum_{x_l<n\leq x_{l+1}} U(n)\cos(2\theta(n))+ O(l^{-3/2})\,.
\end{equation}

Thus what needs to be understood is the exponential sum
\begin{equation*}
  \sum_{x_l<n\leq x_{l+1}}U(n)e^{2i\theta(n)}\,.
\end{equation*}
By Lemma~\ref{lem: linear sum expansion}, this sum can be written as a combination of five terms plus an error. To the first, second, and third term in the right-hand side of Lemma~\ref{lem: linear sum expansion} we apply Theorem~\ref{thm: Precise asymptotics exp sum} with $\Gamma = \Gamma_h$ for 
$$
  h(x) = 2\tilde \gamma(x)-2\gamma(x) = -2\lambda\sqrt{x/F}\,.
$$
To the third term we apply Lemma~\ref{lem: sum inverse square xl to xl+1}. In this way, we obtain
\begin{align*}
    \sum_{x_l<n\leq x_{l+1}}U(n)e^{2i\theta(n)} 
    &=
    \lambda \biggl(\frac{2}{Fl}\biggr)^{1/2}e^{2i\Gamma(l)+2i\tilde\eta(x_l)}
    - \frac{\lambda}{2\pi (l+1)}e^{i(2\tilde\gamma(x_{l+1})+2\tilde\eta(x_{l+1})-\pi/2-\pi(l+1))}
    \\
    &\quad
    + \frac{\lambda}{2\pi l}e^{i(2\tilde\gamma(x_l)+2\tilde\eta(x_{l})-\pi/2-\pi l)}
     + \frac{\lambda^2}{2}\Im\biggl[\,\sum_{x_l<n\leq x_{l+1}}\sum_{n<j\leq x_{l+1}}\frac{e^{2i(\tilde\gamma(n)-\tilde \gamma(j))}}{\gamma'(n)\gamma'(j)}\biggr]\\
    &\quad
    +\frac{i\lambda^2 }{2Fl}e^{4i\Gamma(l)+4i\tilde\eta(x_l)} + O(l^{-5/4})\,.
  \end{align*}  
  Here we also used Theorem~\ref{thm: eq tilde eta, prufer coords vary little between Xl} to estimate $e^{2i(\tilde \eta(x_l)-\tilde \eta(x_{l+1}))} = 1+ O(l^{-1/2})$ to write $\tilde \eta (x_{l+1})$ instead of $\tilde \eta(x_l)$ in the phase of the third term. Taking the real and imaginary parts and combining with~\eqref{eq: linear sum R xl to xl+1} and~\eqref{eq: linear sum eta xl to xl+1} completes the proof of Theorem~\ref{thm: equation R, eta xl to xl+1}.
\end{proof}

\begin{proof}[Proof of Lemma~\ref{lem: linear sum expansion}]
  By a summation by parts we can write the sum as
  \begin{align*}
    \sum_{x_l<n\leq x_{l+1}}U(n)e^{2i\theta(n)} 
    &=
    e^{2i\tilde\eta(x_l)}\sum_{x_l<n\leq x_{l+1}}U(n)e^{2i\tilde\gamma(n)}\\
    &\quad 
     + \sum_{x_l<n\leq x_{l+1}}\Bigl[e^{2i\tilde\eta(n+1)}-e^{2i\tilde\eta(n)}\Bigr]\sum_{n<j\leq x_{l+1}}U(j)e^{2i\tilde\gamma(j)}\,.
  \end{align*}
  By Theorem~\ref{thm: eq tilde eta, prufer coords vary little between Xl} and the expansion $e^{2ix}-1 = 2ix +O(x^2)$,
  \begin{equation*}
    e^{2i\tilde\eta(n+1)}-e^{2i\tilde\eta(n)} = iU(n)\cos(2\theta(n))e^{2i\tilde\eta(n)} + O(n^{-1})\,.
  \end{equation*}
  Therefore, using Theorem~\ref{thm: rough expsum decay Xl interval} to estimate the contribution from the error term,
  \begin{align*}
    \sum_{x_l<n\leq x_{l+1}}&\Bigl[e^{2i\tilde\eta(n+1)}-e^{2i\tilde\eta(n)}\Bigr]\sum_{n<j\leq x_{l+1}}U(j)e^{2i\tilde\gamma(j)}\\
    &=
    i\sum_{x_l<n\leq x_{l+1}}\sum_{n<j\leq x_{l+1}}U(n)U(j)\cos(2\theta(n))e^{2i(\tilde\eta(n)+\tilde \gamma(j))}+ O(l^{-3/2})\\
    &=
    \frac{i}{2}\sum_{x_l<n\leq x_{l+1}}\sum_{n<j\leq x_{l+1}}U(n)U(j)e^{-2i(\tilde\gamma(n)-\tilde \gamma(j))}\\
    &\quad +
    \frac{i}{2}\sum_{x_l<n\leq x_{l+1}}\sum_{n<j\leq x_{l+1}}U(n)U(j)e^{2i(\tilde\gamma(n)+\tilde \gamma(j)+2\tilde\eta(n))} + O(l^{-3/2})\,.
  \end{align*}

  Observing that
\begin{align*}
  \biggl|\sum_{x_l<n\leq x_{l+1}} U(n)e^{2i\tilde\gamma(n)}\biggr|^2 
  &=
  2\Re\biggl[\sum_{x_l<n \leq x_{l+1}}\sum_{n<j\leq x_{l+1}}U(n)U(j)e^{-2i(\tilde\gamma(n)-\tilde\gamma(j))}\biggr]\\
  &\quad 
  +\sum_{x_l<n\leq x_{l+1}}U(n)^2\,,
\end{align*}
the real part of the first sum can be written as
\begin{align*}
  \Re\biggl[\,\sum_{x_l<n \leq x_{l+1}}\sum_{n<j\leq x_{l+1}}U(n)U(j)e^{-2i(\tilde\gamma(n)-\tilde\gamma(j))}\biggr]&=\frac{1}{2}\biggl|\sum_{x_l<n\leq x_{l+1}} U(n)e^{2i\tilde\gamma(n)}\biggr|^2 \\
  &\quad
  - \frac{1}{2}\sum_{x_l<n\leq x_{l+1}}U(n)^2\,.
\end{align*}
Since
\begin{align*}
  \Im\biggl[\sum_{x_l<n\leq x_{l+1}}&\sum_{n<j\leq x_{l+1}}U(n)U(j)e^{-2i(\tilde\gamma(n)-\tilde \gamma(j))}\biggr]\\
  & = -\Im\biggl[\sum_{x_l<n\leq x_{l+1}}\sum_{n<j\leq x_{l+1}}U(n)U(j)e^{2i(\tilde\gamma(n)-\tilde \gamma(j))}\biggr]
\end{align*}
we have extracted the first four terms in the expression of the lemma.

What remains is to analyse the sum
\begin{equation}\label{eq: symmetric double sum}
  \sum_{x_l<n\leq x_{l+1}}\sum_{n<j\leq x_{l+1}}U(n)U(j)e^{2i(\tilde\gamma(n)+\tilde \gamma(j)+2\tilde\eta(n))}\,.
\end{equation}
Our aim is to use a similar symmetry as for the real part of the previous sum. However, in order to do this we must first show that up to an acceptable error we can replace $\tilde\eta(n)$ by $\tilde\eta(x_l)$. We claim that
\begin{equation}\label{eq: extraction of eta bound}
  \biggl|\sum_{x_l<n\leq x_{l+1}}\sum_{n<j\leq x_{l+1}}U(n)U(j)e^{2i(\tilde\gamma(n)+\tilde \gamma(j))}\Bigl[e^{4i\tilde\eta(n)}-e^{4i\tilde\eta(x_l)}\Bigr]\biggr| \lesssim l^{-5/4}\,.
\end{equation}
Thus we can replace $\tilde\eta(n)$ by $\tilde\eta(x_l)$ in~\eqref{eq: symmetric double sum} at the cost of an error $O(l^{-5/4})$ (which is sufficient for our purposes). It is in the estimate~\eqref{eq: extraction of eta bound} that the switch from the original Pr\"ufer angle $\eta$ to the slowly varying $\tilde\eta$ comes into play.

Assuming that we have proved~\eqref{eq: extraction of eta bound} one can write~\eqref{eq: symmetric double sum} as follows
\begin{align*}
  \sum_{x_l<n\leq x_{l+1}}&\sum_{n<j\leq x_{l+1}}U(n)U(j)e^{2i(\tilde\gamma(n)+\tilde \gamma(j)+2\tilde\eta(n))}\\
  &=
  e^{4i\tilde\eta(x_l)}\sum_{x_l<n\leq x_{l+1}}\sum_{n<j\leq x_{l+1}}U(n)U(j)e^{2i(\tilde\gamma(n)+\tilde \gamma(j))} + O(l^{-5/4})\\
  &=
  \frac{e^{4i\tilde\eta(x_l)}}{2}\biggl[\,\sum_{x_l<n\leq x_{l+1}}U(n)e^{2i\tilde\gamma(n)}\biggr]^2
  - \frac{e^{4i\tilde\eta(x_l)}}{2}\sum_{x_l<n\leq x_{l+1}}U(n)^2e^{4i\tilde\gamma(n)} + O(l^{-5/4})\\
  &=
   \frac{e^{4i\tilde\eta(x_l)}}{2}\biggl[\,\sum_{x_l<n\leq x_{l+1}}U(n)e^{2i\tilde\gamma(n)}\biggr]^2
  + O(l^{-5/4})\,,
\end{align*}
where in the final equality we used Theorem~\ref{thm: rough expsum decay Xl interval} to bound the contribution from the diagonal $n=j$. This is the last term in the statement of the lemma. Thus the only thing remaining to complete the proof of Lemma~\ref{lem: linear sum expansion} is a proof of~\eqref{eq: extraction of eta bound}.

To prove~\eqref{eq: extraction of eta bound} we split the sum into several regions depending on the distances of $n, j$ to the resonant point $X_l$. For $\sigma\in [1/2, 1]$ to be determined and some $c>0$, we write
\begin{equation*}
  \sum_{x_l<n\leq x_{l+1}}\sum_{n<j\leq x_{l+1}}U(n)U(j)e^{2i(\tilde\gamma(n)+\tilde \gamma(j))}\Bigl[e^{4i\tilde\eta(n)}-e^{4i\tilde\eta(x_l)}\Bigr] 
  = S_1-S_2+S_3+S_4\,,
\end{equation*}
with
\begin{align}
  S_1 
  &=
  \sum_{x_l<n\leq X_l-cl^\sigma}\sum_{x_l<j\leq x_{l+1}}U(n)U(j)e^{2i(\tilde\gamma(n)+\tilde \gamma(j))}\Bigl[e^{4i\tilde\eta(n)}-e^{4i\tilde\eta(x_l)}\Bigr]\,, \label{eq: eta extraction term1}\\
  S_2
  &=
  \sum_{x_l<n\leq X_l-cl^\sigma}\sum_{x_l<j\leq n}U(n)U(j)e^{2i(\tilde\gamma(n)+\tilde \gamma(j))}\Bigl[e^{4i\tilde\eta(n)}-e^{4i\tilde\eta(x_l)}\Bigr]\,, \label{eq: eta extraction term2}\\
  S_3
  &=
  \sum_{X_l-cl^\sigma <n\leq X_l+cl^\sigma}\sum_{n<j\leq x_{l+1}}U(n)U(j)e^{2i(\tilde\gamma(n)+\tilde \gamma(j))}\Bigl[e^{4i\tilde\eta(n)}-e^{4i\tilde\eta(x_l)}\Bigr]\,, \label{eq: eta extraction term3}\\
  S_4
  &=
  \sum_{X_l+cl^\sigma <n\leq x_{l+1}}\sum_{n<j\leq x_{l+1}}U(n)U(j)e^{2i(\tilde\gamma(n)+\tilde \gamma(j))}\Bigl[e^{4i\tilde\eta(n)}-e^{4i\tilde\eta(x_l)}\Bigr]\,. \label{eq: eta extraction term4}
\end{align}
We shall bound each of the sums~\eqref{eq: eta extraction term1}-\eqref{eq: eta extraction term4} by using Theorems~\ref{thm: expsum away from Xl} and~\ref{thm: rough expsum decay Xl interval} combined with the bounds $|U(n)|\lesssim  l^{-1}$ and $|e^{4i\tilde\eta(n)}-e^{4i\tilde\eta(x_l)}|\lesssim  |\tilde\eta(n)-\tilde\eta(x_l)|\lesssim l^{-1/2}$ (by Theorem~\ref{thm: eq tilde eta, prufer coords vary little between Xl}).

For $S_1$ the key is that the double sum decouples so that Theorem~\ref{thm: rough expsum decay Xl interval} followed by Theorem~\ref{thm: expsum away from Xl} yields
\begin{align*}
  |S_1|&=\biggl|\sum_{x_l<n\leq X_l-cl^\sigma}\sum_{x_l<j\leq x_{l+1}}U(n)U(j)e^{2i(\tilde\gamma(n)+\tilde \gamma(j))}\Bigl[e^{4i\tilde\eta(n)}-e^{4i\tilde\eta(x_l)}\Bigr]\biggr|\\
  &=
  \biggl|\sum_{x_l<j\leq x_{l+1}}U(j)e^{2i\tilde \gamma(j)}\biggr|\biggl|\sum_{x_l<n\leq X_l-cl^\sigma}U(n)e^{2i\tilde\gamma(n)}\Bigl[e^{4i\tilde\eta(n)}-e^{4i\tilde\eta(x_l)}\Bigr]\biggr|\\
  &\lesssim
  l^{-1/2}\biggl|\sum_{x_l<n\leq X_l-cl^\sigma}U(n)e^{2i\tilde\gamma(n)}\Bigl[e^{4i\tilde\eta(n)}-e^{4i\tilde\eta(x_l)}\Bigr]\biggl|\\
  &\lesssim
  l^{-1/2}\biggl(\biggl|\sum_{x_l<n\leq X_l-cl^\sigma}U(n)e^{2i\tilde\gamma(n)+4i\tilde\eta(n)}\biggl|+\biggl|\sum_{x_l<n\leq X_l-cl^\sigma}U(n)e^{2i\tilde\gamma(n)}\biggl|\biggr)\\ 
  &\lesssim
  l^{-1/2-\sigma}\,.
\end{align*}

The sum $S_2$ we can estimate brutally since we are away from the resonant point. By Theorem~\ref{thm: expsum away from Xl}
\begin{align*}
  |S_2| &= \biggl|\sum_{x_l<n\leq X_l-cl^\sigma}\sum_{x_l<j\leq n}U(n)U(j)e^{2i(\tilde\gamma(n)+\tilde\gamma(j))}\Bigl[e^{4i\tilde\eta(n)}-e^{4i\tilde\eta(x_l)}\Bigr]
   \biggr|\\
  &\leq
  \sum_{x_l<n\leq X_l-cl^\sigma}\,\underbrace{\!|U(n)|\!}_{\lesssim l^{-1}}\,\underbrace{\!\Bigl|e^{4i\tilde\eta(n)}-e^{4i\tilde\eta(x_l)}\Bigr|\!}_{\lesssim l^{-1/2}}\,\underbrace{\!\biggl|
  \sum_{x_l<j\leq n} U(j)e^{2i\tilde\gamma(j)}\biggr|\!}_{\lesssim l^{-\sigma}}
  \lesssim l^{-1/2-\sigma}\,.
\end{align*}
In an analogous manner the same bound can be proved for $S_4$.

Finally, to bound $S_3$ instead of the precise bounds in Theorem~\ref{thm: expsum away from Xl} valid away from $X_l$ we use the cruder bound of Theorem~\ref{thm: rough expsum decay Xl interval} valid over the full range, 
\begin{align*}
  \biggl|\sum_{X_l-cl^\sigma <n\leq X_l+cl^\sigma}&\sum_{n<j\leq x_{l+1}}U(n)U(j)e^{2i(\tilde\gamma(n)+\tilde \gamma(j))}\Bigl[e^{4i\tilde\eta(n)}-e^{4i\tilde\eta(x_l)}\Bigr]
  \biggr|\\
  &\leq
  \sum_{X_l-cl^\sigma <n\leq X_l+cl^\sigma}\underbrace{\!|U(n)|\!}_{\lesssim l^{-1}}\, \underbrace{\!\Bigl|e^{4i\tilde\eta(n)}-e^{4i\tilde\eta(x_l)}\Bigr|\!}_{\lesssim l^{-1/2}}\,\underbrace{\!\biggl|\sum_{n<j\leq x_{l+1}} U(j)e^{2i\tilde\gamma(j)}\biggr|\!}_{\lesssim l^{-1/2}}
  \lesssim l^{\sigma-2}\,,
\end{align*}
where we use that the number of terms in the $n$-sum is $\lesssim l^\sigma$.

We have arrived at the bound
\begin{align*}
  \biggl|\sum_{x_l<n\leq x_{l+1}}\sum_{n<j\leq x_{l+1}}U(n)U(j)e^{2i(\tilde\gamma(n)+\tilde \gamma(j))}\Bigl[e^{4i\tilde\eta(n)}-e^{4i\tilde\eta(x_l)}\Bigr] \biggr|\lesssim  l^{\sigma-2}+l^{-1/2-\sigma}\,.
\end{align*}
Choosing $\sigma=3/4$ completes the proof of~\eqref{eq: extraction of eta bound}, and hence the proof of Lemma~\ref{lem: linear sum expansion}.
\end{proof}

The proof of Lemma~\ref{lem: remaining double sum bound} is similar to that of~\eqref{eq: extraction of eta bound} above.

\begin{proof}[Proof of Lemma~\ref{lem: remaining double sum bound}]
For $\sigma \in [1/2, 1]$ to be chosen and some $c>0$ we write the sum as
  \begin{equation*}
    \sum_{x_l<n\leq x_{l+1}}\sum_{n<j\leq x_{l+1}}\frac{e^{2i(\tilde\gamma(n)-\tilde \gamma(j))}}{\gamma'(n)\gamma'(j)}
    = S_1-S_2+S_3+S_4\,,
  \end{equation*}
 where this time
\begin{align}
  S_1
  &=\sum_{x_l<n\leq X_l-cl^\sigma}\sum_{x_l<j\leq x_{l+1}}\frac{e^{2i(\tilde\gamma(n)-\tilde \gamma(j))}}{\gamma'(n)\gamma'(j)}\,,
  \label{eq: double sum term1}\\
  S_2&=
  \sum_{x_l<n\leq X_l-cl^\sigma}\sum_{x_l<j\leq n}\frac{e^{2i(\tilde\gamma(n)-\tilde \gamma(j))}}{\gamma'(n)\gamma'(j)}\,,
  \label{eq: double sum term2}\\
  S_3&=
  \sum_{X_l-cl^\sigma <n\leq X_l+cl^\sigma}\sum_{n<j\leq x_{l+1}}\frac{e^{2i(\tilde\gamma(n)-\tilde \gamma(j))}}{\gamma'(n)\gamma'(j)}\,,
  \label{eq: double sum term3}\\
  S_4 &=
  \sum_{X_l+cl^\sigma <n\leq x_{l+1}}\sum_{n<j\leq x_{l+1}}\frac{e^{2i(\tilde\gamma(n)-\tilde \gamma(j))}}{\gamma'(n)\gamma'(j)}
  \,.\label{eq: double sum term4}
\end{align}
Again we shall bound each of these terms by combined application of Theorems~\ref{thm: expsum away from Xl} and~\ref{thm: rough expsum decay Xl interval}.

For $S_1$ Theorems~\ref{thm: expsum away from Xl} and~\ref{thm: rough expsum decay Xl interval} applied as before yields
\begin{equation*}
  |S_1|
  =
  \biggl|\sum_{x_l<n\leq X_l-cl^\sigma}\frac{e^{2i\tilde\gamma(n)}}{\gamma'(n)}\biggr|
  \biggl|\sum_{x_l<j\leq x_{l+1}}\frac{e^{-2i\tilde \gamma(j)}}{\gamma'(j)}\biggr|
  \lesssim
  l^{-1/2-\sigma}\,.
\end{equation*}

For $S_2$ we are again far from the resonant points and Theorem~\ref{thm: expsum away from Xl} implies
\begin{equation*}
  |S_2|=\biggl|\sum_{x_l<n\leq X_l-cl^\sigma}\sum_{x_l<j\leq n}\frac{e^{2i(\tilde\gamma(n)-\tilde \gamma(j))}}{\gamma'(n)\gamma'(j)}\biggr|
  \leq
  \sum_{x_l<n\leq X_l-cl^\sigma}\frac{1}{\gamma'(n)}\biggl|\sum_{x_l<j\leq n}\frac{e^{-2i\tilde \gamma(j))}}{\gamma'(j)}\biggr| \lesssim l^{-\sigma}\,.
\end{equation*}
Analogously the same bound holds for $S_4$.

For $S_3$ we use Theorem~\ref{thm: rough expsum decay Xl interval} to deduce 
\begin{align*}
  |S_3|  
  &\leq
  \sum_{X_l-cl^\sigma<n \leq X_l+cl^\sigma}\frac{1}{\gamma'(n)}\biggl|\sum_{n<j\leq x_{l+1}}\frac{e^{-2i\tilde \gamma(j)}}{\gamma'(j)}\biggr| \lesssim l^{\sigma-3/2}\,.
\end{align*}

Collecting the four bounds we have proved that
\begin{align*}
  \biggl|\sum_{x_l<n\leq x_{l+1}}\sum_{n<j\leq x_{l+1}}\frac{e^{2i(\tilde\gamma(n)-\tilde \gamma(j))}}{\gamma'(n)\gamma'(j)}\biggr|\lesssim l^{-\sigma} + l^{\sigma-3/2}\,.
\end{align*}
Choosing $\sigma=3/4$ completes the proof Lemma~\ref{lem: remaining double sum bound}.
\end{proof}

\subsection{Proof of Theorem~\ref{thm: l scale equations intro}}

We end this section by proving Theorem~\ref{thm: l scale equations intro}. Most of what is needed for the proof has been accomplished above. What remains is to  make one final modification of our Pr\"ufer variables to obtain the simplified equations stated in Theorem~\ref{thm: l scale equations intro} and interpret what has been done in terms of the asymptotic representation of $\psi$ in~\eqref{eq: generalized eigenfunction asymptotic form}.

\begin{proof}[Proof of Theorem~\ref{thm: l scale equations intro}]
Let $\psi$ be real-valued solution of~\eqref{eq: generalized eigenfunction deterministic}. By the definition of our Pr\"ufer variables
\begin{equation*}
  \psi(x) = \frac{R(x)}{2i}\Bigl(e^{i(\tilde\eta(x)- \lambda \sqrt{\lceil x\rceil/F})}\refsol(x)-e^{-i(\tilde\eta(x)- \lambda \sqrt{\lceil x\rceil/F})}\bar\refsol(x)\Bigr)\,.
\end{equation*}
Our aim is to prove that, up to a relatively small error, we have a corresponding representation with $R, \tilde \eta$ replaced by functions which are constant on each of the intervals $(\frac{\pi^2}{F}(l - \frac{1}{2})^2, \frac{\pi^2}{F}(l + \frac{1}{2})^2]$.

By Theorem~\ref{thm: eq tilde eta, prufer coords vary little between Xl} we have, for any $x \in (\frac{\pi^2}{F}(l - \frac{1}{2})^2, \frac{\pi^2}{F}(l + \frac{1}{2})^2]$,
\begin{equation*}
  \Bigl|\log\Bigl(\frac{R(x)}{R(x_l)}\Bigr)\Bigr| \lesssim l^{-1/2} \quad \mbox{and} \quad |\tilde\eta(x)-\tilde\eta(x_l)| \lesssim l^{-1/2}\,.
\end{equation*}
Consequently
\begin{equation}\label{eq: local error R tilde eta}
  R(x) = R(x_l)(1+ O(l^{-1/2})) \quad \mbox{and} \quad e^{i \tilde \eta(x)} = e^{i\tilde\eta(x_l)}(1+O(l^{-1/2}))
\end{equation}
and
\begin{equation}\label{eq: l scale prufer representation}
  \psi(x) = \frac{R(x_l)}{2i}\Bigl(e^{i(\tilde\eta(x_l)- \lambda \sqrt{\lceil x\rceil/F})}\refsol(x)-e^{-i(\tilde\eta(x_l)- \lambda \sqrt{\lceil x\rceil/F})}\bar\refsol(x)\Bigr) + O(|\refsol(x)|R(x_l)l^{-1/2})\,.
\end{equation}
While this accomplishes our main goal of finding the desired representation, by making slight modifications of $R(x_l), \tilde \eta(x_l)$ we can simplify the equations that the functions satisfy. 

To this end define
\begin{align*}
    \mathcal{R}(l) = R(x_l)e^{\frac{(-1)^{l+1}\lambda}{4\pi l}\cos(2\tilde\gamma(x_l)+2\tilde\eta(x_l))} \quad \mbox{and} \quad \Lambda(l) = \tilde\eta(x_l)+ \frac{(-1)^l\lambda}{4\pi l}\sin(2\tilde\gamma(x_l)+2\tilde\eta(x_l))\,.
\end{align*} 
It is an immediate consequence of the definitions that
\begin{equation}\label{eq: large scale prufer mod error}
  \mathcal{R}(l)= R(x_l) (1+ O(l^{-1})) \quad \mbox{and} \quad \Lambda(l) = \tilde \eta(x_l) + O(l^{-1})
\end{equation}
so we can without changing the error estimate replace $R(x_l)$ by $\mathcal{R}(l)$ and $\tilde\eta(x_l)$ by $\Lambda(l)$ in~\eqref{eq: l scale prufer representation}. Moreover, from Theorem~\ref{thm: equation R, eta xl to xl+1} we see that $\mathcal{R}, \Lambda$ satisfy
\begin{align*}
  \log\Bigl(\frac{\mathcal{R}(l+1)}{\mathcal{R}(l)}\Bigr) 
    &=
    \frac{\lambda}{\sqrt{2Fl}}\sin(2\Gamma(l)+2\tilde\eta(x_l))
    + \frac{\lambda^2}{4F l}\bigl(1+\cos(4\Gamma(l)+4\tilde\eta(x_l))\bigr)+ O(l^{-5/4})\,,\\
  \Lambda(l+1)-\Lambda(l) &=
    \frac{\lambda}{\sqrt{2F l}}\cos(2\Gamma(l)+2\tilde\eta(x_l))- \frac{\lambda^2}{4Fl}\sin(4\Gamma(l)+4\tilde\eta(x_l))\\
    &\quad 
    + \frac{\lambda^2}{4}\Im\biggl[\,\sum_{x_l<n\leq x_{l+1}}\sum_{n<j\leq x_{l+1}}\frac{e^{2i(\tilde\gamma(n)-\tilde\gamma(j))}}{\gamma'(n)\gamma'(j)}\biggr]
    + O(l^{-5/4})\,.
\end{align*}
Indeed, the choice of $\mathcal{R}, \Lambda$ was made precisely to absorb the terms in the equations of Theorem~\ref{thm: equation R, eta xl to xl+1} which are no longer present. 

Defining, as in Theorem~\ref{thm: l scale equations intro}, $\Theta(l) = \Gamma(l)+\Lambda(l)$, and combining the second estimate in~\eqref{eq: large scale prufer mod error} with Taylor expansions of the $\cos$ and $\sin$ terms we arrive at
\begin{equation}\label{eq: large scale Radius eq}
  \log\Bigl(\frac{\mathcal{R}(l+1)}{\mathcal{R}(l)}\Bigr) 
    =
    \frac{\lambda}{\sqrt{2Fl}}\sin(2\Theta(l))
    + \frac{\lambda^2}{4F l}\bigl(1+\cos(4\Theta(l))\bigr)+ O(l^{-5/4})\,,
\end{equation}
and
\begin{equation}\label{eq: large scale angle eq}
\begin{aligned}
  \Lambda(l+1)-\Lambda(l) &=
    \frac{\lambda}{\sqrt{2F l}}\cos(2\Theta(l))- \frac{\lambda^2}{4Fl}\sin(4\Theta(l))\\
    &\quad 
    + \frac{\lambda^2}{4}\Im\biggl[\,\sum_{x_l<n\leq x_{l+1}}\sum_{n<j\leq x_{l+1}}\frac{e^{2i(\tilde\gamma(n)-\tilde\gamma(j))}}{\gamma'(n)\gamma'(j)}\biggr]
    + O(l^{-5/4})\,.
\end{aligned}
\end{equation}
Writing
\begin{equation*}
  \mathcal{S}(l) = \frac{1}{4}\Im\biggl[\,\sum_{x_l<n\leq x_{l+1}}\sum_{n<j\leq x_{l+1}}\frac{e^{2i(\tilde\gamma(n)-\tilde\gamma(j))}}{\gamma'(n)\gamma'(j)}\biggr]
\end{equation*}
these are the equations claimed in Theorem~\ref{thm: l scale equations intro}. Clearly $\mathcal{S}$ is independent of $\psi, \mathcal{R}, \Lambda$, and $|\mathcal{S}(l)|\lesssim l^{-3/4}$ by Lemma~\ref{lem: remaining double sum bound}. 

Finally, by~\eqref{eq: local error R tilde eta} and~\eqref{eq: large scale prufer mod error} it holds that
\begin{equation*}
  \sum_{x_l<n\leq x_{l+1}} \frac{R(n)^2}{2\sqrt{Fn}}= \frac{\mathcal{R}(l)^2}{2\sqrt{F}}(1+O(l^{-1/2}))\sum_{x_l<n\leq x_{l+1}} \frac{1}{\sqrt{n}}= \frac{\pi}{F} \mathcal{R}(l)^2(1+O(l^{-1/2}))\,.
\end{equation*}
Therefore, Lemma~\ref{lem: L2 norm comparability} implies
\begin{equation*}
  \int_{\frac{\pi^2}{F}(l - \frac{1}{2})^2}^{\frac{\pi^2}{F}(l + \frac{1}{2})^2}|\psi(x)|^2 \,dx  = \frac{\pi}{F}\mathcal{R}(l)^2(1+O(l^{-1/2}))\,.
\end{equation*}
This completes the proof of Theorem~\ref{thm: l scale equations intro}.
\end{proof}

\begin{remark}
  While substantial improvements of the error term in the representation~\eqref{eq: generalized eigenfunction asymptotic form} for $x$ in the full range $(\frac{\pi^2}{F}(l- \frac{1}{2})^2, \frac{\pi^2}{F}(l+ \frac{1}{2})^2)$ are unlikely to be possible, it is possible to improve the error on large subsets of these intervals. Indeed, by using the bounds of Theorem~\ref{thm: expsum away from Xl} in place of Theorem~\ref{thm: rough expsum decay Xl interval} in the proof of Theorem~\ref{thm: eq tilde eta, prufer coords vary little between Xl} one finds that, for $x \in (X_{l-1}+Cl^\sigma, X_l-C l^\sigma)$ with $\sigma \in [1/2, 1]$,
  \begin{equation*}
    \Bigl|\log\Bigl(\frac{R(x)}{\mathcal{R}(l)}\Bigr) \Bigr| \lesssim l^{-\sigma} \quad \mbox{and}\quad |\tilde\eta(x)-\Lambda(l)| \lesssim l^{-\sigma}\,.
  \end{equation*}
  Plugging these refinements into the proof of Theorem~\ref{thm: l scale equations intro}, one obtains improved error estimates in~\eqref{eq: generalized eigenfunction asymptotic form} as soon as the distance from $x$ to the resonant points is large enough. Note that Buslaev \cite{Buslaev_KronigPenney_99} predicted a representation of the form~\eqref{eq: generalized eigenfunction asymptotic form} whenever $|X_l-x|\gtrsim l^\sigma$ with $\sigma=1/2$ and Pozharski\u{\i} \cite{Pozharskii_AlgAnal_02} proved such a representation in the case of a more regular periodic potential with $\sigma=2/3$.
\end{remark}

\section{The deterministic model with rationality}
\label{sec: asymptotics R deterministic}

We now specialize to the case when $F \in \pi^2\Q_\limplus$ and as in Theorem~\ref{thm: full-line random} fix $p, q \in \N$ so that $\frac{\pi^2}{3F} = \frac{p}{q}$. While it is not strictly necessary, the statement we prove is strongest when $p, q$ are chosen so that $\gcd(p, q)=1$.

\subsection{A second coarse graining}

 We begin by noting that since $q$ is fixed (and finite) the equation for $\mathcal{R}$ in Theorem~\ref{thm: l scale equations intro} yields, for any $qk \leq l\leq q(k+1)$,
\begin{equation}\label{eq: cR varies little on qk q(k+1)}
  \Bigl|\log\Bigl(\frac{\mathcal{R}(l)}{\mathcal{R}(qk)}\Bigr)\Bigr| \lesssim k^{-1/2}.
\end{equation}
Consequently, if $\lim_{k \to \infty} \mathcal{R}(qk)$ exists then $\lim_{l \to \infty} \mathcal{R}(l)$ exists, and by Theorem~\ref{thm: eq tilde eta, prufer coords vary little between Xl} so does $\lim_{n\to \infty}R(n)$. Moreover, all three limits coincide.

In a manner similar to that in the previous section we wish to compute asymptotic equations for $\mathcal{R}, \Lambda$ when we transition from $l=qk$ to $l=q(k+1)$. The idea of the proof is almost identical to what we have done before, but as we shall see the assumption $\pi^2F \in \Q_\limplus$ will lead to crucial simplifications in the effective model that arises. 

\begin{theorem}\label{thm: transition xqk to xq(k+1)}
  Let $F= \frac{\pi^2 q}{3p}$ with $p, q\in \N$ and set
  \begin{align*}
    \Omega(k) = 3p\frac{E-\lambda}{\pi}k + \frac{5\pi}{8}
     \quad \mbox{and}\quad
    w(E, \lambda, q, p) = \sum_{j=0}^{q-1} e^{-2\pi i \frac{p}{q}j^3+ 6ip\frac{E-\lambda}{q\pi}j}\,.
  \end{align*}
  Let $\psi$ be a real-valued solution of~\eqref{eq: generalized eigenfunction deterministic} and $\mathcal{R}, \Lambda$ the associated Pr\"ufer coordinates as in Theorem~\ref{thm: l scale equations intro}, then
  \begin{align*}
  \log\Bigl(\frac{\mathcal{R}(q(k+1))}{\mathcal{R}(qk)}\Bigr) 
  &=
   \frac{\lambda}{\sqrt{2Fqk}} \Im\Bigl[e^{2i(\Omega(k)+\Lambda(qk))}w(E, \lambda, q, p)\Bigr] + \frac{\lambda^2}{4F qk}|w(E, \lambda, q, p)|^2\\
   &\quad 
   + \frac{\lambda^2}{4Fqk} \Re\Bigl[e^{4i(\Omega(k)+\Lambda(qk))}w(E, \lambda, q, p)^2\Bigr]
  +O(k^{-5/4})
\end{align*}
and
\begin{align*}
  \Lambda(q(k+1))-\Lambda(qk)
  &=
  \frac{\lambda}{\sqrt{2Fqk}}\Re\Bigl[e^{2i(\Omega(k)+\Lambda(qk))}w(E, \lambda, q, p)\Bigr] + O(k^{-3/4})\,.
\end{align*}
\end{theorem}

\begin{proof}
By the equation for $\mathcal{R}$ in Theorem~\ref{thm: l scale equations intro}
\begin{equation}\label{eq: sum qk to q(k+1)}
\begin{aligned}
  \log\Bigl(\frac{\mathcal{R}(q(k+1))}{\mathcal{R}(qk)}\Bigr)\! &= 
  \sum_{j=0}^{q-1} \log\Bigl(\frac{\mathcal{R}(qk+j+1)}{\mathcal{R}(qk+j)}\Bigr)\\
  &=
  \frac{\lambda}{\sqrt{2F}}\sum_{j=0}^{q-1} 
   \frac{\sin(2\Theta(qk+j))}{\sqrt{qk+j}}
    + \frac{\lambda^2}{4F}\sum_{j=0}^{q-1}\frac{1+ \cos(4\Theta(qk+j))}{qk+j} +O(k^{-5/4})\,.
\end{aligned}
\end{equation}

We begin by analysing the sum of terms with decay $\sim k^{-1/2}$,
\begin{equation*}
  \sum_{j=0}^{q-1} 
    \frac{\sin(2\Theta(qk+j))}{\sqrt{qk+j}} = \Im\biggl[\,\sum_{j=0}^{q-1}\frac{e^{2i\Theta(qk+j))}}{\sqrt{qk+j}}\biggr]\,.
\end{equation*}

By the equation for $\Lambda$ in Theorem~\ref{thm: l scale equations intro} and since $q$ is fixed and finite we have for all $qk\leq l \leq q(k+1)$ the estimate
\begin{equation}\label{eq: Lambda varies little q scale}
  |\Lambda(qk)-\Lambda(l)|\lesssim k^{-1/2}\,.
\end{equation}
Combining the equation for $\Lambda$ in Theorem~\ref{thm: l scale equations intro}, the estimate~\eqref{eq: Lambda varies little q scale}, and the expansion $e^{x} = 1+ x+ O(x^2)$, we can write
\begin{equation}\label{eq: 1/2 decay qk to q(k+1)}
\begin{aligned}
  \sum_{j=0}^{q-1}\frac{e^{2i\Theta(qk+j))}}{\sqrt{qk+j}}
  &=
  e^{2i\Lambda(qk)}\sum_{j=0}^{q-1}\frac{e^{2i\Gamma(qk+j)}}{\sqrt{qk+j}}e^{2i(\Lambda(qk+j)-\Lambda(qk))}\\
  &=
  e^{2i\Lambda(qk)}\sum_{j=0}^{q-1}\frac{e^{2i\Gamma(qk+j)}}{\sqrt{qk+j}}\Bigl[1+2i(\Lambda(qk+j)-\Lambda(qk))+ O(k^{-1})\Bigr]\\
  &=
  e^{2i\Lambda(qk)}\sum_{j=0}^{q-1}\frac{e^{2i\Gamma(qk+j)}}{\sqrt{qk+j}} \\
  &\quad
  +i\lambda e^{2i\Lambda(qk)}\sqrt{\frac{2}F}\sum_{j=0}^{q-1}\frac{e^{2i\Gamma(qk+j)}}{\sqrt{qk+j}}\sum_{r=0}^{j-1}\frac{\cos(2\Theta(qk+r))}{\sqrt{qk+r}} + O(k^{-5/4})\,.
\end{aligned}
\end{equation}

Since $F = \frac{\pi^2 q}{3p}$ we have
\begin{align*}
  \Gamma(qk) 
  &= - \pi pq^2k^3+ \frac{\pi qk}{F}(E-\lambda) + \frac{5\pi}{8}\,,\\
  \Gamma(qk+j)-\Gamma(qk) 
  &= - \frac{\pi p}{q}j^3-3\pi j^2kp - 3\pi jk^2pq + \frac{\pi j}{F}(E-\lambda) 
\end{align*}
so that
\begin{align*}
  e^{2i\Gamma(qk)} 
  &=  e^{2i(\frac{\pi qk}{F}(E-\lambda) + \frac{5\pi}{8})} = e^{2i\Omega(k)}\,,\\
  e^{2i(\Gamma(qk+j)-\Gamma(qk))}  
  &= e^{- 2\pi i\frac{p}{q}j^3 + 2i\frac{\pi j}{F}(E-\lambda)}= e^{- 2\pi i\frac{ p}{q}j^3 + 6ip\frac{E-\lambda}{q\pi}j}\,. 
\end{align*}
Therefore, since
\begin{equation*}
  \frac{1}{\sqrt{qk+j}}- \frac{1}{\sqrt{qk}} = O(k^{-3/2})
\end{equation*}
we can rewrite the first sum in the right-hand side of~\eqref{eq: 1/2 decay qk to q(k+1)} as
\begin{equation}\label{eq: 1/2 decay qk to q(k+1) rewritten}
\begin{aligned}
  \sum_{j=0}^{q-1}\frac{e^{2i\Gamma(qk+j)}}{\sqrt{qk+j}}
  &=
  \frac{e^{2i\Gamma(qk)}}{\sqrt{qk}}\sum_{j=0}^{q-1}e^{2i(\Gamma(qk+j)-\Gamma(qk))} + O(k^{-3/2})\\
  &=
  \frac{e^{2i\Omega(k)}}{\sqrt{qk}}w(E, \lambda, q, p) + O(k^{-3/2})\,.
\end{aligned}
\end{equation}
This is exactly the term with $\sim k^{-1/2}$ decay in the equation for $\mathcal{R}$ in Theorem~\ref{thm: transition xqk to xq(k+1)}. 

Using the same idea as in the proof of Theorem~\ref{thm: equation R, eta xl to xl+1} we can write parts of the double sum in the right-hand side of~\eqref{eq: 1/2 decay qk to q(k+1)} as squares of sums. We argue as follows,
\begin{align*}
  \sum_{j=0}^{q-1}&\frac{e^{2i\Gamma(qk+j)}}{\sqrt{qk+j}}\sum_{r=0}^{j-1}\frac{\cos(2\Theta(qk+r))}{\sqrt{qk+r}}\\
  &=
  \frac{1}{2}\sum_{j=0}^{q-1}\sum_{r=0}^{j-1}\frac{e^{2i(\Gamma(qk+j)+\Theta(qk+r))}}{\sqrt{qk+j}\sqrt{qk+r}}
  +
  \frac{1}{2}\sum_{j=0}^{q-1}\sum_{r=0}^{j-1}\frac{e^{2i(\Gamma(qk+j)-\Theta(qk+r))}}{\sqrt{qk+j}\sqrt{qk+r}}\,.
\end{align*}
Using~\eqref{eq: Lambda varies little q scale} and the fact that, for $0\leq r, j\leq q-1$,
\begin{equation*}
  \frac{1}{\sqrt{qk+r}\sqrt{qk+j}} - \frac{1}{qk} = O(k^{-2})
\end{equation*}
we find
\begin{equation}\label{eq: Gamma doublesum split}
\begin{aligned}
  \sum_{j=0}^{q-1}\frac{e^{2i\Gamma(qk+j)}}{\sqrt{qk+j}}\sum_{r=0}^{j-1}\frac{\cos(2\Theta(qk+r))}{\sqrt{qk+r}}
  &=
  \frac{e^{2i\Lambda(qk)}}{2qk}\sum_{j=0}^{q-1}\sum_{r=0}^{j-1}e^{2i(\Gamma(qk+j)+\Gamma(qk+r))} \\
  &
  +
  \frac{e^{-2i\Lambda(qk)}}{2qk}\sum_{j=0}^{q-1}\sum_{r=0}^{j-1}e^{2i(\Gamma(qk+j)-\Gamma(qk+r))}+ O(k^{-3/2})\,.\hspace{-8pt}
\end{aligned}
\end{equation}

We next observe that
\begin{equation}\label{eq: double sum Gamma +}
\begin{aligned}
  \sum_{j=0}^{q-1}\sum_{r=0}^{j-1}e^{2i(\Gamma(qk+j)+\Gamma(qk+r))}
  &=
  \frac{1}{2}\biggl(\sum_{j=0}^{q-1}e^{2i\Gamma(qk+j)}\biggr)^2- \frac{1}{2}\sum_{j=0}^{q-1}e^{4i\Gamma(qk+j)}\\
  &=
  \frac{e^{4i\Omega(k)}}{2}\biggl[\biggl(\sum_{j=0}^{q-1}e^{2i(\Gamma(qk+j)-\Gamma(qk))}\biggr)^2- \sum_{j=0}^{q-1}e^{4i(\Gamma(qk+j)-\Gamma(qk))}\biggr]\\
  &=
  \frac{e^{4i\Omega(k)}}{2}\Bigl[w(E, \lambda, q, p)^2-w(E, \lambda, q, 2p)\Bigr]\,,
\end{aligned}
\end{equation}
and similarly
\begin{equation}\label{eq: double sum Gamma -}
\begin{aligned}
  \sum_{j=0}^{q-1}\sum_{r=0}^{j-1}e^{2i(\Gamma(qk+j)-\Gamma(qk+r))}
  \!&=
  \frac{1}{2}\biggl|\sum_{j=0}^{q-1}e^{2i\Gamma(qk+j)}\biggr|^2- \frac{q}{2}+ i \Im\biggl[\,\sum_{j=0}^{q-1}\sum_{r=0}^{j-1}e^{2i(\Gamma(qk+j)-\Gamma(qk+r))}\biggr]\\
  &=
  \frac{1}{2}\biggl|\sum_{j=0}^{q-1}e^{2i(\Gamma(qk+j)-\Gamma(qk))}\biggr|^2\!\!- \frac{q}{2}+ i \Im\biggl[\,\sum_{j=0}^{q-1}\sum_{r=0}^{j-1}\!e^{2i(\Gamma(qk+j)-\Gamma(qk+r))}\biggr]\\
  &=
  \frac{1}{2}|w(E, \lambda, q, p)|^2- \frac{q}{2}+ i \Im\biggl[\,\sum_{j=0}^{q-1}\sum_{r=0}^{j-1}e^{2i(\Gamma(qk+j)-\Gamma(qk+r))}\biggr]\,.
\end{aligned}
\end{equation}

Collecting~\eqref{eq: 1/2 decay qk to q(k+1) rewritten}--\eqref{eq: double sum Gamma -} and inserting them into~\eqref{eq: 1/2 decay qk to q(k+1)} we arrive at
\begin{equation}\label{eq: 1/2 decay sum final qk to q(k+1)}
\begin{aligned}
    \sum_{j=0}^{q-1}&\frac{e^{2i\Theta(qk+j))}}{\sqrt{qk+j}}\\
    &=
    \frac{1}{\sqrt{qk}}e^{2i(\Omega(k)+\Lambda(qk))}w(E, \lambda, q, p) 
  + \frac{i\lambda e^{4i(\Omega(k)+\Lambda(qk))}}{\sqrt{8F}qk}\Bigl[w(E, \lambda, q, p)^2-w(E, \lambda, q, 2p)\Bigr] \\
  &\quad 
  + 
  \frac{i\lambda}{\sqrt{8F}qk}\Biggl[|w(E, \lambda, q, p)|^2- q+ 2i \Im\biggl[\,\sum_{j=0}^{q-1}\sum_{r=0}^{j-1}e^{2i(\Gamma(qk+j)-\Gamma(qk+r))}\biggr]\Biggr] + O(k^{-5/4})\,.
\end{aligned}
\end{equation}
Therefore, for the first sum in~\eqref{eq: sum qk to q(k+1)} we have proved that
\begin{equation*}
  \begin{aligned}
    \frac{\lambda}{\sqrt{2F}}&\sum_{j=0}^{q-1} 
    \frac{\sin(2\Theta(qk+j))}{\sqrt{qk+j}}\\
    &=
   \frac{\lambda}{\sqrt{2Fqk}}\Im\Bigl[e^{2i(\Omega(k)+\Lambda(qk))}w(E, \lambda, q, p)\Bigr] +\frac{\lambda^2}{4Fqk}\Bigl[|w(E, \lambda, q, p)|^2- q\Bigr]\\
   &\quad
   +\frac{\lambda^2}{4Fqk}\Re\biggl[
  e^{4i(\Omega(k)+\Lambda(qk))}\Bigl[w(E, \lambda, q, p)^2-w(E, \lambda, q, 2p)\Bigr]\biggr] + O(k^{-5/4})\,.
  \end{aligned}
\end{equation*}

For the second sum in~\eqref{eq: sum qk to q(k+1)} the analogous calculation leads to
\begin{align*}
  \frac{\lambda^2}{4F}\sum_{j=0}^{q-1} \frac{1+\cos(4\Theta(qk+j))}{qk+j}
  &=
  \frac{\lambda^2}{4Fqk}\biggl[q + \Re\Bigl[e^{4i(\Omega(k)+\Lambda(qk))}w(E, \lambda, q, 2p)\Bigr]\biggr] + O(k^{-3/2})\,.
\end{align*}
We note that with the decay $\sim k^{-1}$ present in this sum the expansion is of sufficiently high precision already if we bound the double sum that arises in the analogue of~\eqref{eq: 1/2 decay qk to q(k+1)} trivially. 

Inserting all of the above into~\eqref{eq: sum qk to q(k+1)} proves the equation for $\mathcal{R}$ claimed in Theorem~\ref{thm: transition xqk to xq(k+1)}.

For $\Lambda$ the equation in Theorem~\ref{thm: l scale equations intro} combined with~\eqref{eq: 1/2 decay sum final qk to q(k+1)} yields
\begin{align*}
  \Lambda(q(k+1))-\Lambda(qk)
  &=
  \sum_{j=0}^{q-1}(\Lambda(qk+j+1)-\Lambda(qk+j))\\
  &=
  \frac{\lambda}{\sqrt{2F}}\sum_{j=0}^{q-1}
    \frac{\cos(2\Theta(qk+j))}{\sqrt{qk+j}}  + O(k^{-3/4})\\
  &=
  \frac{\lambda}{\sqrt{2Fqk}}\Re\Bigl[e^{2i(\Omega(k)+\Lambda(qk))}w(E, \lambda, q, p)\Bigr] + O(k^{-3/4})\,.
\end{align*}
This completes the proof of Theorem~\ref{thm: transition xqk to xq(k+1)}.
\end{proof}

\subsection{Asymptotics of \texorpdfstring{$R$}{R}}

The main goal of this section is to prove the following theorem.
\begin{theorem}\label{thm: convergence R deterministic}
   Fix $\lambda \in \R$ and $F \in \pi^2 \Q_\limplus$, writing $F= \frac{\pi^2 q}{3p}$ with $q, p \in \N$. If $R$ is the Pr\"ufer radius of a solution $\psi$ of~\eqref{eq: generalized eigenfunction deterministic} with $E \in \R \setminus \{\tfrac{\pi^2}{3p}m + \lambda: m \in \Z\}$, then
   \begin{equation*}
     \lim_{n\to \infty} \log(R(n))
   \end{equation*}
   exists and is finite. Moreover, if $p, q, m$ are such that 
   \begin{equation}\label{eq: w at exceptional energies}
     \sum_{j=0}^{q-1}e^{-2\pi i \frac{pj^3-jm}{q}} =0
   \end{equation}
   then the same limit exists also when $E = \frac{\pi^2}{3p}m +\lambda$.
\end{theorem} 

The sum in the left-hand side of~\eqref{eq: w at exceptional energies} is $w(\frac{\pi^2}{4p}m + \lambda, \lambda, q, p)$. In addition to appearing also in Perelman's argument~\cite{perelman_AsympAnal2005}, this sum appears in work of Fedotov and Klopp~\cite{fedotov_starkwannier_2016}. Specifically, Fedotov and Klopp were interested in resonances of the Schr\"odinger operator
\begin{equation*}
     -\frac{d^2}{dx^2}-Fx +2\cos(2\pi x) \quad \mbox{in }L^2(\R)\,.
  \end{equation*} 
Of particular interest for us are their Lemmas~1 and~2 which state that, for co-prime $p, q\in \N$,
\begin{equation*}
  \sum_{m=0}^{q-1}\Bigl|\sum_{j=0}^{q-1}e^{-2\pi i \frac{pj^3-jm}{q}}\Bigr|^2 = q^2 \quad \mbox{and} \quad \#\Bigl\{0\leq m < q: \sum_{j=0}^{q-1}e^{-2\pi i \frac{pj^3-jm}{q}} \neq 0\Bigr\} \gtrsim q^{2/3}\,.
\end{equation*}
In particular, the sum cannot vanish for all $m$. However, there are cases when the sum vanishes for all but one $m$ (for instance, when $q=2, 3,$ or $6$).

\begin{proof}[Proof of Theorem~\ref{thm: convergence R deterministic}]
We wish to sum the equations of Theorem~\ref{thm: transition xqk to xq(k+1)} in order to conclude that $\log(\mathcal{R}(qk))$ has a limit as $k$ tends to infinity. By~\eqref{eq: cR varies little on qk q(k+1)} and Theorem~\ref{thm: eq tilde eta, prufer coords vary little between Xl} this implies that the claim of the theorem. Our aim is to show that $\log(\mathcal{R}(qk))$ forms a Cauchy sequence in~$k$. That is we want to prove that for any $\eps>0$ there exists $K_0$ such that $\bigl|\log\bigl(\frac{\mathcal{R}(qK_1)}{\mathcal{R}(qK_2)}\bigr)\bigr|<\eps$ for all $K_0\leq K_1< K_2$.

By the equation for $\mathcal{R}$ in Theorem~\ref{thm: transition xqk to xq(k+1)} and since $k^{-5/4}$ is summable,
\begin{equation}\label{eq: xqK1 to xqK2}
\begin{aligned}
  \log\Bigl(\frac{\mathcal{R}(qK_2)}{\mathcal{R}(qK_1)}\Bigr) 
  &= \sum_{K_1\leq k < K_2} \log\Bigl(\frac{\mathcal{R}(q(k+1))}{\mathcal{R}(qk)}\Bigr)\\
  &=
   \frac{\lambda}{\sqrt{2Fq}}\Im\Biggl[w(E, \lambda, q, p)\sum_{K_1\leq k <K_2}\frac{e^{2i(\Omega(k)+\Lambda(qk))}}{\sqrt{k}}\Biggr]  \\
   &\quad  + \frac{\lambda^2}{4Fq} |w(E, \lambda, q, p)|^2\sum_{K_1\leq k <K_2}\frac{1}{k} \\
   &\quad 
   + \frac{\lambda^2}{4Fq}\Re\Biggl[w(E, \lambda, q, p)^2\sum_{K_1\leq k<K_2}\frac{e^{4i(\Omega(k)+\Lambda(qk))}}{k}\Biggr] +O(K_1^{-1/4}) \,.
\end{aligned}
\end{equation}

We begin by considering the first sum. By summation by parts
\begin{align*}
  \sum_{K_1\leq k <K_2}\frac{e^{2i(\Omega(k)+\Lambda(qk))}}{\sqrt{k}}
  &=
   \sum_{K_1\leq k< K_2-1}\biggl(\frac{e^{2i\Lambda(q(k+1))}}{\sqrt{k+1}}-\frac{e^{2i\Lambda(qk)}}{\sqrt{k}}\biggr)\sum_{k<j< K_2} e^{2i\Omega(j)}\\
   &\quad
   + \frac{e^{2i\Lambda(qK_1)}}{\sqrt{K_1}}\sum_{K_1\leq k< K_2}e^{2i\Omega(k)}\,.
\end{align*}
By assumption $E \notin \bigl\{\frac{\pi^2}{3p}m + \lambda : m \in \Z\bigr\}$ and thus $\Omega'(k)=  3p \frac{E-\lambda}{\pi}\notin \pi\Z$. Therefore, the inner sums can be explicitly computed and remain uniformly bounded
\begin{equation*}
  \sum_{k<j<K_2} e^{2i\Omega(j)} 
  = 
  \frac{e^{2 i \Omega(k+1)}}{1-e^{2i\pi q \frac{E-\lambda}{F}}}-\frac{e^{2 i \Omega(K_2)}}{1-e^{2i\pi q \frac{E-\lambda}{F}}} 
  =
  \frac{e^{2 i \Omega(k)}}{e^{-2i\pi q \frac{E-\lambda}{F}}-1}-\frac{e^{2 i \Omega(K_2)}}{1-e^{2i\pi q \frac{E-\lambda}{F}}}\,.
\end{equation*}
Consequently, we can estimate
\begin{align*}
  \sum_{K_1\leq k <K_2}&\frac{e^{2i(\Omega(k)+\Lambda(qk))}}{\sqrt{k}}\\
  &=
   \sum_{K_1\leq k< K_2-1}\biggl(\frac{e^{2i\Lambda(q(k+1))}}{\sqrt{k+1}}-\frac{e^{2i\Lambda(qk)}}{\sqrt{k}}\biggr)\Biggl[\frac{e^{2 i \Omega(k)}}{e^{-2i\pi q \frac{E-\lambda}{F}}-1}-\frac{e^{2 i \Omega(K_2)}}{1-e^{2i\pi q \frac{E-\lambda}{F}}}\Biggr]\\
   &\quad
   + \frac{e^{2i\Lambda(qK_1)}}{\sqrt{K_1}}\Biggl[\frac{e^{2 i \Omega(K_1)}}{1-e^{2i\pi q \frac{E-\lambda}{F}}}-\frac{e^{2 i \Omega(K_2)}}{1-e^{2i\pi q \frac{E-\lambda}{F}}}\Biggr]\\
   &=
   \frac{1}{e^{-2i\pi q \frac{E-\lambda}{F}}-1}\sum_{K_1\leq k< K_2-1}\biggl(\frac{e^{2i\Lambda(q(k+1))}}{\sqrt{k+1}}-\frac{e^{2i\Lambda(qk)}}{\sqrt{k}}\biggr)e^{2 i \Omega(k)} + O(K_1^{-1/2})\,,
\end{align*}
where in the last step we used the fact that the terms in the sum which get multiplied by $e^{2i\Omega(K_2)}$ telescope.

To understand the remaining sum we again appeal to Theorem~\ref{thm: transition xqk to xq(k+1)}, however, this time to the equation for $\Lambda$. Combined with
\begin{equation*}
  \frac{1}{\sqrt{k+1}}-\frac{1}{\sqrt{k}} = O(k^{-3/2})\,,
\end{equation*}
the Taylor expansion $e^x = 1+ x+ O(x^2)$, and~\eqref{eq: Lambda varies little q scale} we find
\begin{align*}
  &\sum_{K_1\leq k< K_2-1}\biggl(\frac{e^{2i\Lambda(q(k+1))}}{\sqrt{k+1}}-\frac{e^{2i\Lambda(qk)}}{\sqrt{k}}\biggr)e^{2 i \Omega(k)}\\
  &\qquad =
  i\lambda\sqrt{\frac{2}{Fq}}\sum_{K_1\leq k< K_2-1}\frac{e^{2i(\Omega(k)+\Lambda(qk))}}{k}\Re\Bigl[e^{2i(\Omega(k)+\Lambda(qk))}w(E, \lambda, q, p)\Bigr]+ O(K_1^{-1/4})\,.
\end{align*}

Therefore, writing $w=w(E, \lambda, q, p)$,
\begin{align*}
  &\frac{\lambda}{\sqrt{2Fq}}\Im\Biggl[w\sum_{K_1\leq k <K_2}\frac{e^{2i(\Omega(k)+\Lambda(qk))}}{\sqrt{k}}\Biggr]\\
  & =
  \frac{\lambda^2}{Fq}\Im\Biggl[\frac{i}{e^{-2i\pi q \frac{E-\lambda}{F}}-1}\sum_{K_1\leq k< K_2-1}\frac{e^{2i(\Omega(k)+\Lambda(qk))}w\,\Re\bigl[e^{2i(\Omega(k)+\Lambda(qk))}w\bigr]}{k}\Biggr] + O(K_1^{-1/4})\,.
\end{align*}
Using the facts that for any $z\in \C$
\begin{equation*}
  z \Re(z)= \frac{|z|^2}{2}+ \frac{z^2}{2}
\end{equation*}
and, for $\phi \in \R\setminus \Z$, 
\begin{equation*}
  \frac{i}{e^{2\pi i\phi}-1} = - \frac{\cot(\pi\phi)}{2}-\frac{i}{2}
\end{equation*}
we furthermore deduce that
\begin{align*}
  &\frac{\lambda}{\sqrt{2Fq}}\Im\Biggl[w\sum_{K_1\leq k <K_2}\frac{e^{2i(\Omega(k)+\Lambda(qk))}}{\sqrt{k}}\Biggr]\\
  & =
  \frac{\lambda^2}{2Fq}\Im\Biggl[\frac{i}{e^{-2i\pi q \frac{E-\lambda}{F}}-1}\sum_{K_1\leq k< K_2-1}\frac{|w|^2 + w^2 e^{4i(\Omega(k)+\Lambda(qk))} }{k}\Biggr] + O(K_1^{-1/4})\\
  & =
  -\frac{\lambda^2}{4Fq}|w|^2\sum_{K_1\leq k< K_2-1}\frac{1}{k} -
  \frac{\lambda^2}{4Fq}\Re\Biggl[w^2\sum_{K_1\leq k< K_2-1}\frac{ e^{4i(\Omega(k)+\Lambda(qk))} }{k}\Biggr]\\
  &\quad 
  +\frac{\lambda^2}{4Fq}\cot\Bigl(\pi q \frac{E-\lambda}{F}\Bigr)\Im\Biggl[w^2\sum_{K_1\leq k< K_2-1}\frac{ e^{4i(\Omega(k)+\Lambda(qk))} }{k}\Biggr] + O(K_1^{-1/4})\,.
\end{align*}

Inserting this back into~\eqref{eq: xqK1 to xqK2} several of the terms cancel, yielding
\begin{align*}
  \log&\Bigl(\frac{\mathcal{R}(qK_2)}{\mathcal{R}(qK_1)}\Bigr)\\
  &=
   \frac{\lambda^2 }{4Fq}\cot\Bigl(\pi q \frac{E-\lambda}{F}\Bigr)\Im\Biggl[w(E, \lambda, q, p)^2\sum_{K_1\leq k< K_2}\frac{ e^{4i(\Omega(k)+\Lambda(qk))} }{k}\Biggr] +O(K_1^{-1/4}) \,.
\end{align*}

Repeating the summation by parts argument used in analysing the first sum yields, due to the additional decay, that
\begin{equation*}
  \sum_{K_1\leq k< K_2}\frac{ e^{4i(\Omega(k)+\Lambda(qk))} }{k} = O(K_1^{-1/2})\,.
\end{equation*}

Therefore, we finally conclude that $\{\log(\mathcal{R}(qk))\}_{k\geq 1}$ forms a Cauchy sequence and hence converges. By~\eqref{eq: cR varies little on qk q(k+1)} and Theorem~\ref{thm: eq tilde eta, prufer coords vary little between Xl} this implies that the full sequence $\{\log(R(n))\}_{n\geq 1}$ converges which completes the proof of the first statement of Theorem~\ref{thm: convergence R deterministic}.

Let $E$ be one of the exceptional energies, i.e.\ $E =\frac{\pi^2}{3p}m +\lambda$ for some $m \in \Z$. Then
\begin{align*}
  w(E, \lambda, p, q) = \sum_{j=0}^{q-1}e^{-2\pi i \frac{pj^3-mj}{q}}\,.
\end{align*}
Therefore, if this sum is $0$, the equation~\eqref{eq: xqK1 to xqK2} implies that $\log(\mathcal{R}(qk))$ forms a Cauchy sequence. The same argument as above implies that $\log(R(n))$ has a limit. This completes the proof of Theorem~\ref{thm: convergence R deterministic}. 
\end{proof}

\subsection{Proof of Theorem~\ref{thm: full-line deterministic}}

With Theorem~\ref{thm: convergence R deterministic} in hand deducing Theorem~\ref{thm: full-line deterministic} is a direct consequence of Gilbert--Pearson subordinacy theory. Specifically we shall use the following result which is a special case of Proposition~7 in~\cite{ChristStolz_94}:
\begin{proposition}[{\cite[Proposition~7]{ChristStolz_94}}]\label{prop: GP only ac}
  Fix an open interval $I \subseteq \R$. Let $L_{F,\lambda}$ be as in Theorem~\ref{thm: full-line deterministic} and assume that for all $E \in I$ there is no subordinate solution of~\eqref{eq: Generalized eigeneq intro}. Then $I\subseteq \sigma(L_{F,\lambda})$ and $\sigma(L_{F,\lambda})\cap I$ is purely absolutely continuous.
\end{proposition}

\begin{proof}[Proof of Theorem~\ref{thm: full-line deterministic}]
  For $E \in \R \setminus\{\frac{\pi^2}{3p}m +\lambda: m\in \Z\}$ Theorem~\ref{thm: convergence R deterministic} implies that for any non-trivial real-valued solution $\psi$ of~\eqref{eq: generalized eigenfunction deterministic} the associated Pr\"ufer radius $R_\psi$ has a limit as $n\to \infty$ which is different from $0$. We denote the limit by $R_\psi(\infty)$. By Theorem~\ref{thm: l scale equations intro}
  \begin{equation*}
    \frac{1}{L}\int_0^L |\psi(x)|^2\,dx = \frac{\pi}{F}R_\psi(\infty)^2(1+o(1))\,,
  \end{equation*}
  and therefore no solution of~\eqref{eq: generalized eigenfunction deterministic} is subordinate at $\infty$. Proposition~\ref{prop: GP only ac} implies $\sigma_{ac}(L_{F,\lambda})=\R$ and that away from the exceptional energies $E= \frac{\pi^2}{3p}m+\lambda$ the spectrum is purely absolutely continuous. Since a discrete set of points cannot support singular continuous spectrum this completes the proof of Theorem~\ref{thm: full-line deterministic}.
\end{proof}

\appendix

\section{The one dimensional method of stationary phase with boundary contributions}

In this appendix we recall the asymptotic expansion of an oscillatory integral in one-dimension given by the method of stationary phase. The following is certainly well-known to experts but we include it since we have been unable to find the statement with the desired uniformity with respect to the involved functions taking into account the contributions from boundary points. 

We begin by first analysing the case when phase function has no stationary point in our interval. The proof we provide follows very closely that of~\cite[Theorem 7.7.1]{Hormander_book1}.

\begin{lemma}\label{lem: Principle of non-stationary phase}
  Let $I=(a, b)\subset \R$ be a bounded interval. Let $u \in C^k(I), \phi \in C^{k+1}(I)$ and assume that $\phi$ is real-valued and $\inf_{x\in I}|\phi'(y)|\geq \delta>0$. Then for $\omega>0$
  \begin{align*}
    \Biggl|\int_I u(x)e^{i\omega \phi(x)}\,dx - &\sum_{j=0}^{k-1} \Bigl(\frac{i}{\omega}\Bigr)^{j+1}\biggl[\frac{\mathcal{B}^j u(a)}{\phi'(a)}e^{i\omega\phi(a)}-\frac{\mathcal{B}^j u(b)}{\phi'(b)}e^{i\omega\phi(b)}\biggr]\Biggr|\\
    & \lesssim |I|\omega^{-k}\sum_{m=0}^k \delta^{m-2k}\|\partial_y^mu\|_{L^\infty(I)}
  \end{align*}
  where 
  \begin{equation*}
    \mathcal{B}^0v(y) = v(y)\,, \quad \mbox{and}\quad
    \mathcal{B}^jv(y) = \partial_y((\phi')^{-1}\mathcal{B}^{j-1}v)(y)\,.
  \end{equation*}
  Moreover, the implicit constant is uniform for $\phi$ in bounded subsets of $C^{k+1}(I)$.
\end{lemma}

By combining the more classical statement of the stationary phase expansion, with compactly supported smooth amplitude, and the previous lemma one readily deduces the following result.

\begin{lemma}\label{lem: Principle of stationary phase}
  Let $I=(a, b)\subset \R$ be a bounded interval and $k\in \N$. Assume that $u\in C^{2k}(I), \phi \in C^{3k+1}(I)$ with $\phi$ real-valued and such that there exists a unique $x_0 \in I$ with $\dist(x, \partial I)>\kappa |I|$ such that $\phi'(x_0)=0$, $\phi''(x_0)\neq 0$, and for all $x \in I$ 
  \begin{equation}\label{eq: non-degeneracy assumption}
    \frac{|x-x_0|}{|\phi'(x)|} \leq K\,.
  \end{equation}
  Then, for $\omega >0$,
  \begin{align*}
    \Biggl|\int_I u(x)e^{i\omega \phi(x)}\,dx& - \frac{(2\pi)^{1/2}e^{i \omega \phi(x_0)+i\pi/4}}{\omega^{1/2}\phi''(x_0)^{1/2}} \sum_{j=0}^{k-1}\omega^{-j}\mathcal{L}_j u(x_0)\\
    &
    -\sum_{j=0}^{k-2} \Bigl(\frac{i}{\omega}\Bigr)^{j+1}\biggl[\frac{\mathcal{B}^j u(a)}{\phi'(a)}e^{i\omega\phi(a)}-\frac{\mathcal{B}^j u(b)}{\phi'(b)}e^{i\omega\phi(b)}\biggr]\Biggr| \lesssim \omega^{-k}\sum_{j\leq 2k} \|\partial_y^ju\|_{L^\infty(I)}\,.
  \end{align*}
  where
  \begin{align*}
  \mathcal{L}_j v(y) &= \sum_{\nu-\mu=j}\sum_{2\nu\geq 3\mu} \frac{1}{i^j2^\nu \nu!\mu!\phi''(x_0)^\nu} (-i\partial_y)^{2\nu}(f_l^\mu v)(y)\,,\\
  \mathcal{B}^0v(y) &= v(y)\,, \quad \mathcal{B}^jv(y) = \partial_y((\phi')^{-1}\mathcal{B}^{j-1}v)(y)\,,
  \end{align*}
  and
  \begin{equation*}
   f_l(y)=\phi(y)-\phi''(x_0)(y-x_0)^2/2\,.
  \end{equation*}
  Moreover, the implicit constant is uniform for $\phi$ in bounded subsets of $C^{3k+1}(I)$ and $\kappa, K, |I|$ in compact subsets of $(0, 1/2)$, $(0, \infty)$, and $(0, \infty)$, respectively.
\end{lemma}

\begin{proof}[Proof of Lemma~\ref{lem: Principle of non-stationary phase}]
Set
\begin{equation*}
  \mathcal{I}_{\phi,\omega}[u] = \int_I u(x) e^{i\omega \phi(x)}\,dx
\end{equation*}
then for $\omega >0$ an integration by parts yields
\begin{align*}
  \mathcal{I}_{\phi, \omega}[u] 
  &= 
  -\frac{i}{\omega}\int_I \frac{u(x)}{\phi'(x)}(e^{i\omega \phi(x)})'\,dx\\
  &=
  \frac{i}{\omega}\biggl[\frac{u(a)}{\phi'(a)}e^{i\omega\phi(a)}-\frac{u(b)}{\phi'(b)}e^{i\omega\phi(b)}\biggr] + \frac{i}{\omega}\mathcal{I}_{\phi, \omega}[\mathcal{B}^1 u]\,.
\end{align*}
Iterating, we arrive at
\begin{equation*}
  \mathcal{I}_{\phi,\omega}[u] = \sum_{j=0}^{k-1} \Bigl(\frac{i}{\omega}\Bigr)^{j+1}\biggl[\frac{\mathcal{B}^j u(a)}{\phi'(a)}e^{i\omega\phi(a)}-\frac{\mathcal{B}^j u(b)}{\phi'(b)}e^{i\omega\phi(b)}\biggr] + \Bigl(\frac{i}{\omega}\Bigr)^{k}\mathcal{I}_{\phi, \omega}[\mathcal{B}^k u]\,.
\end{equation*}

Since trivially
\begin{equation*}
  |\mathcal{I}_{\phi, \omega}[v]| \leq |I|\|v\|_{L^\infty(I)}
\end{equation*}
it suffices to prove that, with a constant $C_\phi$ which can be taken uniform for $\phi$ in bounded subsets of $C^{k+1}(I)$,
\begin{equation*}
  \|\mathcal{B}^k u \|_{L^\infty(I)} \leq C_\phi \sum_{m=0}^k \delta^{m-2k}\|\partial_y^mu\|_{L^\infty(I)}\,.
\end{equation*}
We shall argue by induction to prove that for $r=0, \ldots , k$ and $j = 0, \ldots, k-r$ we have
\begin{equation*}
  \|\partial_y^j\mathcal{B}^ru\|_{L^\infty} \leq C_\phi \sum_{m=0}^{r+j} \delta^{-j+m-2r}\|\partial_y^mu\|_{L^\infty(I)},
\end{equation*}
which when $r=k, j=0$ is the desired bound.

For $r=0$ we by definition have for $j=0, \ldots, k$
\begin{equation*}
  \|\partial^j_y\mathcal{B}^0 u\|_{L^\infty(I)} =  \|\partial_y^ju\|_{L^\infty(I)}\,,
\end{equation*}
since $\delta$ is bounded this implies the claim when $r=0$ and all $j\leq k$.

For $r>0$, $0\leq j\leq k-r$ assume that the statement is known for all smaller $r$. By repeated use of the product rule
\begin{align*}
  \partial^j_y\mathcal{B}^r v = \partial_y^{j+1}((\phi')^{-1}\mathcal{B}^{r-1} v) 
  &= 
  \sum_{m=0}^{j+1} \binom{j+1}{m}(\partial_y^m\mathcal{B}^{r-1} v)\partial_y^{j+1-m}(\phi')^{-1}\\
  &\leq 
  C_\phi \sum_{m=0}^{j+1}\delta^{-j+m-2}|\partial_y^m\mathcal{B}^{r-1} v|
\end{align*}
where $C_\phi$ is uniformly bounded for $\phi$ in bounded subsets of $C^{k+1}(I)$ (or rather $C^{j+2}(I)$ but $j+2\leq k+1$ since $r>0$). Therefore, by the induction hypothesis
\begin{align*}
  \|\partial_y^j\mathcal{B}^r v\|_{L^\infty(I)} 
  &\leq 
  C_\phi \sum_{m=0}^{j+1}\delta^{-j+m-2} \sum_{m'=0}^{r-1+m} \delta^{-m+m'-2(r-1)}\|\partial_y^{m'}u\|_{L^\infty(I)}\\
  &=
  C_\phi \sum_{m=0}^{j+1}\sum_{m'=0}^{r-1+m}\delta^{-j+m'-2r}\|\partial_y^{m'}u\|_{L^\infty(I)}\\
  &\leq
  C_\phi \sum_{m'=0}^{r+j}\delta^{-j+m'-2r}\|\partial_y^{m'}u\|_{L^\infty(I)}
\end{align*}
which completes the proof of the lemma.
\end{proof}

\begin{proof}[Proof of Lemma~\ref{lem: Principle of stationary phase}]

By translating we may without loss of generality assume that $I=(-\frac{|I|}{2}, \frac{|I|}{2})$. Fix $\chi \in C_0^\infty(I)$ with $0\leq \chi \leq 1$ and such that $\chi \equiv 1$ in $(-\frac{|I|(1-\kappa)}{2}, \frac{|I|(1-\kappa)}{2})$ which ensures that $\chi =1$ in a neighbourhood of the stationary point of $\phi$.  

Write
\begin{align*}
  \int_I u(x)e^{i\omega \phi(x)}\,dx &= \int_\R \chi(x)u(x)e^{i\omega\phi(x)}\,dx\\
  &\quad +\int_I (1-\chi(x))u(x)e^{i\omega\phi(x)}\,dx\,.
\end{align*}
By the non-degeneracy assumption of the phase~\eqref{eq: non-degeneracy assumption} $|\phi'|$ is bounded from below uniformly in the support of $1-\chi$. As such we can apply Lemma~\ref{lem: Principle of non-stationary phase} to the latter integral by splitting it into two parts; one close to $\frac{|I|}{2}$ and one close to $-\frac{|I|}{2}$. The proof of the lemma is completed by applying the classical stationary phase method (with compactly supported amplitude) as stated in~\cite[Theorem~7.7.5]{Hormander_book1} to the first integral.
\end{proof}


\end{document}